\documentclass{article}
\usepackage{booktabs} 
\usepackage[ruled]{algorithm2e} 

\SetAlFnt{\small}
\SetAlCapFnt{\small}
\SetAlCapNameFnt{\small}
\SetAlCapHSkip{0pt}
\IncMargin{-\parindent}

\usepackage{float}
\usepackage{graphicx} 
\usepackage{import}
\usepackage{caption}
\usepackage{subcaption}
\usepackage{wrapfig}
\usepackage{enumitem}
\usepackage{bm}
\usepackage{tikz}
\usepackage{array}
\usepackage{hyperref}
\usepackage{geometry, amsmath, amsthm, amsfonts}
\usepackage{thmtools, thm-restate}

\newcolumntype{M}[1]{>{\centering\arraybackslash}m{#1}}

\DeclareMathOperator{\sgn}{sgn}
\declaretheorem[name=Theorem]{thm}
\declaretheorem[name=Definition, sibling=thm]{defn}
\declaretheorem[name=Lemma, sibling=thm]{lemma}
\declaretheorem[name=Proposition, sibling=thm]{prop}

\declaretheorem[name=Fact, sibling=thm]{fact}
\newenvironment{autoproof}[1]
  {\begin{proof}[Proof of \autoref{#1}]}
  {\end{proof}}

\usepackage{xcolor}

\title{Convergence rates of random-order best-response dynamics in public good games on networks\footnote{This work was supported by the Polish National Science Centre through grant 2021/42/E/HS4/00196.}}

\author{Wojciech Misiak\footnote{Institute of Applied Mathematics and Mechanics, University of Warsaw. \texttt{w.misiak@mimuw.edu.pl}}
\and Marcin Dziubi\'{n}ski\footnote{Institute of Informatics, University of Warsaw. \texttt{m.dziubinski@mimuw.edu.pl}}
}
\begin{document}

\begin{titlepage}

\maketitle
\begin{abstract}
    We study convergence rates of random-order best-response dynamics in games on networks with linear best responses and strategic substitutes.
    Combining formal analysis with numerical simulations we identify phenomena that lead to slow convergence. 
    One of the key such phenomena is convergence to stable strategy profiles in parts of the network neighboring sets of nodes which remain inactive until the dynamics is close to converging and then switch to activity, initiating convergence to profiles with a new set of active agents and possibly leading to another iteration of such behavior. 
    We identify structural properties of graphs which make such phenomena more likely. 
    These properties go beyond the spectrum of a graph, which we demonstrate analyzing convergence rates on co-spectral mates.
    
\end{abstract}
\vspace{1cm}
\setcounter{tocdepth}{2} 
\tableofcontents

\end{titlepage}

\section{Introduction}
The models of public good games on networks concern scenarios where strategic agents, connected in a network, decide on providing a good that can be consumed by themselves as well as all their neighbors in the network. Examples of scenarios that feature such characteristics include choosing a level of heating own apartment in a block of flats, where each flat benefits from the neighboring flats being heated, caring for a house garden which neighbors benefit from as well, testing new products and providing reviews in a social network, which benefits observers of the content in the network, investing in R\&D, the results of which spill over to competing companies in the industry.

These types of interactions often feature local effects as well as own efforts and the efforts of the neighbors being strategic substitutes. A prominent class of models of such interactions are network games with linear best responses and strategic substitutes~\cite{Ball06,BraKra07,Bra14}.
As was shown by~\cite{Bra14}, these games are strategically equivalent to a potential game and, therefore, have a Nash equilibrium in pure strategies, that can be reached by a best-response dynamics. 

Although existence, conditions for uniqueness of Nash equilibria in these games and how their properties depend on the topology of the network are quite well understood, very little is known about the time it takes for the interactions of the agents to converge to equilibria, how this time depends on the properties of the network, and how these dynamics proceed. The objective of this paper is to fill this gap and extend our understanding of convergence rates in games with strategic substitutes on networks. As was already argued by~\cite{Ell93}, in the context of coordination games on networks, understanding the time of convergence to equilibrium is important as it allows us to determine whether the dynamic interactions will reach an equilibrium point within a reasonable time horizon.

\subsection{Related literature}
Game theoretic models of strategic interactions in networks attracted interest from the researchers in computer science and economics for the past 25 years. Early such models include graphical games~\cite{Kea01}, coordination~\cite{Morr00,JackWa02a,JackWa02b} and anti-coordination games in networks~\cite{Bra07}, as well as public-good games in networks~\cite{Ball06,BraKra07,Bra14,All15}. The literature distinguished two general classes of games networks: the games of strategic complements~\cite{Calvo04,Calvo07,Ball06} and games of strategic substitutes~\cite{Ball06,BraKra07,Bra14}. The prime example of the former are coordination games on networks and the prime example of the latter are public-good games on networks.

In general, Nash equilibria in games of strategic substitutes are not well behaved and neither existence nor uniqueness are guaranteed. Nevertheless, with additional structure, Nash equilibrium exists and their properties, in particular, uniqueness and multiplicity, can be tied to the properties of the network. A fundamental class of models for which a good understanding of equilibria has been obtained are games on networks with strategic substitutes and linear best responses. The most general model of such games was studied by~\cite{Bra14}. They showed that these games are strategically equivalent to a potential game and have a Nash equilibrium in pure strategies. In addition, they showed that Nash equilibrium in these games is unique, if the network externality factor is lower than the inverse of the absolute value of the lowest eigenvalue of the network. If the externality factor is higher than this threshold, every equilibrium features inactive agents and equilibria may be multiple. 
In addition, \cite{Bra14} obtained conditions for stability of equilibria in the games in question. They considered asymptotic stability in the neighborhood of equilibrium points and showed that an equilibrium is stable if and only if the network externalities are lower than the inverse of the absolute value of the lowest eigenvalue of the network induced by the set of active agents in the equilibrium. 
The literature on games on networks is vast and we restricted attention to games of strategic substitutes, which are closest related to our paper. A broad overview of the literature can be found in an excellent review of~\cite{JackZen15}.

The study of strategic interaction dynamics in networks focused largely on games with strategic complements, particularly on coordination games. The key questions here concern the conditions under which the agents converge to a particular equilibrium, stability of equilibrium points, and speed of convergence. \cite{Morr00} studied a coordination game on a network with binary actions, $0$ and $1$. He characterized the properties of equilibria where some agents coordinate on action $1$  and the remaining agents coordinate on action $0$. He also characterized the conditions under which the best-response dynamics converges to the  equilibrium where all players choose action $1$. \cite{Young93,Mail93,Ell93,Ell00,JackWa02b} study stochastic stability of Nash equilibria in coordination games on networks. These papers consider a random-order best-response dynamics where agents can make random mistakes when choosing their responses. Depending on the particular model of errors, the structure of the network may affect the equilibria reached by the dynamics, as demonstrated by~\cite{Ell93} for cycle network and~\cite{JackWa02b} for star network. In addition, the structure of the network may have a significant effect on the speed of convergence to equilibrium~\cite{Ell93,Ell00}.
In a recent paper, \cite{Huss24} considered convergence of quantal response dynamics to mixed strategy quantal response equilibria and $\varepsilon$-Nash equilibria in arbitrary games on networks with local effects. They provide sufficient conditions on convergence and conduct numerical experiments on the behavior of quantal response dynamics.

Our paper is also related to a large body of literature on numerical analysis of strategic interactions in networks, which includes agent-based computational models~\cite{Tesf06,Wilh06a} and numerical simulations in physics~\cite{Perc13}. In particular, \cite{Wilh06b} study a model where protection against an external threat can be a public good. He conducts numerical experiments in order to study emergence and evolution of different forms of organization (dictatorship, communes, democracy) in this setting. \cite{Zhou16} study a public good game on a network with externality factor equal to $1$ and with the agents facing a budget constraint. They run numerical simulations with agents connected in a complete network and study performance of four behavioral strategies (average, proportional, greedy, and random). 

\subsection{Our contribution}

We study random-order best-response dynamics in network games of strategic substitutes with linear best responses. Our main focus are the total convergence time of the dynamics and how it is influenced by the network topology. 

First, we provide conditions for stochastic stability of equilibria and tie the stability to absolute value of the lowest eigenvalue of the network induced by the set of active agents,  $|\lambda_{\min}(\bm{G}_S)|$. This result provides a bound on the convergence rate in close neighborhood of an equilibrium. Convergence time is exponential but it can be arbitrarily slow when the externality factor converges to $|\lambda_{\min}(\bm{G}_S)|$ from below. Our stability result is analogous to the stability result of~\cite{Bra14} but it is obtained under a different concept of stability that corresponds to a different process. 
Under discrete dynamics, the strategy profile does not change continuously and ensuring the process remains near the equilibrium required more rigor.
Furthermore, we construct a martingale argument to account for the stochasticity of the system and obtain the result.

Next, using numerical simulations on co-spectral mates, we demonstrate that the convergence times are affected by properties of networks that go beyond their spectral characteristic.

To get a better understanding of the behavior of the best-response dynamics in question, we obtain formal characterization of stable equilibria on path graphs and we conduct formal analysis of convergence rates on such graphs with small number of nodes. We also extend this analysis to larger graphs combining formal and numerical based methods.
We show that the main cause of slow convergence is the existence of particular ``problematic'' subgraphs within the graph.  
These include graphs on which which the best-response dynamics converges slowly, even if all agents are active in equilibrium, as well as subgraphs in which some vertices change status from being inactive late, leading to what we call a \emph{reshuffle}.
The subgraphs which are ``problematic'' change as the externality factor changes. If the subgraphs which are ``problematic'' for a given externality factor do not exist within the structure of the graph, the convergence is fast.

We extend the insights obtained for path graphs to more general graph topologies. We first obtain formal results characterizing of sets of active agents in stable equilibria on general graphs. In particular, we show that when the externality factor is sufficiently large (greater than inverse of the golden ratio), the set of active agents in any stable equilibrium consist of disjoint cliques. Next we conduct extensive numerical analysis of best-response dynamic in question on various classes of deterministic as well as random graphs. We find that in some regimes, many ``problematic'' subgraphs may exist, leading to slow convergence. To analyze best-response dynamics on different types of networks, we created a web application available at https://networkconvergencerates.up.railway.app/.

The rest of the paper is organized as follows. In Section~\ref{sec:prelim} we define the model of the games in question, the best-response dynamics and the concept of stability we consider. We also prove the result providing sufficient and necessary conditions on  stability of equilibria. The analysis is conducted in Section~\ref{sec:analysis}. In particular, in Section~\ref{sec:path} we study the equilibria and the convergence times for path graphs and in Section~\ref{sec:gengraphs} we consider more general classes of graphs. We summarize and discuss our results in Section~\ref{sec:summary}. The proofs are given in the Appendix.

\section{Preliminaries}
\label{sec:prelim}
\subsection{The model}
We study a game on networks as considered by~\cite{Bra14}.
There is a set $N = \{1,\ldots,n\}$ of agents connected in an undirected graph represented by a symmetric adjacency matrix $\bm{G} = (g_{ij})_{i,j\in N} \in \{0,1\}^{N\times N}$ with $g_{ij} = g_{ji}$, for all $(i,j) \in N\times N$. Agent $i \in N$ is connected to agent $j\in N$ (or is a \emph{neighbor} or $j$ in $\bm{G}$) if and only if $g_{ij} = 1$. No agent is connected to herself, $g_{ii} = 0$, for all $i \in N$. The agents are involved in an interaction where each agent $i\in N$ chooses an activity level $x_i \in [0,1]$ and payoff of $i$ depends on her activity level as well as on activity levels of her neighbors in graph $\bm{G}$, scaled by an \emph{externality factor} $\delta \in [0,1]$. The externality factor determines the extent to which activity of neighbor of an agent affect her own payoff. A profile of activity levels $\bm{x} = (x_j)_{j \in N}$ results in payoff $U_i(\bm{x} \mid \bm{G}, \delta)$. An activity level $x_i \in [0,1]$ of agent $i\in N$ is a \emph{best response} to a profile of activity levels of the remaining agents, $\bm{x}_{-i} = (x_j)_{j \in N\setminus \{i\}}$, if it maximises the value of $U_i(x_i,\bm{x}_{-i} \mid \bm{G},\delta)$. We consider games with \emph{linear best responses}, where the best response of agent $i \in N$ to activity levels of other agents, $\bm{x}_{-i}$, is given by
\begin{equation*}
x_i = \max\!\left(0, 1 - \delta\sum g_{ij}x_j\right).
\end{equation*}
An example of payoff functions that result in such best responses is as follows. Given a strategy profile $x = (x_j)_{j \in N}$, payoff to agent $i\in N$ is
\begin{equation*}
U_i(\bm{x} \mid \bm{G}, \delta) = x_i - \frac{x_i^2}2 - \delta \sum g_{ij}x_i x_j.
\end{equation*}

This simple model of network games with linear best responses encompasses a large class of network games with local effects (c.f. \cite{Bra14}). The assumption that $\delta \in [0,1]$ results in agents' activities being strategic substitutes with the activities of their neighbors.

Given an $\epsilon \geq 0$, a strategy profile $\bm{x}$ is a \emph{$\epsilon$-Nash equilibrium} of the game in question if, for any agent $i \in N$ and any activity level $x'_i \in [0,1]$,
\begin{equation*}
U_i(x'_i,\bm{x}_{-i} \mid \bm{G}, \delta) - U_i(x_i,\bm{x}_{-i} \mid \bm{G}, \delta) \leq \epsilon.
\end{equation*}
A strategy profile $\bm{x}$ is a Nash equilibrium if it is a $\epsilon$-Nash equilibrium with $\epsilon = 0$. Throughout the remaining part of the paper we will use the term equilibrium rather than Nash equilibrium, for short.

As was shown in~\cite{Bra14}, the game in question has a Nash equilibrium. In every equilibrium there is a non-empty set of agents (called \emph{active agents}) choosing positive activity levels, $x_i > 0$, and a (possibly empty) set of agents (called \emph{inactive agents}) choosing activity level $0$. An agent is called \emph{strictly inactive} if $\delta\sum_{j}g_{ij}x_j>1$.\footnote{Strictly inactive agents are relevant, as they stay inactive under a small perturbation of the strategy profile. For almost any $\delta$, every inactive agent is strictly inactive. "Almost any $\delta$" refers to all $\delta\in[0,1]$ except for finitely many points. This is shown in \cite[Footnote~16]{Bra14}} The set of active agents is called an \emph{active set}.
Uniqueness of equilibrium is connected to the spectral properties of the adjacency matrix $A$, more specifically, its lowest eigenvalue, $\lambda_{\min}(\bm{G})$. We summarise the key results from~\cite{Bra14} in the theorem below.

\begin{thm}[\cite{Bra14}]
\label{th:bka:ne}
Equilibrium exists for any $\bm{G} \in \{0,1\}^{N\times N}$ and $\delta \in [0,1]$. If $\delta < 1/|\lambda_{\min}(\bm{G})|$ then equilibrium is unique.
If $\delta > 1/|\lambda_{\min}(\bm{G})|$ then there exists an equilibrium with at least one inactive agent.
\end{thm}

\subsection{Best-response dynamics and stability}
We consider a random best-response dynamics with discrete time, which is defined as follows. Fix a best-response function, 
$f : [0,1]^N \rightarrow [0,1]^N$ which, for every strategy profile $\bm{x}\in [0,1]^N$ and every agent $i \in N$, selects a best response $f_i(\bm{x})$ of $i$  to $x$. The play starts at time $t = 1$ with a strategy profile $\bm{x}(0)$.
In each time-step  $t = 1,2,\ldots$ one of the agents $i \in N$ is picked uniformly at random and chooses best response $f_i(\bm x_{-i}(t-1))$ to the strategy profile $\bm{x}_{-i}(t-1)$. We then set $\bm{x}(t)=(f_i(\bm{x}_{-i}(t-1)), \bm{x}_{-i}(t-1))$.

We are interested in the stochastic process resulting from this dynamics. In particular, we are interested in stable strategy profiles under this process. Throughout the paper we adopt the following concept of stability.
\begin{defn}
An equilibrium $\bm{x}^* = f(\bm{x}^{*})$ is \emph{locally asymptotically stable in probability} whenever the following two conditions hold:
\begin{enumerate}
    \item it is \emph{Lyapunov stable almost surely}:
    $$\forall\epsilon>0\exists r>0\forall \bm x(0): \left\|\bm{x}(0)-\bm{x}^*\right\|<r\Rightarrow \mathbb P(\sup_{ t>0}\left\| \bm{x}(t)-\bm{x}^* \right\| < \epsilon)=1,$$
    \item it is \emph{attractive}:
    $$\exists \epsilon>0\ \forall \bm{x}(0):\ \left\|\bm{x}(0)-\bm{x}^*\right\| < \epsilon \Rightarrow \mathbb P\left(\lim_{t\to\infty}\left\|\bm{x}-\bm{x}(t)\right\|_\infty= 0\right)=1.$$
\end{enumerate}
\end{defn}

Informally, an equilibrium being Lyapunov stable implies that once the process is near that profile, it stays near it. Being attractive means that if the process is near enough, it converges to the strategy profile almost surely.
We will refer to such equilibria as \emph{stable}. If either of these conditions is not met, we will refer to the equilibrium as \emph{unstable}. 

The following proposition provides a necessary and sufficient condition for stability of equilibria.
\begin{thm}
\label{th:stability}
Let $\bm{x}^*$ be an equilibrium with active set $S\subseteq N$ and strictly inactive agents $N\setminus S$. Let $\bm{G}_{S}$ be the adjacency matrix restricted to $S$.
\begin{itemize}
\item If $\delta\geq 1/|\lambda_{\min}(\bm{G}_{S})|$, $\bm{x}^*$ is unstable.
\item If $\delta<1/|\lambda_{\min}(\bm{G}_{S})|$, $\bm{x}^*$ is stable. Furthermore, we have $$\exists \epsilon>0\ \forall \bm{x}(0):\ \left\|\bm{x}(0)-\bm{x}^*\right\| < \epsilon \Rightarrow\mathbb P\left(\lim_{t\to\infty}\left(1-\frac{1+\delta\lambda_{\min}(\bm G_{S})}{n}\right)^{-t}\left\|\bm{x}^*-\bm{x}(t)\right\|_\infty \text{ converges}\right)=1.$$
\end{itemize}
\end{thm}
Besides showing stability, Theorem~\ref{th:stability} shows that when an equilibrium is stable, convergence is exponential. 
The additional stability criterion can be interpreted as stating that the distance from equilibrium behaves asymptotically as $\left(1-\frac{1+\delta\lambda_{\min}(\bm G_{S})}{n}\right)^t\approx \exp\left(-(1+\delta\lambda_{\min}(\bm G_{S}))\frac tn\right)$ for large $n$.
With $\delta$ converging to $1/|\lambda_{\min}(\bm{G}_{S})|$ from below, the convergence can become arbitrarily slow.
\subsection{The criterion for convergence}
In general, the convergence to equilibrium only happens in the limit $t\to\infty$. Since large part of our results is based on numerical simulations, we need a criterion for stopping the simulations and deciding whether the given strategy profile can be deemed an equilibrium. The criterion we adopt in this paper is as follows. Given a strategy profile $\bm{x}$, let
\begin{equation*}
d(\bm{x})=\|\bm{x}-\max(0,1-\delta \bm{G} \bm{x})\|_\infty,
\end{equation*}
where the $\max$ function in the formula is applied to each coordinate separately, be the distance between $\bm{x}$ and the best response to $\bm{x}$. 

We will say that the process \emph{converged} at time $t$ if $d(\bm{x}(t)) < \epsilon$ for some given value of $\epsilon$. When this happens for the first time, we stop the dynamics. The time at which we stop the dynamics is called the \emph{time to convergence}.\footnote{
Notice that if $U$ is $C$-Lipschitz with constant $C$ then a strategy profile $\bm{x}$ that satisfies $d(\bm{x}) < Q$ is an $\epsilon$-Nash equilibrium with $\epsilon$ that depends on $C$ and $Q$.} To remain scale free, it will be expressed as the number of \emph{rounds} $T$, with each round consisting of $n$ time-steps. The number $T$ is therefore the average number of times each agent updated their strategy.
%
%

\section{Analysis}
\label{sec:analysis}

Our main objective is to obtain an understanding which characteristics of graph structure affect the speed of convergence of the best-response dynamic in question and how these characteristics interact with the externality factor, $\delta$.
The results of~\cite{Bra14} as well as Theorem~\ref{th:stability} suggest that an important factor would be the eigenvalues of the adjacency matrix, i.e. the spectrum of the graph. To check whether this is the case we consider cospectral mates: graphs with the same spectrum but different topology. 

\begin{defn}
Two graphs with adjacency matrices $\bm{G}$ and $\bm{H}$ are \emph{cospectral mates} if $\bm{G}$ and $\bm{H}$ have the same sets of eigenvalues.
\end{defn}

Figure~\ref{fig:cospectral} presents convergence rates of two cospectral mates: the start with $5$ vertices and the cycle with $4$ vertices with a single vertex added (c.f. Figure~\ref{fig:cospectral_graphs}).

The horizontal axis in the plots represents the values of $\delta$, while the vertical axis represents the number of rounds to convergence $T$. Each dot represents a sample trajectory at the time of convergence. For blue dots, the terminal active set supports a stable equilibrium. For red dots, it does not, meaning that the process stopped before reaching it. For each value of $\delta$, $10$ sample trajectories are presented. The distance between consecutive values of $\delta$ is $0.005$.
The red and green vertical lines represent the values $1/|\lambda_i(\bm G)|$, corresponding to negative and positive eigenvalues respectively.
The leftmost red line is therefore the $\delta-$threshold for the possible loss of uniqueness of the equilibrium.

\begin{figure}[!htb]
    \centering
    \begin{minipage}{.49\textwidth}
    \begin{subfigure}[t]{0.49\textwidth}
        \centering
        \centering
        \begin{tikzpicture}[main_node/.style={circle,fill=blue!20,draw,minimum size=1em,inner sep=3pt]}]
            \node[main_node] (1) at (0,0) {1};
            \node[main_node] (2) at (1.5, -1.5)  {2};
            \node[main_node] (3) at (1.5, 1.5) {3};
            \node[main_node] (4) at (-1.5, 1.5) {3};
            \node[main_node] (5) at (-1.5, -1.5) {3};
            \draw (2)--(1)--(3);
            \draw (4)--(1)--(5);
        \end{tikzpicture}
        \caption{The star $K_{1,4}$}
    \end{subfigure}%
    \hfill
    \begin{subfigure}[t]{0.49\textwidth}
        \centering
        \begin{tikzpicture}[main_node/.style={circle,fill=blue!20,draw,minimum size=1em,inner sep=3pt]}]
            \node[main_node] (1) at (0,0) {1};
            \node[main_node] (2) at (1.5, -1.5)  {2};
            \node[main_node] (3) at (1.5, 1.5) {3};
            \node[main_node] (4) at (-1.5, 1.5) {3};
            \node[main_node] (5) at (-1.5, -1.5) {3};
            \draw (2)--(3)--(4)--(5)--(2);
        \end{tikzpicture}
        \caption{The cycle with an added vertex $C_4\cup K_1$}
    \end{subfigure}%
    \caption{An example of cospectral graphs}
    \label{fig:cospectral_graphs}
    \end{minipage}%
    \hfill
    \begin{minipage}{0.49\textwidth}
        \centering
        \begin{subfigure}[t]{0.49\textwidth}
        \centering
        \includegraphics[width=\linewidth]{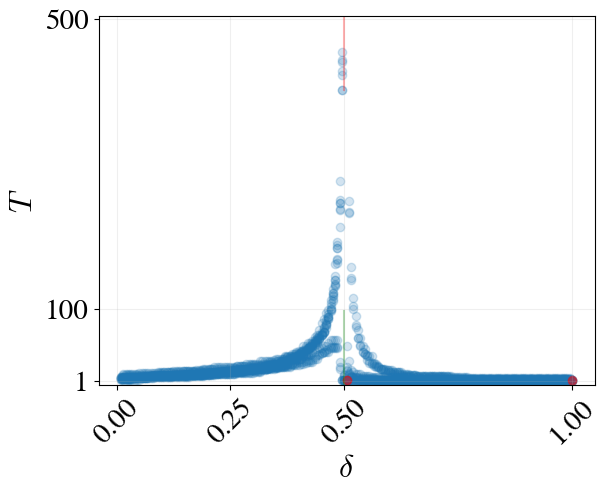}%
        \caption{Convergence times $T$ for \hbox{$C_4\cup{K_1}$}.}
    \end{subfigure}%
    \hfill
        \begin{subfigure}[t]{0.49\textwidth}
        \centering
        \includegraphics[width=\linewidth]{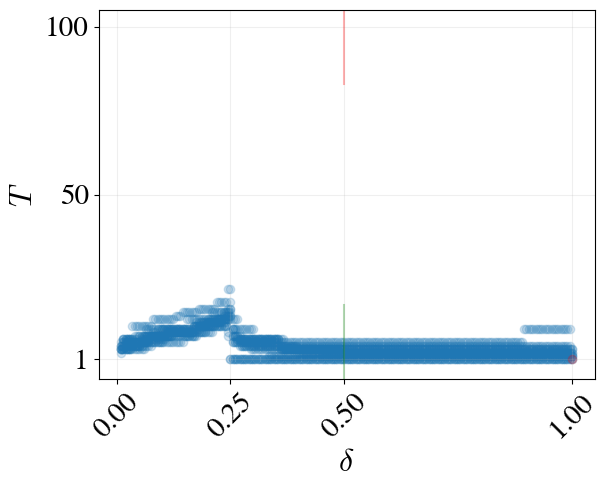}%
        \caption{Convergence times $T$ for $K_{1,4}$.}
    \end{subfigure}
    \caption{Convergence times for cospectral mates.}
    \label{fig:cospectral}
    \end{minipage}
\end{figure}

The star graph $K_{1,4}$ and the cycle with an added vertex $C_4\cup K_1$ have the same spectrum (c.f. \cite[p.~7]{brouwer2012spectra}).
For both of them, $\lambda_{\min}(\bm G)=-2$, implying that uniqueness of equilibria may be lost at $\delta=1/2$.
For the cycle with an extra edge, we observe that dynamics are significantly slower near this value and otherwise fast.
For the star graph, this behavior does not occur, with only a small slowdown visible for a different value of $\delta$.

This showcases that the convergence times can be different for cospectral mates.
The differences are caused by different \emph{subgraph structures}, which lead to different equilibria and are not determined by spectrum.

To get a better understanding of the properties of graph structure and the properties of best-response dynamics that lead to slower convergence rates we first address these questions by studying very simple graphs analytically. Then we extend this analysis to arbitrary graphs

\subsection{Path graphs}
\label{sec:path}

In this section we consider path graphs, that is connected graphs over $n \geq 2$ vertices where $g_{ij} = 1$, if $|j-i| = 1$,and $g_{ij} = 0$, otherwise (c.f. Figure~\ref{fig:path}).

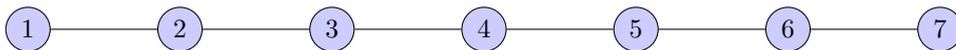
\begin{figure}[h!]
    \centering
        \begin{tikzpicture}[main_node/.style={circle,fill=blue!20,draw,minimum size=1em,inner sep=3pt]}]
            \node[main_node] (1) at (-6,0) {1};
            \node[main_node] (2) at (-4,0)  {2};
            \node[main_node] (3) at (-2,0) {3};
            \node[main_node] (4) at (0,0) {4};
            \node[main_node] (5) at (2,0) {5};
            \node[main_node] (6) at (4,0) {6};
            \node[main_node] (7) at (6,0) {7};
            \draw (1)--(2)--(3)--(4)--(5)--(6)--(7);
        \end{tikzpicture}
    \caption{The path graph over $n=7$ vertices.}
    \label{fig:path}
\end{figure}

It turns out that many of the observed qualitative phenomena regarding the best-response dynamics that we observe for arbitrary graphs occur already for path graphs. In addition path graphs have equilibria which are relatively simple to describe. For the remaining part of this section $\bm{G}$ is an adjacency matrix of a path graph.

\subsubsection{Structure of equilibria of path graphs}
By Theorem~\ref{th:bka:ne}, if $\delta < 1/|\lambda_{\min}(\bm G)|$, then there exists a unique equilibrium.
In addition, since degree of a node in a path graph is at most $D = 2$, it is easy to see that for $\delta < 1/D < 1/|\lambda_{\min}(\bm G)|$, all the agents in the unique equilibrium are active.
Indeed, an agent with $d$ neighbors can obtain external activity of up to $d\delta$. 
Hence, for the agent to be inactive we need to satisfy $d\delta\geq 1$.
If this does not hold for a given agent, the agent must be active in equilibrium. Hence if $D\delta<1$, no agent meets the criteria for inactivity and, therefore, all agents are active.

We have $D=2$ and $\lambda_{\min}(\bm G)=2\cos(\frac{n}{n+1}\pi)$ as shown in~\cite[p.~9]{brouwer2012spectra}.
This means that, for $\delta < 1/2$, there is a unique equilibrium with all agents active and for $\delta<-1/(2\cos \frac{n}{n+1}\pi)$ there is a unique equilibrium.
In particular, these values are not the same and so we may obtain a unique equilibrium where not all of the agents are active.
This happens for paths of odd length as we will see in more detail later.

Now let us classify the possible active sets. Every subgraph of a path graph is a disjoint union of smaller path graphs.
Hence, the active sets will be composed of blocks (connected subgraphs of $G$ which are paths or independent vertices) of varying lengths, with agents who choose non-zero strategies and agents who choose inactivity between the blocks. 
The level of activity of endpoints of the consecutive components must sum to at least $1/\delta$ for the inactive agent between them to remain inactive.
For $\delta<1$, this implies that the inactive gaps can have length of at most $1$ and that the endpoints need to be active.
Hence a configuration of a set of active agents in equilibrium can be represented by a sequence of positive integer values, $(a_1,\ldots,a_k) \in \mathbb{Z}_{> 0}$, representing the sizes of subsequent blocks. We will denote such a configuration by $a_1-\dots-a_k$ (c.f. Figure~\ref{fig:path_configuration} for an example).
\begin{figure}[h]
    \centering
        \begin{tikzpicture}[main_node/.style={circle,fill=blue!20,draw,minimum size=1em,inner sep=3pt]}]
            \node[main_node] (1) at (-6,0) {1};
            \node[main_node][fill=red] (2) at (-4,0)  {2};
            \node[main_node] (3) at (-2,0) {3};
            \node[main_node][fill=red] (4) at (0,0) {4};
            \node[main_node] (5) at (2,0) {5};
            \node[main_node][fill=red] (6) at (4,0) {6};
            \node[main_node] (7) at (6,0) {7};
            \node[main_node] (8) at (8,0) {8};
            \draw (1)--(2)--(3)--(4)--(5)--(6)--(7)--(8);
        \end{tikzpicture}
    \caption{The configuration $1-1-1-2$ on the path graph of length $8$. Inactive agents are marked red.}
    \label{fig:path_configuration}
\end{figure}
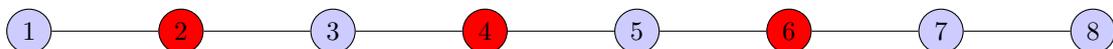
Given a configuration $a_1-\ldots-a_k$, we will say that it \emph{supports an equilibrium}, for a given $\delta\in [0,1]$, if there exists an equilibrium with the active set consisting of blocks of sizes $(a_1,\ldots,a_k)$, separated by inactive agents. We will also use $\bm{G}(a_i)$ to denote the adjacency matrix of the sub-path of the path graph restricted to $a_i$ agents in the $i$'th block, or an \emph{$a_i-$block}, for short.
For each inactive agent lying between two blocks of active agents, we will say that this agent \emph{separates} the two blocks.

We will now analyze which blocks of active agents can occur in equilibria.
The analysis naturally consists of two steps.
Firstly, we need to establish which active configurations are viable by themselves.
Secondly, we need to establish the conditions for a configuration to support its inactive agents.

To classify the valid blocks we begin by noting that the eigenvalues of a path graph over $n$ vertices are given by $2\cos(\frac{k\pi}{n+1})$, for $k\in\{1,\dots n\}$, with the smallest eigenvalue, $\lambda_{\min}(\bm{G}) = 2\cos{(\frac{n}{n+1}\pi)}$~\cite[p.~9]{brouwer2012spectra}.

We start with a proposition that characterizes the unique stable equilibria in the case of $\delta < 1/|\lambda_{\min}(\bm{G})|$.

\begin{prop}
\label{pr:path:uneq}
	Let $G$ be a path graph over $n$ nodes.\\
	If $\delta < 1/2$, the unique equilibrium has $n$ active agents.
	If $1/2<\delta<1/(2\cos{(\frac{n}{n+1}\pi)})$	then
	\begin{enumerate}
	\item if $n = 2m$, the unique equilibrium has $n$ active agents,
	\item if $n = 2m+1$, the unique equilibrium has $m+1$ active agents and the configuration of the set of active agents is $1-\ldots-1$.
\end{enumerate}
\end{prop}



In general, equilibria on a path graph may feature inactive agents. In particular, by Theorem~\ref{th:bka:ne}, such an equilibrium always exists when the externality factor is above the lowest eigenvalue threshold: $\delta > 1/|\lambda_{\min}(\bm{G})|$. The next result provides conditions for stability of such equilibria.
 
As we discussed above, the active set of agents in equilibrium consists of blocks of different lengths, separated by single agents.
Hence each such block is adjacent to two or one inactive agents. If the adjacent inactive agents remain inactive throughout the dynamic, 
the sub-path consisting of the active agents in the block is stable whenever it would be stable as a standalone path equilibrium.
If it is not stable as such, it will also not be stable in the larger graph. 
Hence, we need to investigate whether the inactive agents will remain inactive.

Let $b(\delta, k)$ denote the activity level of an endpoint agent, in the unique all-active equilibrium for a path of length $k$ and a given externality factor $\delta\in [0,1]$, assuming that such an equilibrium exists for the given $\delta$ and $k$. 
For $\delta$ such that the all-active equilibrium for the path of length $k$ is stable, this equilibrium is also unique, making $b(\delta, k)$ well-defined.

If, in an equilibrium, there is an inactive agent separating two blocks of active agents of lengths $k$ and $l$, it must be that
$$\delta(b(\delta, k)+ b(\delta, l))\geq 1,$$
for otherwise this agent would benefit from choosing a positive activity.

Finally, there is one more condition that must be satisfied by a stable equilibrium with inactive agents. The lengths of active paths in the configuration and the number of inactive agents must add up to the total path length. Combining all these conditions, we obtain necessary and sufficient conditions for a given configuration of block sizes to support a stable equilibrium.
\begin{thm}
\label{thm:paths}
    Let $\bm{G}$ be a path of length $n \geq 2$. For almost any $\delta \in [0,1]$, a configuration $a_1 - \ldots - a_k$ supports a stable equilibrium if and only if the following conditions hold:
    \begin{enumerate}
        \item $\forall i \in \{1,\ldots,k\}:\ \delta< -1/(2\cos(\frac{a_i\pi}{a_i+1}))$,\footnote{For $a_i=1$, we consider this condition met.}\label{p:paths:1}
        \item $\forall i \in \{1,\ldots,k\}:\ (\bm{I} + \delta\bm{G}(a_i))^{-1}1>0$,\label{p:paths:2}
        \item $\forall i \in \{1,\ldots,k\}: b(\delta, a_i) + b(\delta,a_{i+1}) > 1/\delta$,\label{p:paths:3}
        \item $\sum_{i = 1}^k a_i = n-k+1$.\label{p:paths:4}
    \end{enumerate}
\end{thm}

Furthermore, condition~\ref{p:paths:3} of Theorem~\ref{thm:paths} can be expressed more concretely for particular ranges of $\delta$.
\begin{prop}
    \label{pr:sequences}
    For $\delta\in (1/|\lambda_{\min}\bm G(2m+2)|, 1/|\lambda_{\min}(\bm G(2m))|)$, we have
    \begin{enumerate}
        \item $b(\delta, 2m)+b(\delta, 1)> 1/\delta$,\label{p:seq:1}
        \item $b(\delta, 1)+b(\delta, 2k)<1/\delta$ for $1<k<m$,\label{p:seq:2}
        \item $2\delta b(\delta, 2m)<1$ and $\lim_{\delta\uparrow\delta_c}2\delta b(\delta, 2m)=1$, where $\delta_c=1/|\lambda_{\min}(\bm G(2m))|$.\label{p:seq:3}
    \end{enumerate}
    In particular, the stable equilibria for $\delta\in (1/|\lambda_{\min}\bm G(2m+2)|, 1/|\lambda_{\min}(\bm G(2m))|)$ are supported exactly by the block configurations $a_1-\dots-a_k$ in which $\forall_i (a_i, a_{i+1})\in\{(1,1), (1, 2m), (2m, 1)\}$.
\end{prop}

We illustrate the findings of this section in the following example. 
A block of length $1$ is always stable. A block of length $2$ is stable for $\delta<1$ and both agents in the block choose activity level $1/(1+\delta)$ in equilibrium. A block of length $3$ supports a stable equilibrium for $\delta<1/2$.
Finally, the block of length $4$ is stable for $\delta < (\sqrt5-1)/2=1/\phi$.\footnote{The constant $\phi\approx1.618$ is also known as the \emph{golden ratio}. It has $1/\phi\approx0.618$} 
Larger blocks are also unstable for $\delta \geq 1/\phi$.
Thus it follows that stable equilibria for $\delta \in (1/\phi, 1)$ only consist of blocks of lengths $1$ and $2$.
For this to be possible for two blocks of length $2$ separated by a common inactive agent, we need
\begin{align*}
 \frac{2\delta}{1+\delta}\geq 1,
\end{align*}
which is not satisfied for $\delta<1$. Hence of a block of length two share a common inactive agent with another block, that block must be of length $1$.
The condition for the separating agent to be inactive in equilibrium is
\begin{align*}
    \delta+\frac{2\delta}{1+\delta}\geq 1\iff\delta\geq 1/\phi\approx0.618,
\end{align*}
which is always satisfied for $\delta \in (1/\phi, 1)$.
This reasoning is generalized for blocks of length $l>2$ in Proposition~\ref{pr:sequences}.

These observations allow us to construct the recurrence for the number of equilibria on a path of length $n$ for $\delta\in(1/\phi,1)$.
By either placing a block of length $1$ at the start and then filling the rest with any other valid equilibrium or by placing a block of length $2$ and then following with a block of length $1$ and anything afterwards we obtain the following recurrence:
\begin{align*}
    e_n=e_{n-2}+e_{n-5},
\end{align*}
with the initial conditions $a_1=1, a_2=1, a_3=1, a_4=2, a_5=1$.

Solving the appropriate quintic polynomial numerically one can show that $e_n$ grows approximately as $1.2365^n$.
Similar analysis can be done for smaller values of $\delta$ and also leads to exponentially increasing number of stable equilibria.

In general dynamical systems, we expect the convergence to slow down as we approach the parameters for which the system is not stable. 
Theorem~\ref{thm:paths} informs us what the stability thresholds of particular components of the system are, suggesting the values of externality factor for which the convergence will be slow.
When analyzing the dynamics, we will use informal terms to describe the convergence time plots.
We will use the word \emph{peak} to describe the behavior near the values of $\delta$ at which local maxima of convergence time are. 
Pinpointing the exact value of $\delta$ at which the local maxima occurs can be difficult and we will often refer to approximate values.
When explaining slow convergence, we will refer to equilibria as \emph{close to losing stability}.
This phrase means that the given point of interest $\delta_c$ is the stability threshold of the equilibrium and the phenomenon occurs for $\delta\in(\delta_c-\epsilon,\delta_c)$ for some value of $\epsilon>0$ small enough. For example, we will observe that the convergence times increase as $\delta\rightarrow1=-1/(2\cos{\frac{2\pi}{3}})$ with $\delta<1$ and that the convergence is fast at $\delta=1$. 
We will refer to such behavior as the convergence slowing down due to the loss of stability near the threshold $\delta=1$.
Similarly, we can refer to a phenomenon occurring whenever an equilibrium is \emph{barely stable} whenever it occurs for $\delta\in(\delta_c,\delta_c+\epsilon)$ for some $\epsilon>0$ small enough.
Finally, we can refer to a phenomenon occurring whenever a condition of the form $f(x)>0$ for some function $f$ is \emph{barely not met} whenever it occurs for $f(x)\in(-\epsilon, 0)$ for $\epsilon>0$ sufficiently small.

\subsubsection{Small path dynamics}
In this section we consider the best-response dynamics on paths with small number of nodes.
In the case of $n=1$ the dynamics converges in a single round.
The case of $n = 2$ is more interesting. 

The eigenvalues of $\bm G$ are $\pm 1$ and so for $\delta<1$ the equilibrium is unique. It is given by $\bm x^*=(\frac{1}{1+\delta},\frac{1}{1+\delta})$ and we can explicitly compute the rate at which we converge to it.
Consider a system where we start at strategy profile $\bm{x}(0) =(x_1(0), x_2(0))$.
Suppose we let the agents change their strategies for $T$ rounds. 
If the same agent, $i$, is chosen to change twice in a row, the change is idempotent, as $i$ is already playing the best response. 
Therefore we introduce effective time $\tau$, distributed according to $Binom(2T ,1/2)$ and interpreted as the number of alternating changes. 
Without loss of generality, we assume that agent $1$ changes its strategy first.

Let $f(z)=\max(0, 1-\delta z)$. 
We have that 
\begin{align*}
    x_1(2\tau+1) &= f(x_2(2\tau))\\
    x_2(2\tau+2) &= f(x_1(2\tau+1)).
\end{align*}
It suffices to analyze the sequence
\begin{equation*}
    x_1(2\tau+1)=f(f(x_1(2\tau-1))).
\end{equation*}
For $\delta,z \leq 1$, we have $f(z)=1-\delta z$ and so $f(f(z))=1-\delta+\delta^2 z$ and, for $\delta< 1$,
\begin{equation}
    x_1(2\tau+1) = (1-\delta^{2\tau}) \frac{1-\delta}{1-\delta^2}+\delta^{2\tau}x_1(1)
              = \frac{1-\delta^{2\tau}}{1+\delta}+\delta^{2\tau}x_1(1),
\end{equation}
with $x_1(1) = 1 - \delta x_2(0)$.
Hence, the convergence rate is exponential and it becomes much slower when we approach $\delta=1$.
The expected number of rounds to converge to an equilibrium is of the order of $\log(\epsilon)/\log(\delta)$ rounds.
At $\delta = 1$ the equilibrium is no longer unique and any strategy profile of the form $(1-z, z)$, for any $z \in [0,1]$, is an equilibrium.
Hence the convergence rate slows down when $\delta$ approaches a value at which the set of equilibria changes. 
At $\delta=1$, we get instant convergence to any one of the infinitely many equilibria:
\begin{align*}
    x_1(2t) = x_1(2) = 1 - x_2(0).
\end{align*}
This behavior is consistent with the simulated results in Figure~\ref{fig:p2}. The convergence rate slows down near $\delta=1$ and is again fast at $\delta=1$.

\begin{wrapfigure}{r}{0.3\textwidth}
    \centering
    \includegraphics[width=\linewidth]{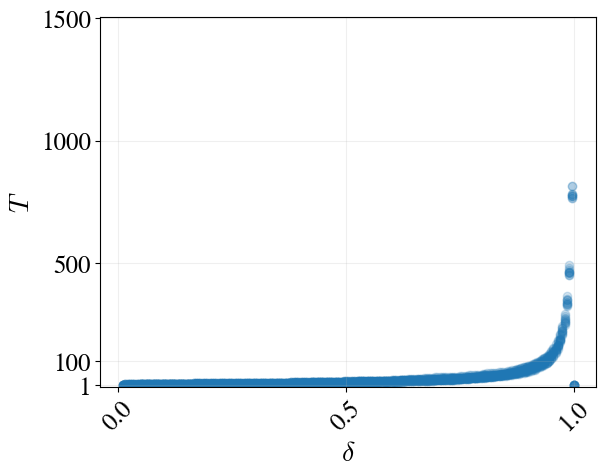}
    \caption{Convergence times $T$ for the path graphs over $n = 2$ nodes.}
    \label{fig:p2}
\end{wrapfigure}
We now turn to analyzing the convergence rate for $n \geq 3$ nodes. To do this, we examine how the classification of equilibria on small paths and their qualitative changes affect the dynamics. The analysis shows, in particular, what happens when we get close to invalidating some of the conditions for stability stated in Theorem~\ref{thm:paths}. The phenomena that we observe for small graphs will be helpful later to explain the results for larger graphs.

When the number of nodes grows, the number of possible configurations of equilibrium blocks of active agents and, consequently, the complexity of the best-response dynamics, increases significantly. This makes the problem prohibitive for formal analysis and most of our findings are based on numerical simulations. Still, we are able to support some of them findings by formal results.

We start with the case of $n = 3$. The convergence rates, for different values of $\delta\in [0,1]$, are presented in Figure~\ref{fig:p3}.
We observe that the diagram features a single peak located around the value of $\delta = 1/2$, which is different from the threshold for the possible loss of uniqueness of the equilibrium, which occurs at $1/|\lambda_{\min}(\bm{G})| = \sqrt{2}/2$.

\begin{wrapfigure}{r}{0.3\textwidth}
    \centering
    \includegraphics[width=\linewidth]{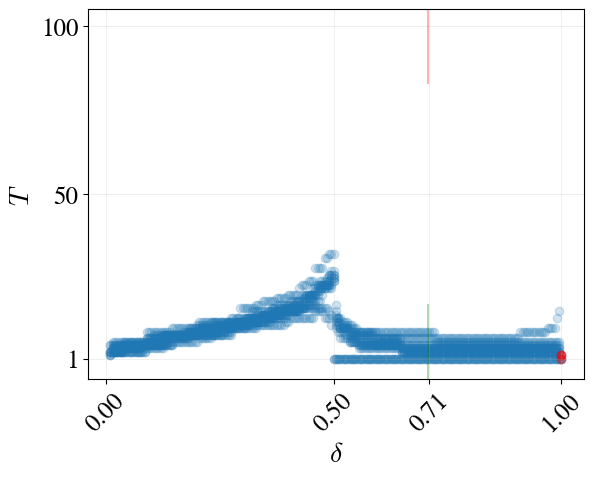}
    \caption{Convergence times $T$ for the path graphs over $n = 3$ nodes.}
    \label{fig:p3}
\end{wrapfigure}
This is because, although the equilibrium remains unique, it changes qualitatively. Instead of all agents being active, we arrive at an equilibrium where the inside agent does not invest, while the outer agents do (c.f. Proposition~\ref{pr:path:uneq}).
While the agents' activity levels near the value of $\delta=1/2$ change continuously in $\delta$, the change of convergence times seems to be sharp. 
This is because the active set we converge to changes. 

Indeed, when $\delta$ crosses the threshold $1/2$, the best-response dynamics converges to the equilibrium $1-1$ and convergence to this equilibrium is quick. Once the middle agent becomes inactive, the other agents converge to the equilibrium in their next move.
Furthermore, crossing the threshold of $1/|\lambda_{\min}(\bm{G})|$ (indicated by the red vertical line in Figure~\ref{fig:p3}) does not affect the convergence rate significantly. 
This threshold is related to the stability of the equilibrium with active set $S=\{1,2,3\}$, which is not the active set for $\delta>1/2$ and so crossing this threshold does not affect the dynamics.

The largest values observed near the peak are significantly smaller than the values observed near the peak near $\delta=1$ for $n=2$ (c.f. Figure~\ref{fig:p2}).
The main difference between the two cases is that, for $n=2$, the unique equilibrium has $n$ active agents for all values of $\delta<1/|\lambda_{\min}(\bm{G})| = 1$, whereas for $n=3$, all agents are active only for $\delta<1/2<1/|\lambda_{\min}(\bm{G})| = \sqrt{2}/2$. 
If the unique equilibrium for $n=3$ had $n$ active agents for all values of $\delta<1/|\lambda_{\min}(\bm{G})| = \sqrt{2}/2$, we could expect to see a similar peak as for $n=2$ near $\delta=1$. 
As the unique equilibrium's active set changes earlier than at the loss of stability threshold, we only observe a small increase in the convergence times.

We observe no change in the speed of convergence when $\delta$ approaches $1$. 
While the equilibrium is no longer unique at $\delta=1$,\footnote{
In this case, there are only $2$ equilibria at $\delta=1$. They have active sets $\{1, 3\}$ and $\{2\}$.
} the convergence times do not increase as we get near this threshold. For $n=2$, this increase is caused by the existence of active block of length $2$ in the configuration and blocks of length $2$ loose stability at $1$. Here, the unique equilibrium has configuration $1-1$ and so has no such block.

Next, we consider the case of $n=4$ (c.f. Figure~\ref{fig:p4}, Figure~\ref{fig:p4_lastSchange} and Figure~\ref{fig:p4_convfail}). 

\begin{figure}[h]
\centering
\begin{minipage}{.3\textwidth}
\centering
    \includegraphics[width=\linewidth]{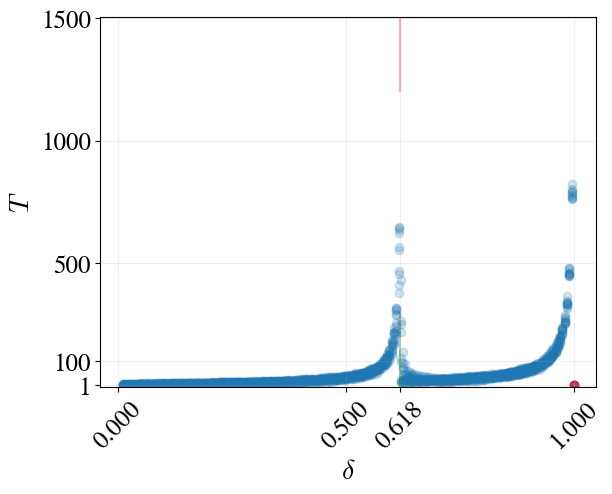}
    \caption{Convergence times $T$ for the path graphs over $n = 4$ nodes.}
    \label{fig:p4}
\end{minipage}%
\hfill
\begin{minipage}{.3\textwidth}
\centering
    \includegraphics[width=\linewidth]{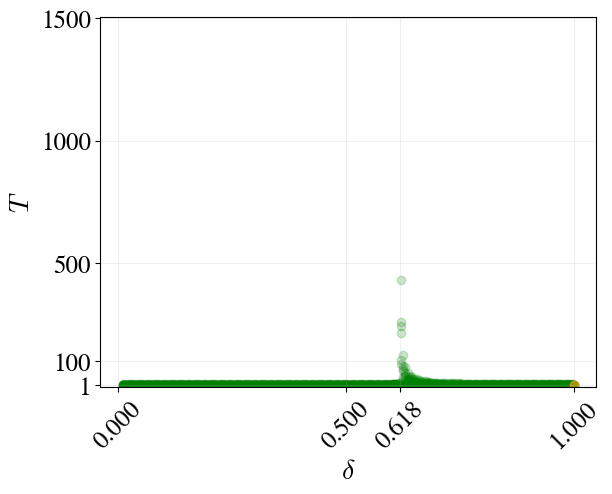}
    \caption{Time of the final active set change for the path graphs over $n = 4$ nodes.}
    \label{fig:p4_lastSchange}
\end{minipage}%
\hfill
\begin{minipage}{.3\textwidth}
\centering
    \includegraphics[width=\linewidth]{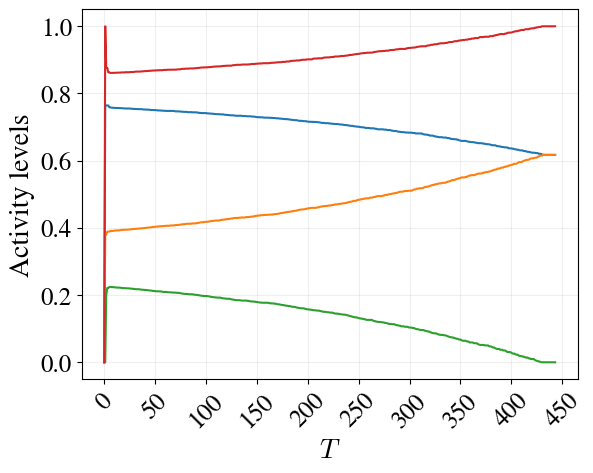}
    \caption{Sample trajectory for a path graph over $n = 4$ nodes and $\delta=0.62$}
    \label{fig:p4_convfail}
\end{minipage}%
\end{figure}

We observe two significant slowdowns of the convergence rate near the values of $\delta = 1/|\lambda_{\min}(\bm{G})|$ and $\delta = 1$, 
where one of the equilibria loses its stability.
Unlike in the case of $n=3$, past the threshold $\delta = \lambda_{\min}(\bm{G})$ equilibrium is no longer unique. 
Three different equilibria exist, described by the configurations: $4, 2-1, 1-2$. 
The equilibrium with all agents being active does not satisfy condition~\eqref{p:paths:1} of Theorem~\ref{thm:paths} and is not stable for $\delta\geq\lambda_{\min}(\bm{G})$. The latter two equilibria are stable for $\lambda_{\min}(\bm{G})\leq \delta<1$.

Furthermore, in Figure~\ref{fig:p4_lastSchange} we observe the time taken to establish the final active set. 
In this plot, a single peak is visible near the value $\delta = \lambda_{\min}(\bm{G})$, with Figure~\ref{fig:p4_convfail} showcasing sample dynamics near this value.
The agents diverge from the equilibrium with $4$ active agents and converge to the equilibrium with configuration $2-1$.
The divergence is slow due to the value of $\delta$ being near the stability threshold.
This is the first example we observe where the shuffle phase is long compared to the asymptotic phase, as it only end while the active set is fully established.

For large values of $\delta$, the active set is again established quickly.
After it it established, the dynamics in the block of $2$ are the same as those described in the analysis for $n=2$, as the inactive agent does no influence them.
This explains the presence of the peak near $\delta=1$.

Next, we consider the case of $n=5$ (c.f. Figure~\ref{fig:p5}).
Overall, the best-response dynamics with $n = 5$ is similar to that for~\hbox{$n = 3$}.
\begin{wrapfigure}{r}{0.3\textwidth}
        \includegraphics[width=\linewidth]{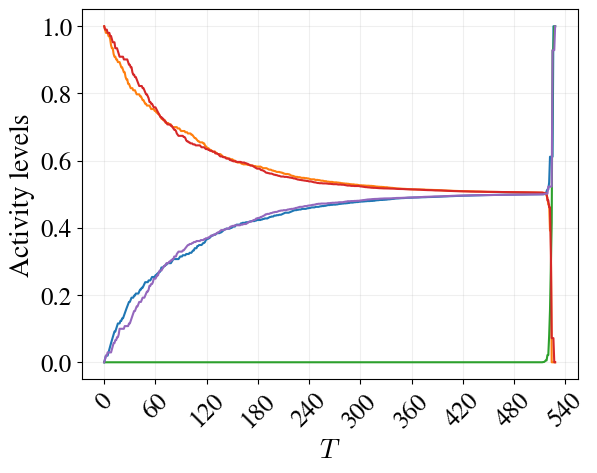}
        \caption{Sample trajectory for a path graph over $n = 5$ nodes and $\delta=0.99$ with the stopping condition using $\epsilon=10^{-5}$ to demonstrate the reshuffle.}
        \label{fig:switchup}
\end{wrapfigure}
\begin{wrapfigure}{r}{0.3\textwidth}
        \includegraphics[width=\linewidth]{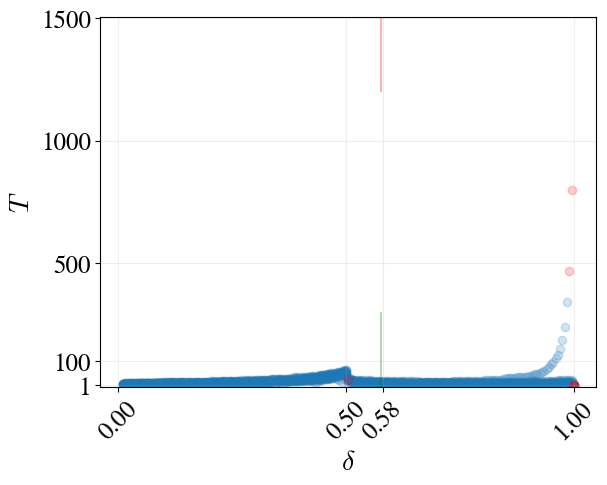}
    \caption{Convergence times $T$ for the path graphs over $n = 5$ nodes.}
    \label{fig:p5}
\end{wrapfigure}

We observe a small peak near the value $\delta=1/2$ and a larger peak near the value $\delta=1$.
Furthermore, only a part of the trajectories converge slow near $\delta=1$.

Let us classify important configurations and explain the behavior observed in the plot.
Below the threshold of $\delta=1/2$ there is a unique stable equilibrium with all agents being active.
Beyond this threshold, the unique equilibrium is of the form $1-1-1$ and it is the only equilibrium for the values of $\delta\in[1/2,1)$.
This was also the case for $n=3$ and explains why the behavior is similar near $\delta=1/2$.

There is another configuration, $2-2$, which is a valid equilibrium candidate as it satisfies conditions \ref{p:paths:1}, \ref{p:paths:2} and \ref{p:paths:4} of Theorem~\ref{thm:paths}. 
The blocks, by themselves, would be stable equilibria but they do not meet condition~\ref{p:paths:3} and so do not induce an equilibrium for the entire graph in the case $n=5$.
The equilibrium activity levels of the pairs of agents in the two active blocks are insufficient to support the separating, middle agent, to remain inactive in equilibrium.
Nevertheless, the best-response dynamics may still converge to these two stable configurations, locally. When the two pairs are converging, it is possible for agents neighboring the separating agent to have the higher activity level level throughout. 
As long as this is the case, the middle agent has an incentive to remain inactive. This stops being the case only when the two sub-paths are close to their own, local, equilibria. This behavior results in slower convergence rate, which becomes more pronounced when the value of $\delta$ is close to $1$.

Such phenomena, where fragments of a graph, separated by inactive agents, converge locally to stable activity profiles and then get perturbed by the separating agents switching from inactivity to activity, play a major role in general. We will refer to them as \emph{reshuffles}. To investigate this closer on a path with $n=5$ nodes, we present the dynamics of the individual activity levels in Figure~\ref{fig:switchup}. Each trajectory in the figure corresponds to a different agent. In particular, the trajectory of the inactive agent separating two pairs of active agents is represented by the green color. As we can see in the figure, the activity level of this agent remains at $0$ for a long time, until the remaining agents converge close enough to their local equilibria. Only then the activity level around the inactive agent becomes insufficient and the agent responds with increasing its activity level and becoming active.
This initiates convergence of the best-response dynamic to a different activity profile, $1-1-1$, which is an equilibrium.
Such a process of convergence to a ``false equilibrium'' may take a long time. During that process the activity profile barely fails to satisfy 
condition~\eqref{p:paths:3} of Theorem~\ref{thm:paths}, which becomes clear after a long time due to condition~\eqref{p:paths:1} of Theorem~\ref{thm:paths} being only barely satisfied.

When the value of $\delta$ is close to $1$, the two pairs converge slowly. The value they converge to is close to satisfying the threshold for inactivity of the separating agent and insufficient activity in the neighborhood of the separating agent may occur arbitrarily late as $\delta$ goes to $1$.
It is important to note here that if the value of $\epsilon$ which determines the stopping condition for the simulation of the best-response dynamics is too large, we may not notice the aforementioned behavior at all and be mistakenly convinced that the dynamics converged to a stable equilibrium.

The order of the updates determines which of the behaviors occurs in the dynamics. In particular, the slow convergence happens if we draw the update order leading to the state $01010$ in the beginning. It is simple to see that we then have the pair dynamics described above on the sides and constant inactivity for a long time in the center. This scenario causes the peak to appear near $\delta = 1$ for some of the simulations, while many of the other trials converge quickly. The existence of slow trajectories is formalized by Fact~\ref{thm:p5}. This points illustrates why the random order of updates in the best-response dynamics is important. In particular, it allows for discovering behaviors that could, otherwise, remain unnoticed. 

\begin{fact}
\label{thm:p5}
Suppose $\frac12<\delta<1$ and $n=5$. Let $T$ denote the number of rounds taken until the last active set change.
We have that
\begin{align*}
\mathbb P(5T > C(\delta))\geq2/25
\end{align*}
where $C(\delta)\to\infty$ with $\delta\to 1$.
\end{fact}
Empirically, reshuffles are more common than the above theorem suggest. In the sample size of $10000$, slow convergence occurred for $20.89\%$ of the total trajectories.

\begin{figure}[h]
\centering
\begin{minipage}{.3\textwidth}
\centering
    \includegraphics[width=\linewidth]{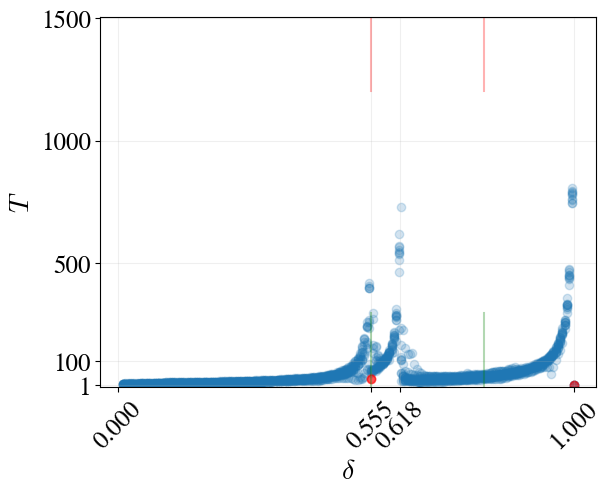}
    \caption{Convergence times $T$ for the path graphs over $n = 6$ nodes.}
    \label{fig:p6}
\end{minipage}%
\hfill
\begin{minipage}{.3\textwidth}
\centering
    \includegraphics[width=\linewidth]{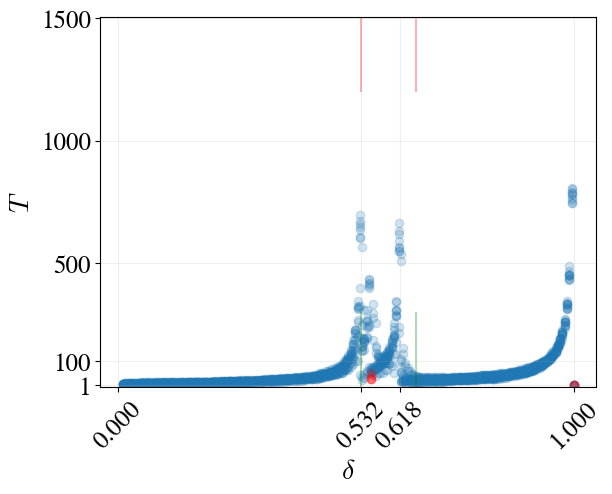}
    \caption{Convergence times $T$ for the path graphs over $n = 6$ nodes.}
    \label{fig:p8}
\end{minipage}%
\hfill
\begin{minipage}{.3\textwidth}
\centering
    \includegraphics[width=\linewidth]{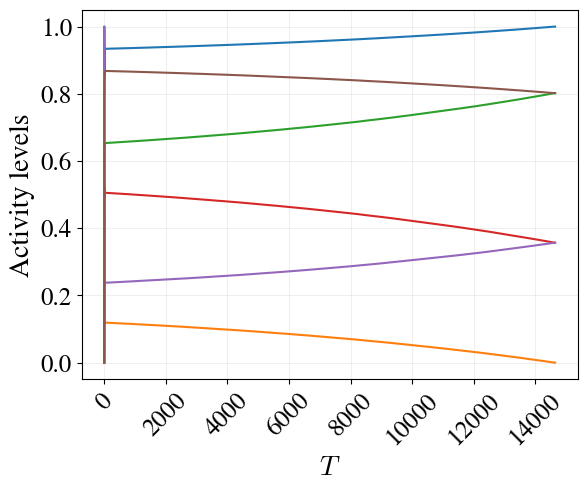}
    \caption{Sample trajectory for $n=6$ and $\delta=0.555$}
    \label{fig:p6_critical}
\end{minipage}%
\end{figure}

The behavior of the best-response dynamics in the cases of $n = 6$ and $n = 8$ are similar. The convergence dynamics for these cases are presented in Figures~\ref{fig:p6} and~\ref{fig:p8}.
For the first time, we observe multiple peaks for intermediate values of $\delta \in (1/2,1/|\lambda_{\min}(\bm{G})|)$.
We observe a similar behavior in the case of larger, general graphs.

In the cases $n = 6$ and $n = 8$ the possible equilibria are still simple to classify and help explain the singularities.
The singularities occur near the values of $\delta_1 = 1/|\lambda_{\min}(\bm{G})| = -1/(2\cos(\frac{n\pi}{n+1}))$ and $\delta_2 =(\sqrt{5}-1)/2$.
For $\delta\in[0,\delta_1)$ there is a unique equilibrium with all agents being active, which loses stability at $\delta = \delta_1$, resulting in the slowdown close to that point.
In fact, we can observe critical behavior here. The sample trajectory in Figure~\ref{fig:p6_critical} has an \emph{almost linear convergence rate}.
This is because the active set on $6$ vertices is barely unstable here, with $\delta_1 \approx 0.554958$.
In effect, the convergence away from the unstable all-active equilibrium is very slow, triggering the stopping criterion (as signified by the corresponding point in Figure~\ref{fig:p6} being red).
This is a more extreme version of what we observed for $n=4$.
After the value of $\delta$ crosses the threshold $\delta_1$, two equilibria of the form $1-4$ and $4-1$ arise for $n=6$. These are the only two stable equilibria and the overall convergence is quick, until the value of $\delta$ reaches the threshold of $\delta_2$, as it is the stability threshold for the block of length $4$.

\begin{wrapfigure}{r}{0.3\textwidth}
        \centering
        \includegraphics[width=\linewidth]{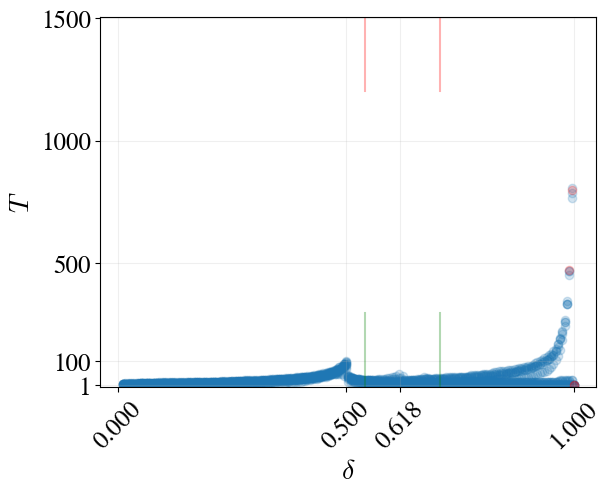}
        \caption{Convergence times $T$ for the path graphs over $n = 7$ nodes.}
        \label{fig:p7}
\end{wrapfigure}
The case of $n = 7$ is similar to the case of $n = 5$ (c.f. Figure~\ref{fig:p7}.). 
The convergence rate around the value of $\delta = 1/2$ is slower than in the case of $n=5$, which was slower than for $n=3$.
This is again caused by the shift of the stability threshold.
The point $\delta=1/2$ where the unique equilibrium changes qualitatively is closer to the threshold $\delta=1/|\lambda_{\min}(\bm G)|$, slowing the convergence more than for the previous odd values of $n$.
The configuration $1-1-1-1$ supports the unique equilibrium for $1/2<\delta<1/\phi$. For larger values of $\delta$, the configuration $2-1-2$ is also a stable equilibrium.
We also obtain the configurations $1-2-2$ and $2-2-1$, which fail only condition~\ref{p:paths:3} of Theorem~\ref{thm:paths}.
Reaching the active sets induced by these configurations leads to different behaviors.
Convergence to the equilibrium induced by $1-1-1-1$ is quick for all values of $\delta$.
For the stable equilibrium induced by $2-1-2$, the convergence is slow near the value $\delta=1$ with the convergence terminating near the equilibrium.
For the configurations $1-2-2$ and $2-2-1$, we observe the phenomenon described in detail for $n=5$.
We slowly converge to a strategy profile which is not an equilibrium, finally changing the active set to the one induced by $1-1-1-1$.
The convergence time of the latter two scenarios is similar, despite one of the trajectories converging to an equilibrium and the other changing at the last moment.

\subsubsection{Large paths}
\begin{wrapfigure}{r}{0.3\textwidth}
    \centering
    \begin{subfigure}[t]{0.3\textwidth}
    \includegraphics[width=\linewidth]{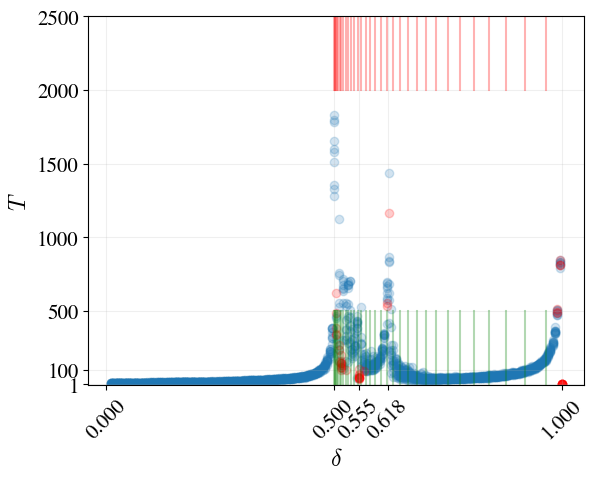}
    \caption{Path with $n=100$ nodes.}
    \end{subfigure}\\
    \begin{subfigure}[t]{0.3\textwidth}
    \includegraphics[width=\linewidth]{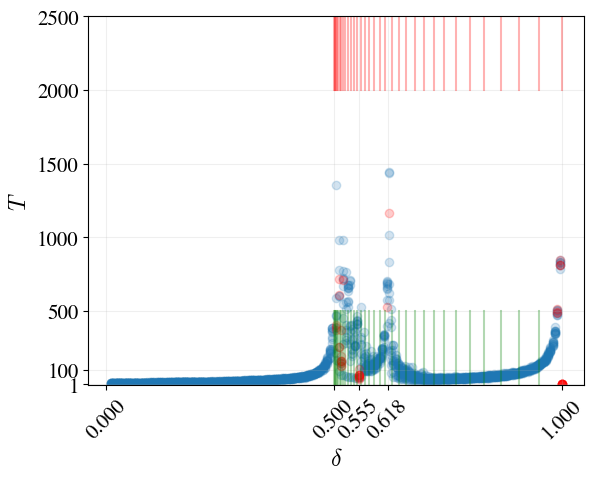}
    \caption{Path with $n=101$ nodes. The trajectories for $\delta=0.5$ are not shown due to long convergence times $T\in(3500, 5000)$ to aid comparability of the overall plots.}
    \end{subfigure}
    \caption{Convergence times $T$ for the path
graphs over $n$ nodes.}
    \label{fig:p100}
\end{wrapfigure}
In this section we analyze the speed of convergence of best-response dynamics on path graphs with large $n$. When $n$ becomes sufficiently large, these plots share similar characteristics across all path graphs, regardless of the parity of $n$. Examples of such plots, for $n = 100, 101$, are presented in Figure~\ref{fig:p100}.

The first peak of convergence times appears around $\delta=1/2$. This is the limit of the sequence of thresholds $1/|\lambda_{\min}(\bm G)|
=1/|2\cos{\frac{n\pi}{n+1}}|$, where equilibrium with all agents active ceases to be stable.
As the thresholds for stability loss approach the value $\delta=1/2$, the difference between paths of different parity diminishes. 
In what we observed, the convergence time near the threshold of stability depended largely on the distance between $\delta$ and this threshold.
For $n=3,5,7$, the peaks at $\delta=1/2$ increased as the value of $1/|\lambda_{\min}(\bm G)|>1/2$ decreased.
For $n=101$, the threshold $1/|\lambda_{\min}(\bm G)|$ is close to $1/2$, leading to a large peak.
At $\delta=1/2$, convergence is slower for odd $n$ than it is for even $n$. The paths of different parity converge to qualitatively different equilibria near this value,\footnote{This is showcased in Appendix~\ref{appendix:parity}. This qualitative difference is only visible close to $\delta=1/2$.} possibly causing the discrepancy between the convergence times.

The differences between paths of different parity are only visible close to $\delta=1/2$. For the values of $\delta<1/2$, the unique equilibrium is different and convergence is slower. This difference persists up to the value $\delta=1/|\lambda_{\min}(\bm G)|$. As a result, the overall plots for odd and even $n$ differ only on a small interval and are overall similar, with the peak near $\delta=1/2$ visible for both $n=100$ and $n=101$ (c.f. Figure~\ref{fig:p100}).

For $\delta\in(\frac1\phi, 1)$ the plot is similar to those observed for small values of $n$ (cf. Figures~\ref{fig:p4},~\ref{fig:p6},~\ref{fig:p8}).
Only the equilibria comprised of blocks of lengths $1, 2$ are stable and the only slowdowns we observe are near the endpoints of the interval. 
As for small values of $n$, this corresponds to either the blocks of length $2$ converging slowly (either to a stable equilibrium or to a non-equilibrium strategy profile, causing a reshuffle) or to blocks of length $4$ diverging slowly due to being close to their respective stability threshold $1/\phi$.

For intermediate values of $\delta\in(\frac12, \frac1\phi)$, there are many slow trajectories and multiple trajectories which stop before reaching a stable equilibrium's active set.

We explain this behavior by again referring to the stable equilibria classification.
Direct enumeration of all equilibria is difficult in this case. To aid understanding of the convergence, we classify which blocks may exist subsequently.
To do this, we will use Proposition~\ref{pr:sequences}. For a system to converge slowly, we only need a single agent to converge slowly. Local understanding of the configurations will therefore be enough to understand the slow convergence.

\begin{figure}[h]
    \centering
    \includegraphics[width=0.5\linewidth]{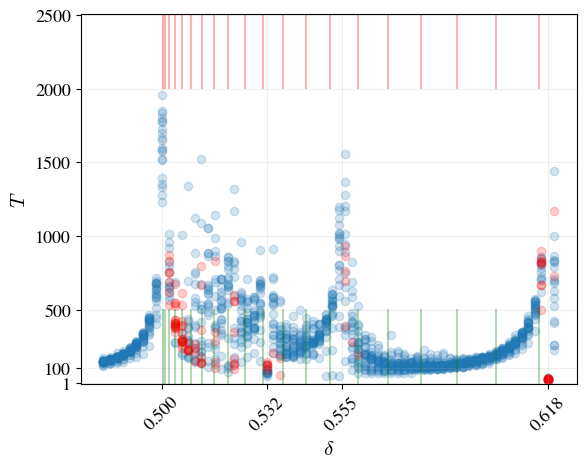}
    \caption{Convergence times $T$ for the path
graphs over $n=100$ nodes, for $\delta\in(0.45, 0.62)$, with increments of $0.002$ and $20$ trajectories per value. The values $0.532, 0.555, 0.618$ are equal to $1/|\lambda_{\min}(\bm G(8))|,1/|\lambda_{\min}(\bm G(6))|,1/|\lambda_{\min}(\bm G(4))|$ respectively.}
    \label{fig:p100zoom}
\end{figure}
We can now focus on the region with the more complex behavior, presented in more detail in Figure~\ref{fig:p100zoom}, with individual trajectories in Figure~\ref{fig:p100_samples}.
In the plot, distinct peaks occur near the values $1/|\lambda_{\min}(\bm G(8))|$, $1/|\lambda_{\min}(\bm G(6))|$, $1/|\lambda_{\min}(\bm G(4))|$ and for the values $\delta\in(1/2, 1/|\lambda_{\min}(\bm G(8))|)$ many slow samples are observed.
This behavior corresponds directly to the study of stable equilibria in the graph, as the consecutive even length blocks lose stability at these thresholds (and odd length blocks support no equilibrium for $\delta>1/2$ as per Proposition~\ref{pr:path:uneq}).

In each part of $(1/2, 1/|\lambda_{\min}(\bm G(4))|)$, described by the thresholds of Proposition~\ref{pr:sequences}, different equilibria are valid. 
In the region $\delta\in (1/|\lambda_{\min}(\bm G(2k+2))|, 1/|\lambda_{\min}(\bm G(2k))|)$, only the equilibria comprised of blocks of lengths $2k$ and $1$, with no two blocks of length $2k$ neighboring one another, are stable.
Indeed, if two such blocks appear sequentially, the trajectory converges to the corresponding strategy profile, possibly only failing to satisfy the inactivity threshold of the separating agent when close to convergence, leading to a reshuffle.
This is similar to the behavior described for $n=5, 7$. 
As we approach the value of $\delta=1/|\lambda_{\min}(\bm G(2k))|$, the blocks of length $2k$ lose stability, leading to slow convergence, and the separating agents' requirements are close to being satisfied (c.f. Statement~\ref{p:seq:3} of Proposition~\ref{pr:sequences}) when surrounded by two blocks of length $2k$, leading to late activity change.
An example of reshuffles is showcased in Figure~\ref{fig:p100_0615}. The trajectory in the figure has $\delta=0.615$ and showcases two reshuffles, occurring in different parts of the graph. As $\delta$ is close enough to the threshold $1/\phi$, the blocks of length $4$ converge slowly. With the threshold for inactivity of an agent separating two blocks of length $4$ is barely not satisfied, reshuffles are possible. 
One of the reshuffles happens earlier than the other. For the later reshuffle, both endpoints of adjacent $4-$ blocks have activity levels higher than the activity level at equilibrium (represented by one of the trajectories converging to $x\approx0.8$ from above). For the earlier reshuffle, one of the endpoints is below the equilibrium activity level (represented by one of the trajectories converging to $x\approx0.8$ from below), making the reshuffle quicker.\footnote{When a $2k-$block converges to the equilibrium, the even agents are above their equilibrium activity levels and the odd agents are above, or vice-versa. In the trajectory plots, pairs of agents with indices $i$ and $2k+1-i$ are represented by trajectories converging to the same point, one from above and the other from below.}
There exist $4-$blocks which are not adjacent to other $4-$blocks. They converge to their equilibrium with no reshuffle.

Another phenomenon may occur within the interval $\delta\in (1/|\lambda_{\min}(\bm G(2k+2))|, 1/|\lambda_{\min}(\bm G(2k))|)$. The blocks of length $2k+2$ may be barely unstable, leading to slow divergence away from them, similar to the behavior noted for $n=4, 6$. 

For $k=1,2,3, 4$, the above behaviors correspond to their corresponding peaks in the dynamics, each subsequent one less pronounced.
For large values of $k$, the phenomena described above overlap, leading to the peaks being less defined.
This is caused by different thresholds being close to one another.
For example, in the region $\delta\in(1/2, 0.532)$, blocks of various lengths are active for long times during the dynamics, despite not being a part of any equilibrium. They remain active and diverge slowly because their thresholds for not being stable are barely met.
We highlight two examples of this in Figures~\ref{fig:p100_0502},~\ref{fig:p100_0515}.

In Figure~\ref{fig:p100_0502}, we observe slow divergence from the equilibrium with all agents active and towards an equilibrium described by the configuration with a block of length $34$ and many $1$'s.
In Figure~\ref{fig:p100zoom}, this trajectory corresponds to a red point as the last active set change occurs when the trajectories are almost constant.
This change of activity leads from a block of length $36$ to a configuration containing $34-1$ in place of $36$. In particular, it does not lead to a reshuffle, despite red points being associated exclusively with reshuffles earlier.

In Figure~\ref{fig:p100_0515}, convergence of two different parts of the graph overlap.
For $\delta=0.515$, the blocks which may appear in the equilibrium are of lengths $1$ and $10$.
The smooth lines visible in the trajectory plot converging to activity levels in $(0,1)$ correspond to a block of length $10$. The other lines correspond to longer blocks, which divide into configurations consisting of blocks of length $1$.
Due to closeness of the different thresholds, the divergence from the longer blocks took a relatively long time.

Large values of $n$ allow the inspection of another key property of the dynamics.
No cases of fast convergence occur near the value of $\delta=1$, whereas they did occur for small $n$.
Out of the exponential number of stable equilibria, at most one contains no $2-$blocks.
The system rarely converges to this equilibrium.
Similarly, for other values of $\delta$, the slowly converging blocks are more likely to appear throughout the dynamics for paths of larger length $n$. For $n=100$, the quickly converging trajectories are not infrequent near $\delta=0.532\approx1/\lambda_{\min}(\bm G(8))$ (c.f. Figure~\ref{fig:p100zoom}). With larger values of $n$, this is not expected to happen. Nevertheless, $n=100$ was sufficient to analyze much of the behavior for large paths.

\begin{figure}
    \centering
    \begin{subfigure}[t]{0.3\textwidth}    
        \centering
        \includegraphics[width=\linewidth]{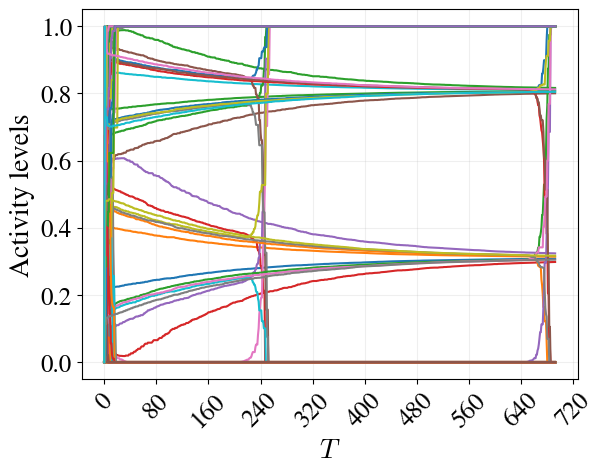}
        \caption{$\delta=0.615$.}
        \label{fig:p100_0615}
    \end{subfigure}
    \begin{subfigure}[t]{0.3\textwidth}
        \centering
        \includegraphics[width=\linewidth]{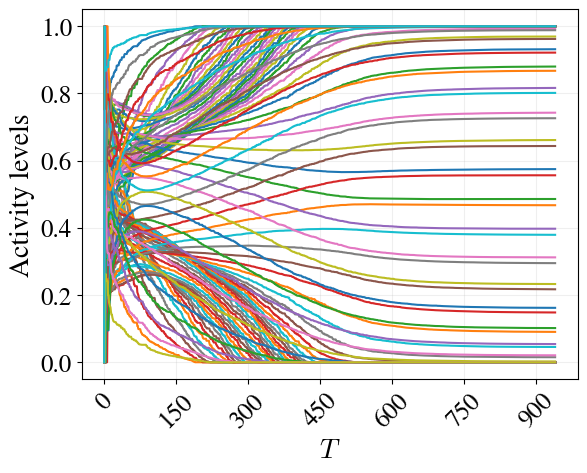}
        \caption{$\delta=0.502$.}
        \label{fig:p100_0502}
    \end{subfigure}
    \begin{subfigure}[t]{0.3\textwidth}    
        \centering
        \includegraphics[width=\linewidth]{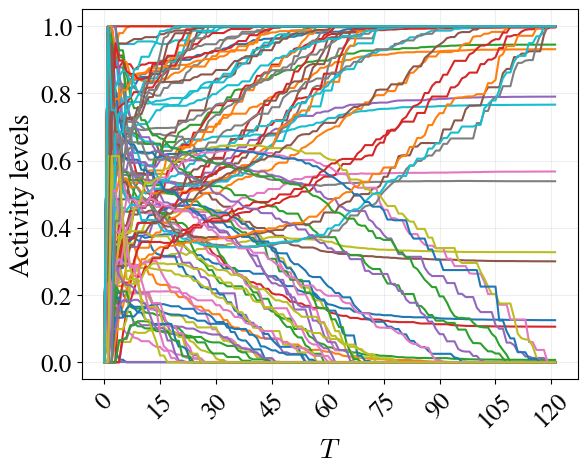}
        \caption{$\delta=0.515$.}
        \label{fig:p100_0515}
    \end{subfigure}
    \caption{Sample trajectories for $n=100$ and values of $\delta\in(0.5, 0.62)$.}
    \label{fig:p100_samples}
\end{figure}

To summarize, we observed the following phenomena leading to slow convergence in paths:
\begin{itemize}
\item stable equilibria consisting of blocks close to losing stability,
\item the trajectory being close to a barely unstable equilibrium, leading to slow divergence away from it,
\item the \emph{shuffle phase} consisting of semi-convergence to a strategy profile which is only an equilibrium on the active set, leading to a change of active set upon reaching activity levels insufficient for inactive agents.
\end{itemize}
The first two phenomena are related to Condition~\ref{p:paths:1} of Theorem~\ref{thm:paths}.
The last one needs Condition~\ref{p:paths:3} together with Condition~\ref{p:paths:1} to be close to failing.
We also observed trajectories skipping the asymptotic convergence phase and arriving at the exact equilibrium. This occurs in the parts of the graph comprised of independent sets.

\subsection{General graphs}
\label{sec:gengraphs}

Throughout this section, we will relate the phenomena observed for paths into a more general setting.
We will study both deterministic and random graphs.
In particular, we will see how the characteristic properties of various graph classes impact the convergence times.

\subsubsection{Structure of equilibria general graphs}

Firstly, we classify the stable equilibria and their active sets with a result similar to Theorem~\ref{thm:paths}.
\begin{thm}
\label{thm:equilibria}
    Let $S\subseteq V(G)$ and let $\delta$ be given.
    Let $(S_l)$ be the maximal connected components of subgraph induced by $S$.
    Suppose that 
    \begin{enumerate}
        \item $\forall l: \delta<1/|\lambda_{\min}(\bm G_{S_l})|.$\label{p:eqilibria:1}\\
    Let now $\bm x^*$ be the equilibrium that is unique on each of the $S_l$ and suppose further that 
        \item $\forall i\in S: x_i^*>0$\label{p:eqilibria:2},
        \item $\forall i\not\in S: \sum g_{ij}x_j^*> 1/\delta$\label{p:eqilibria:3}.
    \end{enumerate}
    Then, $S$ is the active set of a stable equilibrium of the graph $G$ for $\delta$, for all $\delta\in[0,1]$.\\
    Furthermore, for almost any $\delta$, if Condition~\ref{p:eqilibria:1} is not met but an equilibrium with active set $S$ exists, it is unstable.
\end{thm}
While not as intuitive as the classification provided by Theorem~\ref{thm:paths}, this theorem serves a similar purpose.
It enables connected components of active sets to be referred to as close to losing stability or barely insufficient for their inactive neighbors, thereby explaining slow convergence.

Conversely, it does not yield a simple description of all stable equilibria.
We are often not able to establish a priori which sets will be active throughout the dynamics and will be only able to explain slow convergence retrospectively.

For $\delta\in(1/\phi, 1)$, we can strengthen this theorem to fully characterise the equilibria.
In previous work of \cite{Bra07}, independent sets of a given order are listed as possible equilibria.
If they exist, they satisfy the conditions of Theorem~\ref{thm:equilibria} and are therefore stable for almost any $\delta$.
They may, however, not exist.
Here, we show that the active sets of stable equilibria must be disjoint unions of cliques\footnote{A clique is a graph in which $\forall i,j\not=i\in V, g_{ij}=1$}.
The proof is based on the fact that the family of graphs satisfying $\lambda_{\min}(\bm G)/\phi>-1$ is small.
\footnote{This relates to a classic problem of classifying the graphs with $\lambda_{\min}(\bm G)>c$ for some constant $c$. For $c=-2$ this is covered in  \cite{cameron1976line} and could be used to classify the stable equililbria for smaller values of $\delta$ than presented here. The resulting class of graphs is significantly more complicated than what is presented here.}
When combined with the observation that if a clique of size $k$ is active and has no active neighbors, each agent in that clique chooses activity $1/(1+k\delta)$, we get the following result.

\begin{thm}
\label{thm:large_delta_equils}
    For $\delta\in(1/\phi, 1)$, the active set $S$ of any stable equilibrium is comprised of disjoint cliques.\\
    Furthermore, for $i\not\in S$,
    \begin{equation}
    \sum_k \frac{\delta n^i_k}{1+(k-1)\delta}\geq 1,
\end{equation}
where $n^i_k$ denotes the number of neighbors of $i$ inside active cliques of size $k$.
\end{thm}

The function $\delta\mapsto\sum_k \frac{\delta n^i_k}{1+(k-1)\delta}$ is strictly increasing, meaning that if an active set supports a stable equilibrium for a given $\delta$, it also supports it for larger values of $\delta$.

\subsubsection{Deterministic graphs}

\begin{wrapfigure}{r}{0.3\textwidth}
        \includegraphics[width=\linewidth]{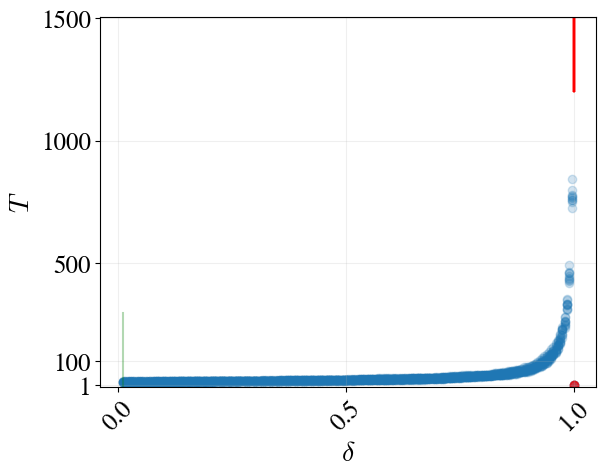}
    \caption{The clique on $100$ vertices}
    \label{fig:k100}
\end{wrapfigure}
We can now analyze specific classes of graphs, beginning with the cliques.
The convergence times on a large clique (c.f. Figure~\ref{fig:k100}) are similar to those observed for the path of length $n=2$.
This is the case for cliques of all sizes\footnote{Note that the times were rescaled by the number of vertices.}.
The equilibrium is unique for $\delta<1$ and it is the one with all agents active.
The activity level at equilibrium is uniform and, for a clique of size $n$, has the value of $\frac{1}{1+(n-1)\delta}$.
As $\delta$ approaches $1$, the convergence slows down drastically.
At $\delta=1$, all strategy profiles with activity levels summing to $1$ are equilibria.

Next, we consider complete bipartite graphs\footnote{A complete bipartite graph is a graph whose set of vertices is comprised of two disjoint parts, $A$ and $B$. Two vertices are connected if and only if they belong to different parts.} with sizes of parts equal to $m$ and $l\leq m$.
The convergence is quick for the majority of values of $\delta$. 
A peak occurs near the value $\delta=1/m$. This peak is smaller for $m\not=l$ compared to the one for $m=l$. Furthermore, for $m=l$, the trajectories converge to states which are not stable equilibria, as signified with the red points.
There are no peaks near $\delta=1$ in any of the plots.
\begin{figure}[h]
    \centering
    \begin{subfigure}[t]{0.25\textwidth}
        \centering
        \includegraphics[width=\linewidth]{pictures/cospectral/K14.png}%
        \caption{$m=4, l=1$.}
    \end{subfigure}%
    \begin{subfigure}[t]{0.25\textwidth}
        \centering
        \includegraphics[width=\linewidth]{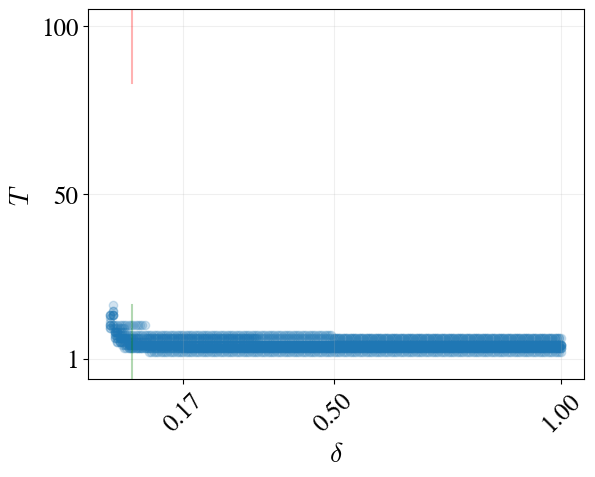}%
        \caption{$m=60, l=5$.}
    \end{subfigure}%
    \begin{subfigure}[t]{0.25\textwidth}
        \centering
        \includegraphics[width=\linewidth]{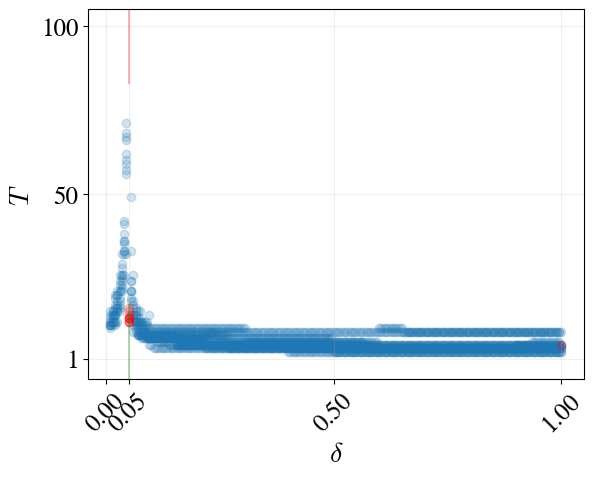}%
        \caption{$m=l=20$.}
    \end{subfigure}%
    \caption{Convergence times $T$ for the complete bipartite
graphs with different $m, l$.}
\end{figure}

The equilibria of complete bipartite graphs can be fully classified.
The smallest eigenvalue is $\lambda_{\min}(\bm G)=-\sqrt{ml}$ (c.f.\cite[p.~8]{brouwer2012spectra}).
For $\delta\leq 1/m$, all agents are active, as no agent can have their activity threshold satisfied by their neighbors. 
For $l=m$, this overlaps the threshold induced by the smallest eigenvalue. The strategy profile with agents of one part choosing activity level $a$ and agents of the other part choosing activity level $1-a$, for any $a\in[0,1]$, is an equilibrium for $\delta=1/m$, whenever $m=l$.
For $\delta\in (1/m, 1/l)$, in the equilibrium, the agents in the part of size $m$ choose activity $1$ and agents in the part of size $l$ are inactive. Finally, for $\delta\in [1/l,1]$, either part has an activity level of $1$ and the other part is inactive.

For $l\not=m$, the peak resembles the peaks observed for small paths of odd length. Indeed, in these cases the equilibrium stability threshold is also higher than the threshold at which the equilibrium changes qualitatively, explaining the small peak.
\begin{wrapfigure}{r}{0.25\textwidth}
        \centering
        \includegraphics[width=\linewidth]{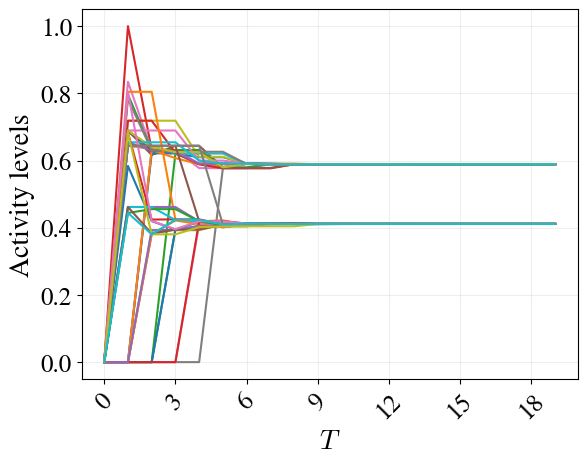}%
        \caption{Sample trajectory for a bipartite graph with $m=20, l=20$ and  $\delta=1/20$.}
        \label{fig:k2020_critical}
\end{wrapfigure}
For $m=l$, the peak occurs near the value $\delta=1/|\lambda_{\min}(\bm G)|=1/m$.
This is where the stability of the all-active equilibrium is lost.
When $\delta$ is slightly larger than $1/|\lambda_{\min}(\bm G)|$, the divergence away from this equilibrium is slow.
At the threshold $\delta=1/m$, the trajectories converge to strategy profiles which are not stable equilibria and the trajectories are represented by red points. Figure~\ref{fig:k2020_critical} showcases an example of such trajectory. The equilibrium is not stable due to lack of attractivity: There exist other equilibria arbitrarily close to it and convergence cannot occur.

For the values $\delta\in (1/m, 1)$, the stable equilibria consist only of independent sets. 
In particular, they contain no active edges, leading to no peak near unity.

\subsubsection{Random graphs}
We now proceed to the study of random graphs.
When studying them, we will be sampling a single graph from a given model and generating the convergence time plot for this single sample.
This is done to distinguish the properties of representatives, as averaging over many samples may make particular features difficult to distinguish.

The convergence times for different types of random graphs are presented in Figures~\ref{fig:ba}, \ref{fig:er}, \ref{fig:rr}.
Before explaining the details observed for each random graph model, we will briefly discuss the similarities observed across all of them.

Firstly, for $\delta<1/|\lambda_{\min}(\bm G)|$, for which the equilibrium is unique, the convergence is quick.
This corresponds to the behavior observed before and no phenomena causing slowdowns tend to occur.

\begin{wrapfigure}{r}{0.4\textwidth}
    \begin{subfigure}[t]{0.2\textwidth}
        \centering
        \includegraphics[width=\linewidth]{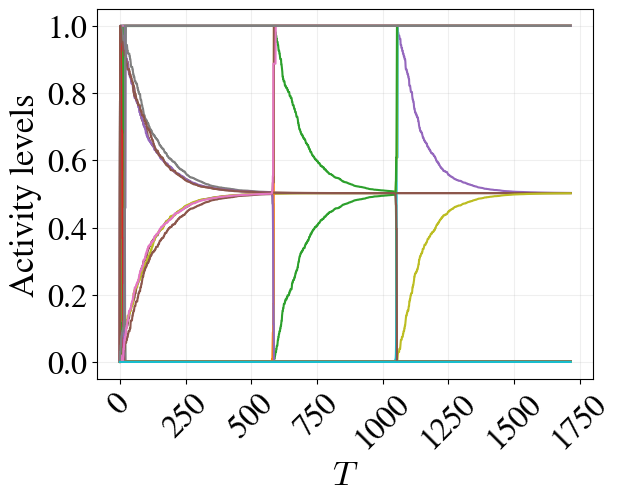}%
        \caption{Three reshuffles.}
        \label{fig:rr_reshuffle_1}
    \end{subfigure}%
    \begin{subfigure}[t]{0.2\textwidth}
        \centering
        \includegraphics[width=\linewidth]{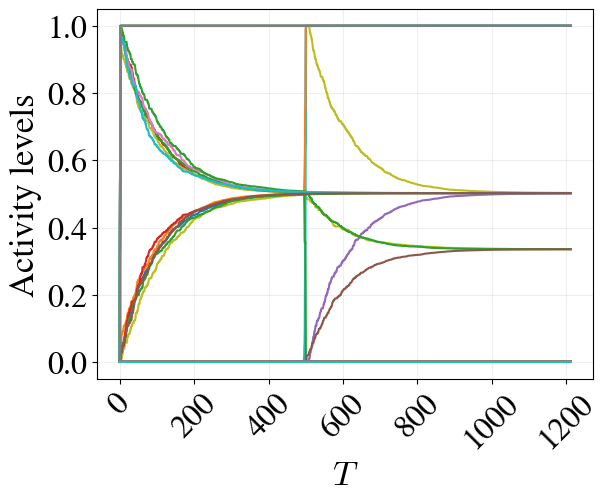}%
        \caption{A reshuffle resulting in a larger active clique.}
        \label{fig:rr_reshuffle_2}
    \end{subfigure}%
    \caption{Sample trajectories of the process for $\delta=0.99$ in random regular graphs with degree $d=10$ and $n=100$ vertices.}
    \label{fig:rr_reshuffles}
\end{wrapfigure}
For $\delta\in(1/\phi, 1)$, the plots are similar to the results shown for paths and cliques.
This behavior can be explained with Theorem~\ref{thm:large_delta_equils}.
Only cliques appear in the active sets of the stable equilibria for $\delta\in(1/\phi, 1)$ and the convergence is therefore the same as that of a clique, with the added possibility of reshuffles occurring if the cliques making up the active set induce a strategy profile which is not an equilibrium.
The reshuffles occurring in general graphs are more complicated than those observed in the path graphs. We will now consider some of the reshuffles we observed during the study of random graphs (c.f. Figure~\ref{fig:rr_reshuffles}).
Figure~\ref{fig:rr_reshuffle_2} showcases that different clique sizes can be involved. A reshuffle of multiple cliques with $n=2$ vertices leads to the emergence of an active clique with $n=3$ vertices.
In Figure~\ref{fig:rr_reshuffle_1}, we observe a trajectory where three reshuffles are chained together, with each one finishing initiating the subsequent one.
Such a sequence of reshuffles can be arbitrarily long (c.f. Section~\ref{sec:reshuffles}).

\begin{figure}[h]
    \centering
    \begin{subfigure}[t]{0.25\textwidth}
        \centering
        \includegraphics[width=\linewidth]{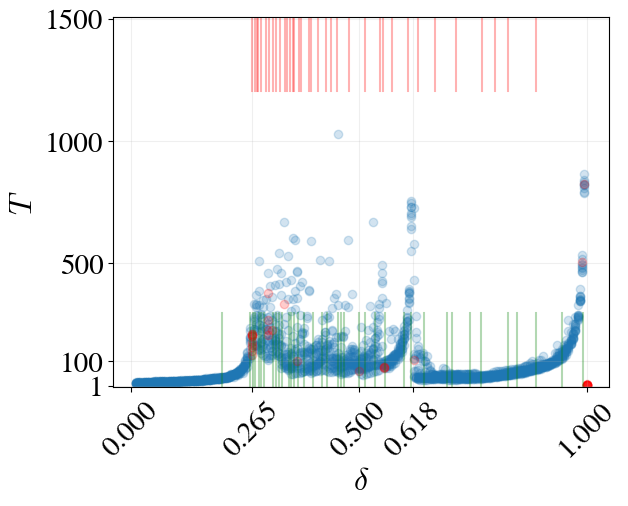}%
        \caption{$d=5$}
    \end{subfigure}%
    \begin{subfigure}[t]{0.25\textwidth}
        \centering
        \includegraphics[width=\linewidth]{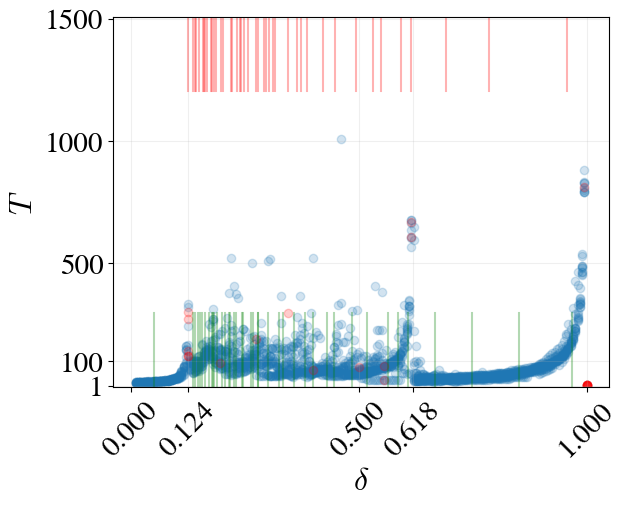}%
        \caption{$d=20$}
    \end{subfigure}%
    \begin{subfigure}[t]{0.25\textwidth}
        \centering
        \includegraphics[width=\linewidth]{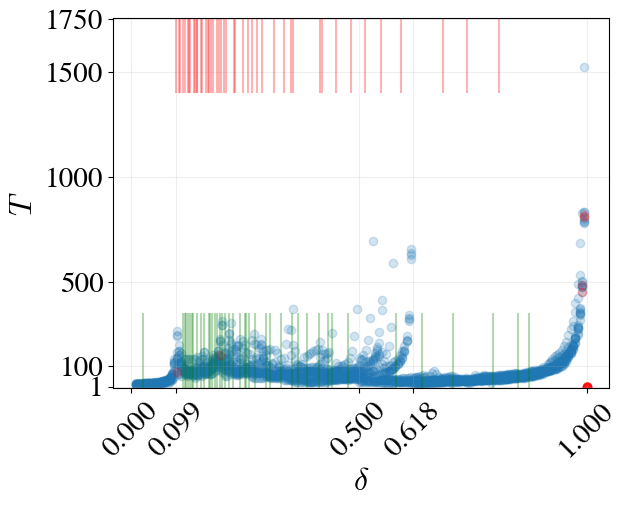}%
        \caption{$d=40$}
    \end{subfigure}%
    \begin{subfigure}[t]{0.25\textwidth}
        \centering
        \includegraphics[width=\linewidth]{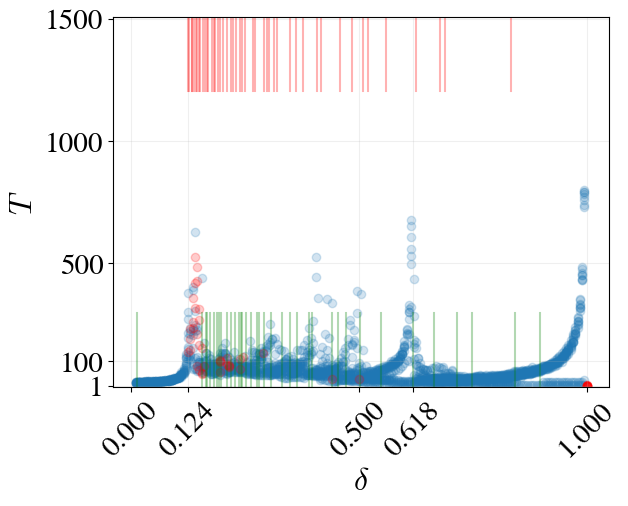}%
        \caption{$d=80$}
    \end{subfigure}\\
    \caption{Convergence plots for random regular graphs with degree $d$ and $n=100$ vertices.}
    \label{fig:rr}
\end{figure}

\begin{wrapfigure}{R}{0.3\textwidth}
    \includegraphics[width=\linewidth]{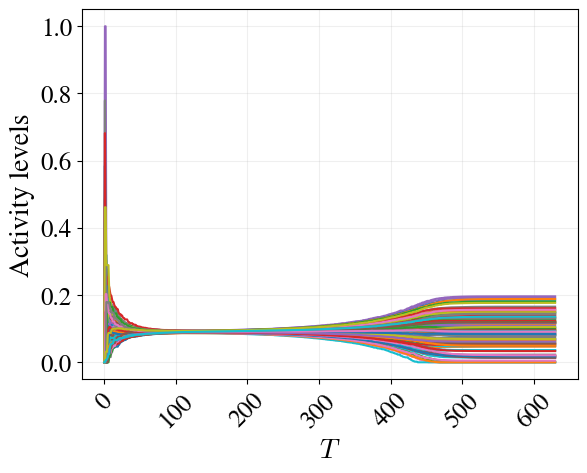}%
    \caption{Sample trajectory of the process for $\delta=0.125$ in a random regular graphs with degree $d=80$ and $n=100$ vertices.}
    \label{fig:rr_unif}
\end{wrapfigure}
The differences between the different random graph models are the most pronounced for $\delta\in (1/|\lambda_{\min}(\bm G)|, 1/\phi)$.
We begin by studying random regular graphs.\footnote{
Random regular graphs are sampled uniformly from the set of all $d-$regular graphs on $n$ vertices. A graph is $d-$regular if all vertices has degree $d$.}
In a $d-$regular graph on $n$ vertices, there exists an equilibrium with activity levels uniform and equal to $1/(1+\delta d)$.
In particular, whenever the equilibrium is unique, it is the equilibrium with uniform messages.
This equilibrium is stable for $\delta<1/|\lambda_{\min}(\bm G)|$.
The loss of its stability leads to the peaks observed near the value $\delta=1/|\lambda_{\min}(\bm G)|$.
Near this value, especially for $d=80$, we observe many trajectories which do not converge to a stable equilibrium with the assigned stopping criterion.
Figure~\ref{fig:rr_unif} includes a sample such trajectory.
The process first approaches a strategy profile with every agent choosing the same strategy.
For $\delta=0.125$, this profile is not stable, leading to slow divergence away from it and to an equilibrium in which there are inactive agents.
This is similar to the slow divergence observed for paths, where whenever an equilibrium existed but was not stable, the divergence away from it was slow.

For $\delta$ within the interval $(1/|\lambda_{\min}(\bm G)|, 1/\phi)$, the shuffle phase of the dynamics is often long compared to the asymptotic phase.
The number of distinct active sets visited throughout the dynamics can be large and it is difficult to notice patterns within the changes.
In Figure~\ref{fig:rr_samples}, we present representative trajectories.
The graph sample is the same between all of them.
For $\delta=0.35$, in Figure~\ref{fig:rr1}, the shuffle phase is long.
Throughout this phase, the activity levels change with little regularity. The individual agent's trajectory can change monotonicity multiple times or remain constant for long periods of time.
Once the agents enter the asymptotic phase, they converge at different rates.
For the same value of $\delta$, the random update order can lead to very different trajectories which become mostly regular relatively quickly (c.f. Figure~\ref{fig:rr2}).
While the asymptotic phase is only entered around the time $T=300$, when the last activity change occurs, the strategy profile is close to the stable equilibrium the system converges to around the time $T=100$ for the majority of the agents.

For $\delta=0.45$, the trajectories are more structured and fewer agents change their activity levels significantly throughout the shuffle phase (c.f. Figures~\ref{fig:rr3}, \ref{fig:rr4}). Furthermore, in the equilibrium the system converges to, fewer distinct activity levels are chosen and groups of distinct agents converge to the same activity levels.
For $\delta=0.55$, the shuffle phase is again shorter and the asymptotic convergence is slow (c.f. Figures~\ref{fig:rr5}, \ref{fig:rr6}).
The stable equilibria the dynamics converge to have fewer distinct activity levels than both of the cases $\delta=0.35, 0.45$. 
The component stability condition limits the number of distinct graphs which can make up the connected components of the active set, leading to the different connected components of the active set being isomorphic and therefore approaching the same activity levels.
The inspection of the final active sets confirms this. For $\delta=0.35$, the active sets tend to comprise a large component, multiple isolated vertices and occasional small components. For $\delta=0.45$, they include a smaller component and more isolated vertices and, for $\delta=0.55$, mostly isolated vertices and several small components. 
For different values of $d$, the analysis of individual trajectories yields similar results. The larger $\delta$ is, the simpler the trajectories and the active sets are.

This behavior is difficult to explain by analyzing the equilibria of the graph. 
The variance of terminal activity levels in the case $\delta=0.35$ suggests that the equilibria analysis is difficult.
The components of the active sets supporting stable equilibria are comprised of the induced subgraphs of the random regular graph, which can be arbitrary.
In some respects, we may view the varied behavior of the trajectories as a consequence of this. 
No clear structure of the equilibria leads to intricate interaction between the subgraphs, consequently leading to complex dynamics.

\begin{figure}
    \centering
    \begin{subfigure}[t]{0.16\textwidth}
    \centering
            \includegraphics[width=\linewidth]{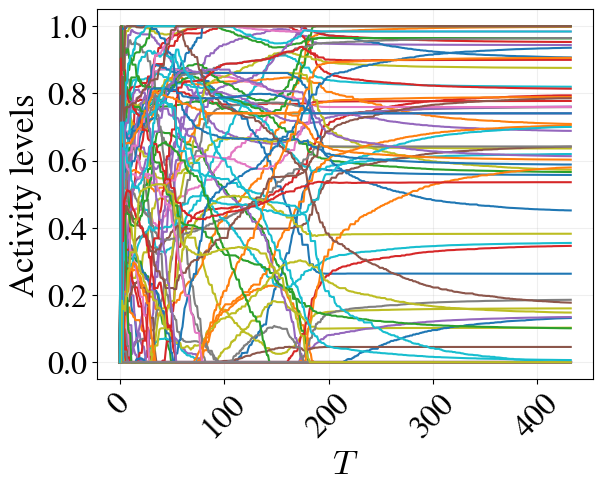}
        \caption{$\delta=0.35$}
        \label{fig:rr1}
    \end{subfigure}%
    \begin{subfigure}[t]{0.16\textwidth}
    \centering
            \includegraphics[width=\linewidth]{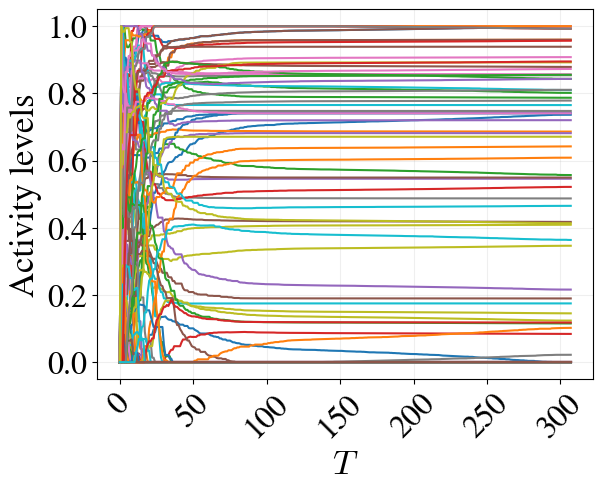}
        \caption{$\delta=0.35$}
        \label{fig:rr2}
    \end{subfigure}%
    \begin{subfigure}[t]{0.16\textwidth}
    \centering
            \includegraphics[width=\linewidth]{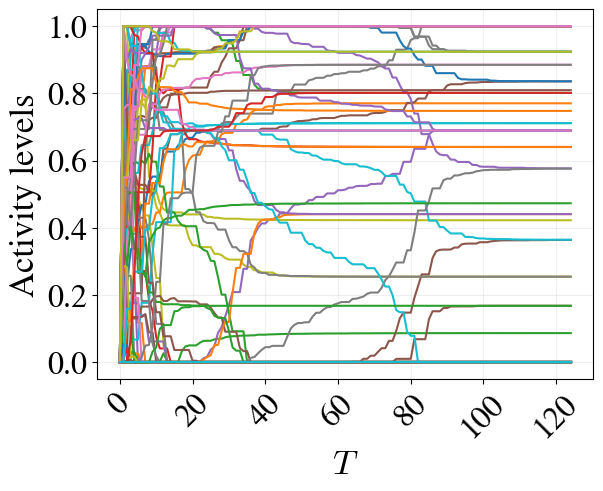}
        \caption{$\delta=0.45$}
        \label{fig:rr3}
    \end{subfigure}%
    \begin{subfigure}[t]{0.16\textwidth}
    \centering
            \includegraphics[width=\linewidth]{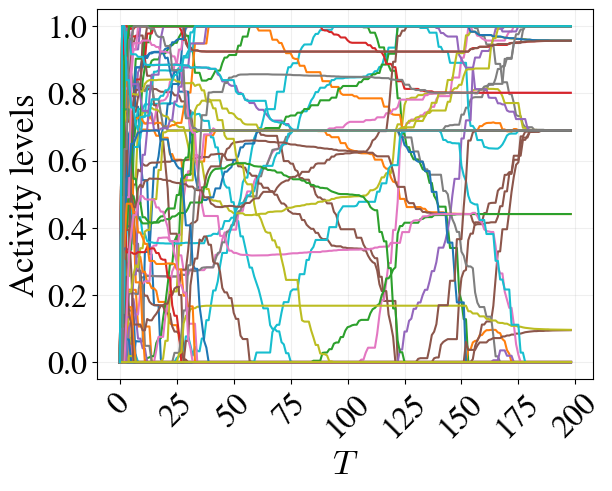}
        \caption{$\delta=0.45$}
        \label{fig:rr4}
    \end{subfigure}%
    \begin{subfigure}[t]{0.16\textwidth}
    \centering
            \includegraphics[width=\linewidth]{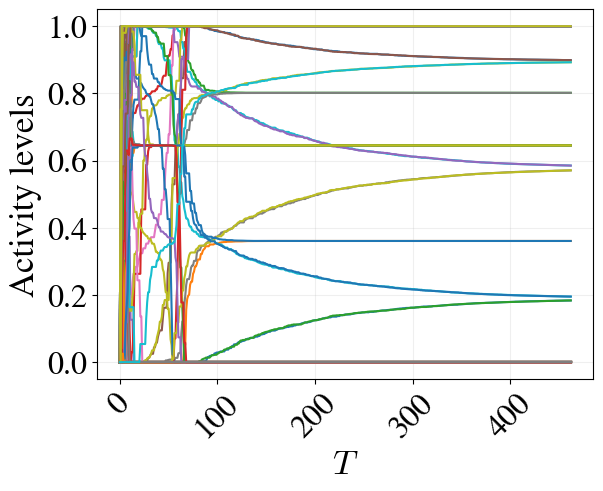}
        \caption{$\delta=0.55$}
        \label{fig:rr5}
    \end{subfigure}%
    \begin{subfigure}[t]{0.16\textwidth}
    \centering
            \includegraphics[width=\linewidth]{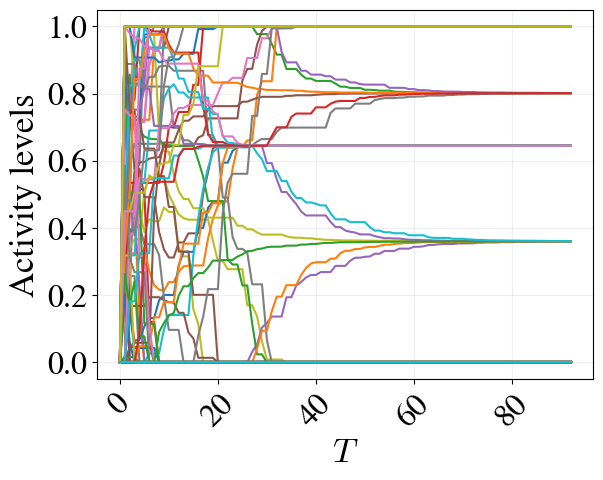}
        \caption{$\delta=0.55$}
        \label{fig:rr6}
    \end{subfigure}
    \caption{Sample trajectories for a sample random regular graph with $d=5$ and $n=100$ for various values of $\delta$.}
    \label{fig:rr_samples}
\end{figure}

\begin{figure}[h!]
    \centering
    \begin{subfigure}[t]{0.25\textwidth}
        \centering
        \includegraphics[width=\linewidth]{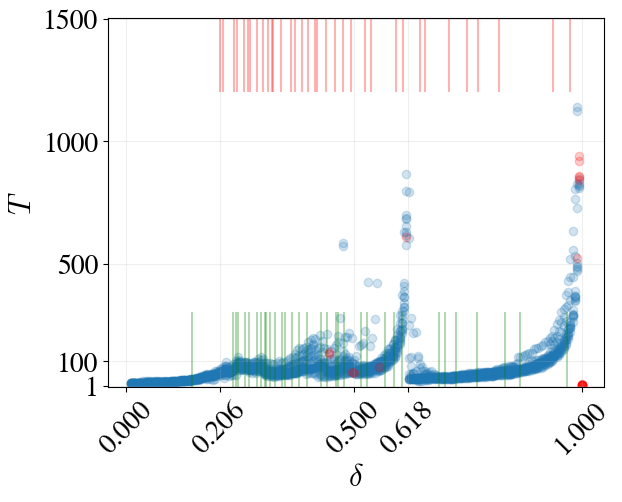}%
        \caption{$p=0.05$}
    \end{subfigure}%
    \begin{subfigure}[t]{0.25\textwidth}
        \centering
        \includegraphics[width=\linewidth]{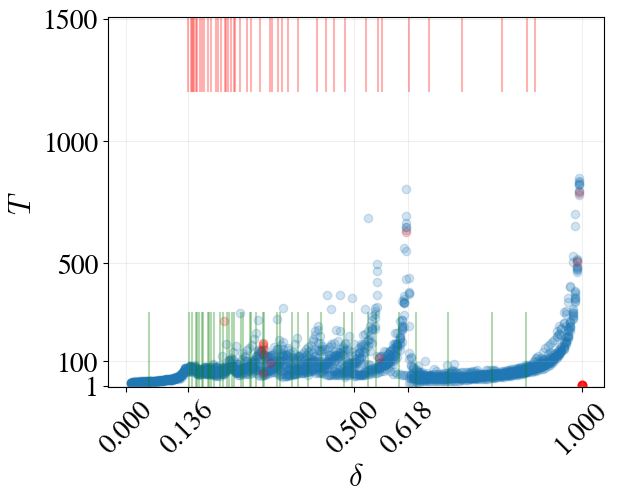}%
        \caption{$p=0.2$}
    \end{subfigure}%
    \begin{subfigure}[t]{0.25\textwidth}
        \centering
        \includegraphics[width=\linewidth]{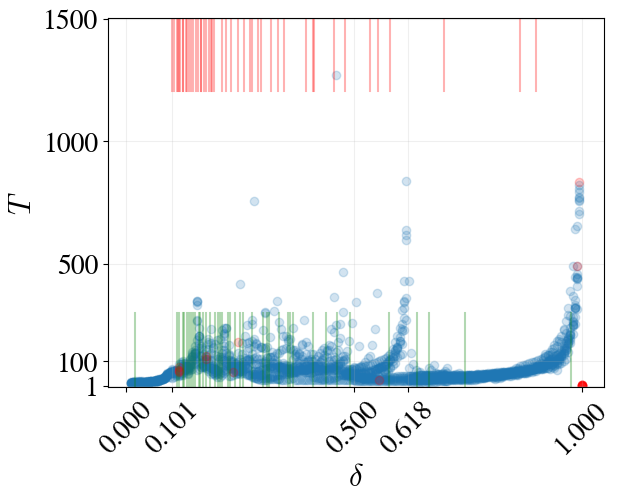}%
        \caption{$p=0.5$}
    \end{subfigure}%
    \begin{subfigure}[t]{0.25\textwidth}
        \centering
        \includegraphics[width=\linewidth]{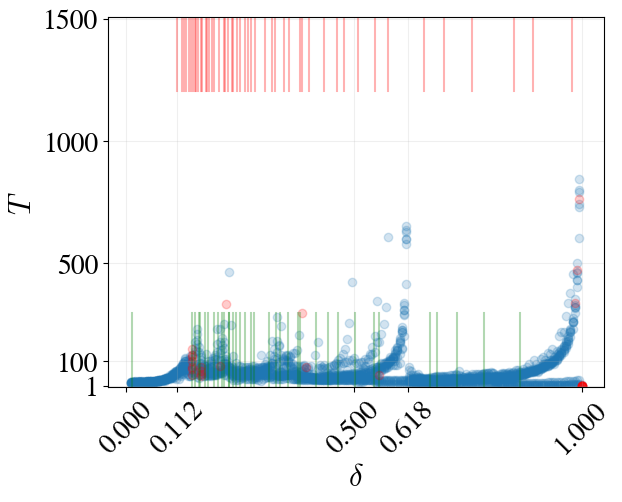}%
        \caption{$p=0.8$}
    \end{subfigure}\\
    \caption{Simulation results for Erd\H{o}s-R\'enyi random graph model samples on $n=100$ vertices with $p$ varied}
    \label{fig:er}
\end{figure}
We now proceed to the Erd\H{o}s-R\'enyi (ER) random graph model.
\footnote{In this model, the graph has $n$ vertices and $g_{ij}$ are independent random variables with $\mathbb P(g_{ij}=1)=1-\mathbb P(g_{ij}=0)=p$ for a given parameter $p\in[0,1]$.}
The overall convergence properties are similar to those on random regular graphs.
The main difference occurs near the loss of uniqueness threshold.
For $\delta<1/|\lambda_{\min}(\bm G)|$, the unique equilibrium can have inactive agents.
Consequently, no loss of stability occurs near the threshold, leading to no significant increase of convergence times.
The property is similar to what was observed for odd-length paths and complete bipartite graphs with unequal numbers of vertices in either part.

Within the range $\delta\in(1/|\lambda_{\min}(\bm G)|, 1/\phi)$, the dynamics resemble the dynamics on random regular graphs.
The individual trajectories are diverse, with regularity increasing as the value of $\delta$ approaches $1/\phi$.
Graphs sampled from both the ER and RR models have complex structures.
In particular, the connected components of active sets supporting the stable equilibria are diverse and more complex than the graphs we studied earlier.
This leads to intricate interactions between the stable equilibria and the phenomena observed earlier co-occurring. Similar behavior was observed within large paths in the interval $(0.5, 0.532)$, where the proximity of thresholds of different length blocks lead to complex behavior. In this case, the blocks are replaced with arbitrary graphs and the behavior is even more complex.
\begin{figure}[htb]
    \centering
    \begin{subfigure}[t]{0.25\textwidth}
        \centering
        \includegraphics[width=\linewidth]{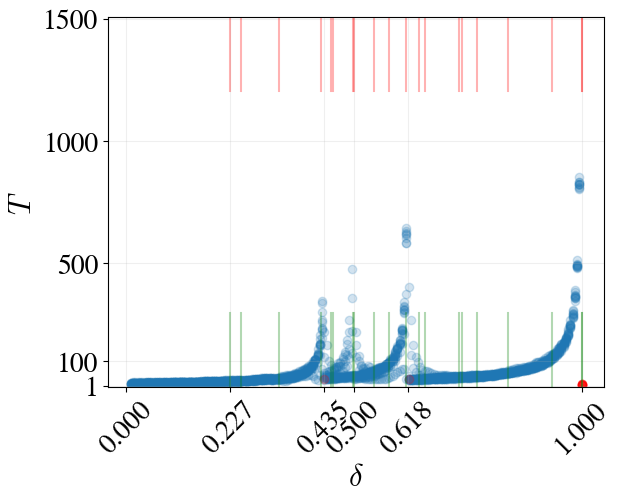}%
        \caption{$n=100, d=1$}
    \end{subfigure}%
    \begin{subfigure}[t]{0.25\textwidth}
        \centering
        \includegraphics[width=\linewidth]{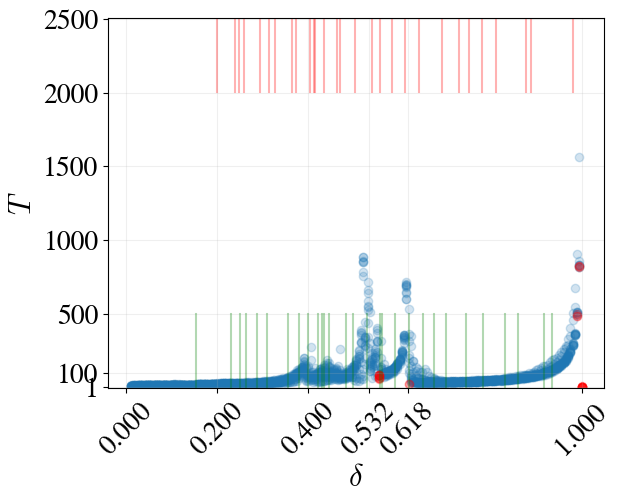}%
        \caption{$n=100, d=2$}
    \end{subfigure}%
    \begin{subfigure}[t]{0.25\textwidth}
        \centering
        \includegraphics[width=\linewidth]{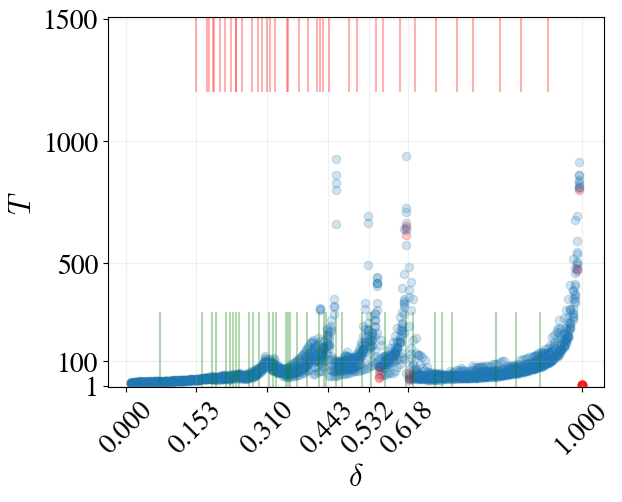}%
        \caption{$n=100, d=5$}
    \end{subfigure}%
    \begin{subfigure}[t]{0.25\textwidth}
        \centering
        \includegraphics[width=\linewidth]{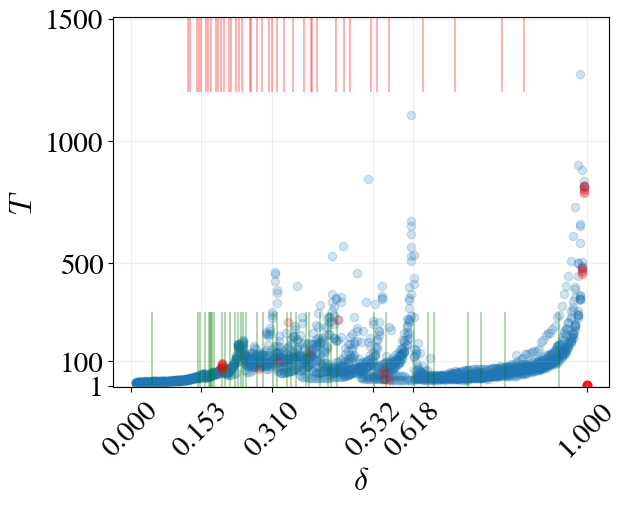}%
        \caption{$n=100, d=10$}
    \end{subfigure}
    \caption{The results for BA random graph samples with $n$ vertices and parameter $m$}
    \label{fig:ba}
\end{figure}

For Barab\'{a}si-Albert (BA) random graphs\footnote{BA random graphs are generated procedurally. They start with a clique of size $m$. Then, a node is added and connected to $m$ existing vertices, which is done by adding a single edge $m$ times. When adding each edge, the probability of connecting the new node to an existing node $i$ not connected to it is proportional to the number of neighbors of $i$.}, we observe fewer trajectories with long convergence times than in other random graphs and the overall plot appears more structured (c.f.Figure~\ref{fig:ba}).
By investigating the dynamics in closer detail, we uncover the induced subgraphs responsible for the slowdowns.
For $m=1$, the peaks occur near the values $\delta=0.434, 0.5, 0.618, 1$.
The latter two values correspond to the loss of stability thresholds of the paths on $4$ and $2$ vertices.
The value $\delta=0.434$ is close to the instability threshold of the graph $G_1$ presented in Figure~\ref{fig:ba_peak1_graph}, which is equal to $1/|\lambda_{\min}(\bm G_1)| \approx 0.4343$.
The trajectory, restricted to the graph $G_1$, converges to a strategy profile in which all agents are active (c.f. Figure~\ref{fig:ba_peak1_traj}).
The behavior near the threshold $\delta=0.5$ is showcased in Figure~\ref{fig:ba_peak2}.
For $\delta=0.51$, the stable equilibrium of the graph shown in Figure~\ref{fig:ba_peak2_graph} is comprised of an inactive vertex in the middle and three pairs of agents, each with activity $1/(1+\delta)$.
This equilibrium becomes emerges at $\delta=1/2$, as only then the three pairs provide sufficient activity for the middle agent to be inactive.
This threshold is also the loss of stability threshold for the equilibrium with all agents active. The divergence from this strategy profile is slow, leading to long overall convergence time.
\begin{figure}[!htb]
    \centering
    \begin{minipage}{.49\textwidth}
    \centering
    \begin{subfigure}[t]{0.5\textwidth}
    \includegraphics[width=1\linewidth]{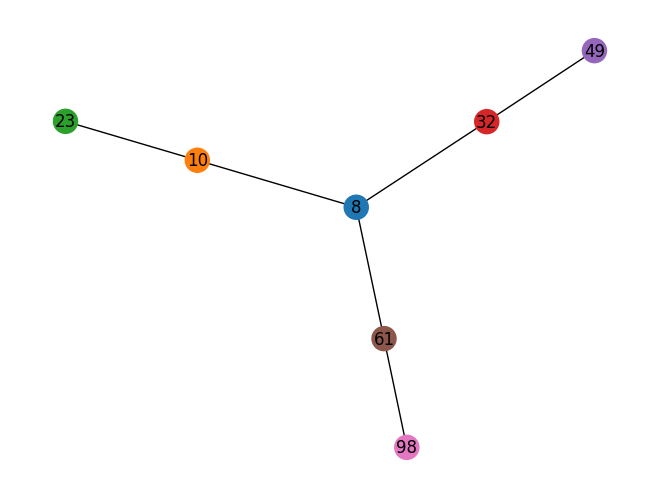}%
    \caption{The structure of the subgraph}
    \label{fig:ba_peak2_graph}
    \end{subfigure}%
    \hfill
    \begin{subfigure}[t]{0.5\textwidth}
    \includegraphics[width=1\linewidth]{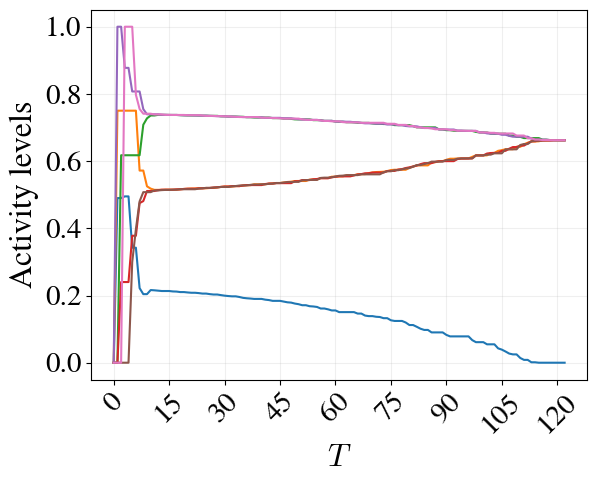}%
    \caption{The trajectories of agents withing the subgraph}
    \label{fig:ba_peak2_traj}
    \end{subfigure}
    \caption{The induced subgraph of a sample BA graph with $n=100, m=1$ causing slow convergence at $\delta=0.51$.}
    \label{fig:ba_peak2}
    \end{minipage}%
    \hfill
    \begin{minipage}{0.49\textwidth}
    \centering
    \begin{subfigure}[t]{0.5\textwidth}
    \includegraphics[width=1\linewidth]{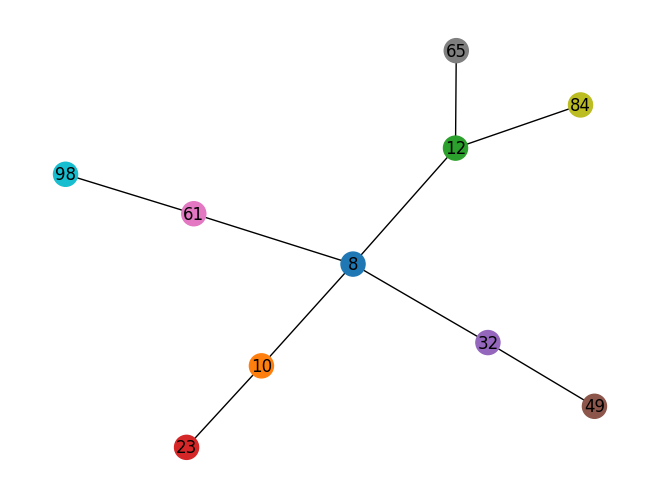}%
    \caption{The structure of the subgraph}
    \label{fig:ba_peak1_graph}
    \end{subfigure}%
    \hfill
    \begin{subfigure}[t]{0.5\textwidth}
    \includegraphics[width=1\linewidth]{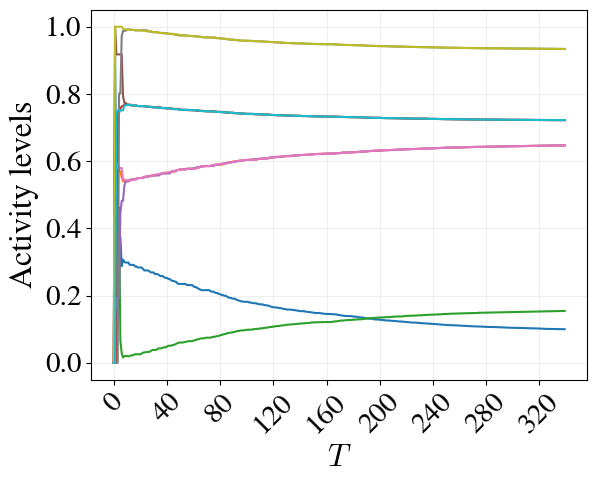}%
    \caption{The trajectories of agents withing the subgraph}
    \label{fig:ba_peak1_traj}
    \end{subfigure}
    \caption{The induced subgraph of a sample BA graph with $n=100, m=1$ causing slow convergence at $\delta=0.434$.}
    \label{fig:ba_peak1}
    \end{minipage}
\end{figure}

For $m=2$ and $m=5$ we observe more peaks form. 
They indicate slower convergence times and are more prominent than what was observed for other classes of random graphs.
When inspecting individual trajectories and slow converging subgraphs, similarities with the case $m=1$ arise, with star-like patterns slowing down the overall convergence.

For $m=10$, while the convergence time plot includes numerous peaks, they are less defined.
The convergence time plots resemble those of RR and ER random graphs more than those of BA random graphs for smaller values of $m$.
The individual trajectories confirm this, with the shuffle phase showing more disordered behavior, similar to the trajectories observed for RR graphs.\footnote{Appendix~\ref{appendix:random_trajs} contains more sample trajectories and further comparison of the random graphs.}
While some of the cases of slow convergence may be traced to subgraphs characteristic of BA graphs, the overall structure is more complex and the intricate connections between different parts of the graph dominate.

\subsubsection{Graphs showcasing reshuffles}
\label{sec:reshuffles}
The ability to explain slow convergence times relies heavily on the ability to describe the possible equilibria.
The more structured a graph is, the simpler it is to describe the stable equilibria and therefore explain the dynamics.
For graphs with less structure, many distinct phenomena can occur within individual trajectories and coherent explanations are difficult to provide.
We therefore move to study one of the discovered phenomena, the reshuffle, by itself, in more detail.
It is perhaps more interesting than the other phenomena, as it can lead to significant changes in activity levels when the system seems to have already converged.

\begin{figure}[h]
\centering
\hfill
\begin{minipage}{.4\textwidth}
\centering
    \includegraphics[width=\textwidth]{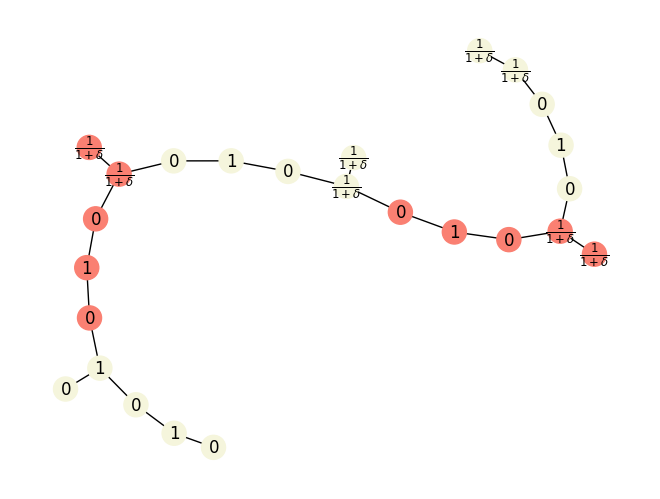}
    \caption{The graph leading to $5$ reshuffles, with starting activity levels annotated on the nodes. Color changes indicate different copies of the path graph.}
    \label{fig:switchups_graph}
\end{minipage}%
\hfill
\begin{minipage}{.4\textwidth}
\centering
\includegraphics[width=\textwidth]{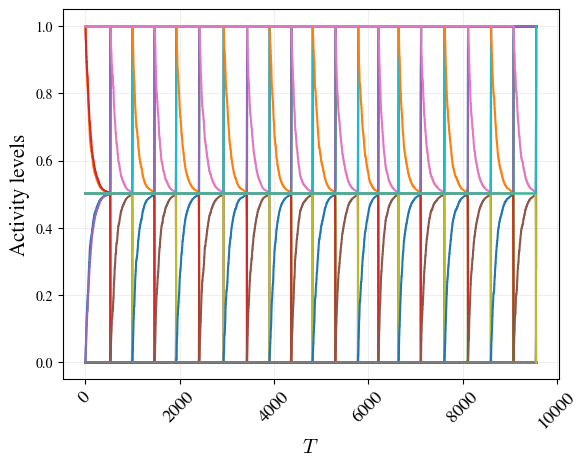}
\caption{The trajectory started at the configuration leading to $10$ reshuffles for $\delta=0.99$.}
\label{fig:switchups_traj}
\end{minipage}%
\hfill
\end{figure}

\begin{wrapfigure}{r}{0.37\textwidth}
\centering
\begin{subfigure}[t]{0.37\textwidth}
\centering
    \includegraphics[width=0.7\textwidth]{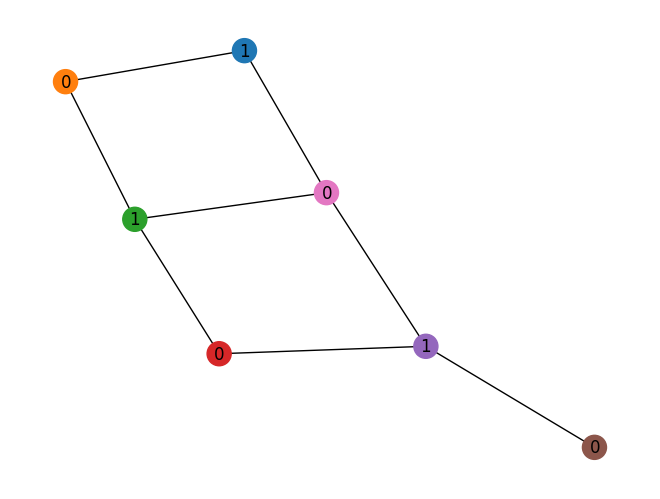}
    \caption{The graph leading to a reshuffle only involving one component.}
    \label{fig:single_comp_reshuffle_graph}
\end{subfigure}\\
\hfill
\centering
\begin{subfigure}[t]{0.37\textwidth}
\centering
    \includegraphics[width=0.6\textwidth]{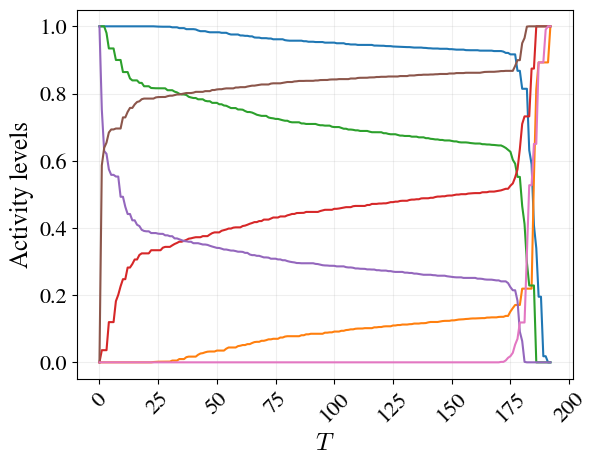}
    \caption{The trajectory started at the configuration leading to the reshuffle only involving one component for $\delta=0.55$.}
    \label{fig:single_comp_reshuffle_traj}
\end{subfigure}
\caption{An example of a system with a reshuffle involving only one connected component.}
\label{fig:single_comp_reshuffle}
\end{wrapfigure}

We construct two graph examples which showcase more properties of the reshuffle.
First we show that arbitrarily many reshuffles can occur.
To do this we construct a graph and an initial strategy profile which leads to a trajectory with $k$ reshuffles, for any $k\in \mathbb N$.
A detailed version of this construction is included in Appendix~\ref{appendix:many_switchups}.
In Figure~\ref{fig:switchups_graph} we showcase the graph and the initial activity levels.

The graph is constructed by connecting many copies of the path of length $n=5$  together such that one of the copies reaching its equilibrium leads to a reshuffle in the next copy. In particular, when the path with a starting strategy profile of $(0,1,0,1,0)$ reaches its stable equilibrium $(1,0,1,0,1)$, the first agent in the next path does not have his neighbor activity threshold met and initiates the next reshuffle.
The constructed graph can be comprised of arbitrarily many paths, leading to an arbitrary number of reshuffles.
For $10$ reshuffles, this leads to the trajectory presented in Figure~\ref{fig:switchups_traj}.

By taking many copies of such graph, we can achieve arbitrarily large expected convergence times for $\delta$ close enough to $1$.
\begin{thm}
\label{pr:many_reshuffles}
    For any $\delta\in(1/\phi,1)$ and $T\in \mathbb N$, there exists a graph $G$ such that the expected time to obtain the final active set is at least $T$.
\end{thm}
Note that this implies that, while we do understand the equilibria for $\delta\in(1/\phi,1)$, the convergence can still be slow due to a large number of reshuffles.

The reshuffles observed in the previous examples happened when an inactive agent separating two parts of the graph became active.
Now, we will show that the inactive agent may only be connected to a single active component of the graph and still cause a reshuffle, changing the equilibrium the system converges to significantly.
In Figure~\ref{fig:single_comp_reshuffle_graph} we present a graph leading to this behavior.
It consists of a path of length $n=6$ with a single vertex added and connected to vertices with odd indices. In this example, we have $\delta=0.55$. The trajectory starts at a strategy profile where the agents in the path with odd indices have activity level of $1$ and the remaining agents have an activity level of $0$ (c.f. Figure~\ref{fig:single_comp_reshuffle_graph}).
With no added agent, the dynamics would converge to the equilibrium with all agents active, as the threshold for the loss of stability is $\delta\approx0.555$. 
During those dynamics, the activity levels of agents with odd indices keep decreasing, eventually leading to the added vertex activating due to insufficient activity levels as showcased in Figure~\ref{fig:single_comp_reshuffle_traj}.
Finally, the system reaches an equilibrium in which $x_i^*=1-x_i(0)$ for all agents $i$.

\subsection{Web application}
The characteristics of the dynamics often depend on particular induced subgraphs, with different subgraphs being important for different values of $\delta$.
The analysis conducted in this section showed that the trajectories need to be analyzed separately to understand the system.
To facilitate the inspection of different trajectories, we created a web application along with interactive tools meant to help analyze them.
This resource, along with its documentation, is available at \href{https://networkconvergencerates.up.railway.app/}{https://networkconvergencerates.up.railway.app/}.

\section{Summary of the results and discussion}
\label{sec:summary}

By investigating many sample trajectories, we uncovered much of the complex behavior of the system, identifying phenomena governing convergence rates.
Throughout the entire analysis, we observed the following phenomena to be the causes of slow convergence.
\begin{enumerate}[label=(\alph*)]
    \item When the dynamics converges to a stable equilibrium which is close to its stability threshold, the convergence is slow.\label{p:summary:1}
    \item When the dynamics diverges from an unstable equilibrium close to its stability threshold, the divergence is slow, leading to long overall convergence times.\label{p:summary:2}
    \item The dynamics may converge to a strategy profile which is a stable equilibrium of the active set but provides insufficient activity levels for inactive agents and is therefore not an equilibrium for the entire network. When the dynamics is close to converging, the inactive agents become active and initiate convergence to a different strategy profile. We call this phenomena the \emph{reshuffle} and show that there can be a sequence of reshuffles of arbitrary length.\label{p:summary:3}
    \item The dynamics may undergo many changes in a short span of time due to intricate relationships between different parts of graph. This behavior is especially prominent in unstructured networks.\label{p:summary:4}
\end{enumerate}
With Theorem~\ref{thm:equilibria}, we classify the stable equilibria in terms of connected components of the active sets. 
This classification leads to identifying the observed phenomena with the structural properties of the network.
In particular, the asymptotic converge occurs at a rate proportional to $1+\delta\lambda_{\min}(\bm G_{S})$. Hence, for $\delta$ close to $1/|\lambda_{\min}(\bm G_{S})|$ the dynamics are slow.
If $\delta<1/|\lambda_{\min}(\bm G_{S})|$, the trajectories slowly converge to the equilibrium, and if $\delta>1/|\lambda_{\min}(\bm G_{S})|$, they diverge slowly away from it. This corresponds to the phenomena \ref{p:summary:1} and \ref{p:summary:2} listed above.
For graphs in which the equilibria can be identified, such as the path graphs, this leads to the enumeration of the potential values of $\delta$ for which the time to converge is large.

Failure of Condition~\ref{p:eqilibria:3} of Theorem~\ref{thm:equilibria} leads to more interesting behavior.
For a given subset of agents $S$, the equilibria of the connected components of $S$ may provide insufficient activity levels for the inactive agents to remain inactive.
In the dynamics, this leads to inactive agents changing their activity levels and initiating convergence to a different strategy profile, supported by a different active set.
This change of active set can happen quickly and involve different subgraphs, leading to the intricate behavior observed for random graphs and captured by point \ref{p:summary:4}.
It can also happen slowly, leading to a reshuffle described by point \ref{p:summary:3}.
Reshuffles occur whenever the strategy profile, restricted to the set of active agents, is a stable equilibrium with $\delta$ close to its stability threshold. 
The trajectory then slowly converges to the strategy profile, possibly fulfilling the neighbor activity level requirements of the inactive agents until the system is close to converging. Once those requirements are not satisfied, the inactive agents enter the dynamics and the system begins convergence to a different strategy profile. 
This behavior shows that, even after a trajectory reaches an $\epsilon-$equilibrium, the agents may still change their strategies significantly.
Furthermore, such reshuffles can lead to other reshuffles, repeating similar behavior multiple times.

\bibliographystyle{ACM-Reference-Format}
\bibliography{InvDynCit1}

\newpage
\appendix
\section{Proofs and comments}
\subsection{Stability}
\label{appendix:theoretical}
\begin{autoproof}{th:stability}

We begin by presenting the problem in the language of dynamical systems.
Let $\bm x(t)$ be the state of the system at step $t$.
We choose $i_{t+1}$ uniformly from the set $\{1, \dots, n\}$ and update the state with the map
\begin{align*}
    \bm x(t+1) = M_{i_{t+1}}(\bm x(t)) = \bm x(t)+\bm e_{i_{t+1}}(-x_{i_{t}}(t)+\max(0, 1-(\delta\bm A\bm x(t))_{i_t})),
\end{align*}
where $\bm e_i$ are the standard basis vectors. 
Each of the maps $M_i$ increases the game's potential\footnote{The potential game given by this potential has the same best responses as the abstract game we study here (c.f. \cite{Bra14}). Therefore, this potential increases with each best-reply strategy change. } given by 
$$V(\bm x) = \bm x^\top\bm 1-\frac12\bm x^\top(\bm I + \delta\bm G)\bm x.$$

First we show Lyapunov stability for $\delta<1/|\lambda_{\min}(\bm G)|$.
From \cite[Proof of Proposition 4]{Bra07} we know that $\bm x^*$ is a strict maximum of $V$.
The notion of a strict maximum is local and therefore if agents change their strategies by large amounts at a time, they may find a larger value of $V$.
We show that this is not a case by showing that their activity changes are small enough near the equilibrium.

Let $d=\inf_i |1-\sum g_{ij}x_j^*|$. As agents are active or strictly inactive, $d > 0$. Let $D$ denote the maximal degree in the graph.
Observe that if $\|\bm x(t)- \bm x^*\|_\infty<d/D$,for any $i\not \in S$, the best move is to change to $0$, so move by at most $\|\bm x(t)- \bm x^*\|_\infty$.
For $i\in S$, the best move differs from $x^*_i$ by at most $|(1-\sum g_{ij}x_j(t)) - (1-\sum g_{ij}x_j(t))|<D\|\bm x(t)- \bm x^*\|_\infty$.
Therefore every agent changes their strategy by at most $(D+1)\|\bm x(t)- \bm x^*\|_\infty$.

Now, let $R<d/D$ be such that $\bm x^*$ is the unique maximum of $V$ on $B(\bm x^*, R)=\{\bm x: \|\bm x(t)- \bm x^*\|_\infty<R\}$.
Let $V_0$ be the maximal potential achieved on $B(\bm x^*, R)\setminus B(\bm x^*, R/(2+D))$. As $\bm x^*$ is the unique maximum, we must have $V_0<V(\bm x^*)$. Let $V_1=(V_0+V(\bm x^*))/2$ and let $r\in(0, $ be such that the minimum value of $V$ on $B(\bm x^*, r)$ is some $V_2\geq V_1$ (possible by continuity).
Now we know that trajectory started in $\bm x(0)\in B(\bm x^*, r)$ will not leave $B(\bm x^*, R/(2+D))$, as all steps of size up to $(D+1)\|\bm x(t)- \bm x^*\|_\infty$ with $\bm x(t)\in B(\bm x^*, R/(2+D))$ would land in $B(\bm x^*, r)$ and the subset of $B(\bm x^*, R)$ in which the potential is at least $V_2$ is contained within $B(\bm x^*, R/(2+D))$.

We therefore obtain 
$$\forall R<d/D\exists r: \|\bm x(0)- \bm x^*\|_\infty<r\implies\forall t>0\|\bm x(t)- \bm x^*\|_\infty<R.$$
If the given $R$ is larger than $d/2D$, we can choose $r$ corresponding to $d/2D$ to obtain Lyapunov stability.

Now, we can proceed to check attractiveness.
Suppose we start with $\|\bm x(0)-\bm x^*\|<\epsilon$, with $\epsilon$ chosen such that $\|\bm x(t)-\bm x^*\|<d/D$ throughout the dynamics.
For all $j\not\in S$, if $i_k$ is drawn to be $j$, agents $j$ changes their strategy to $0$. Hence, 
$$\mathbb P(\exists t_0\forall j\not \in S\forall t>t_0: x_j(t)=0)=1.$$
We condition on this event and start the dynamics at $t_0$, assuming only players in $i$ are active (and by the choice of $\epsilon$, they remain active).

We track the evolution of the error term $\bm f(t) = \bm x(t)-\bm x^*$.
Note that the transition maps are given by 
\begin{align*}
    M_i(\bm x)=\begin{cases}
        \bm x \text{ if } i\not\in S,\\
        \bm x+\bm e_i(-x_i+1-(\delta\bm A\bm x)_{i})\text{ if } i\in S.
    \end{cases}
\end{align*}
By letting $\bm B = \bm I+\delta \bm G_{S}$ and noting that
\begin{align*}    
x_i+(\delta\bm A\bm x)_{i}  &=\bm e_i^\top (\bm I+ \delta\bm A)\bm x= \bm e_i^\top \bm B \bm x,
\end{align*}
this rearranges to 
\begin{align*}
    M_i(\bm x)=\begin{cases}
        \bm x \text{ if } i\not\in S,\\
        \bm e_i+(\bm I-\bm e_i\bm e_i^\top \bm B)x\text{ if } i\in S.
    \end{cases}
\end{align*}

Combining this with $M_i(\bm x^*)=\bm x^*$, we obtain, for $i\in S$
\begin{align*}
    M_i(\bm x)-\bm x^*&=M_i(\bm x)-M_i(\bm x^*)\\
    &=(\bm I-\bm e_i\bm e_i^\top \bm B)(\bm x-\bm x^*)\\
    &=(\bm x-\bm x^*) - \bm e_i(\bm B(\bm x-\bm x^*))_i.
\end{align*}

Consider now the quantity $\|\bm f\|_{\bm B}^2 = \bm f^\top \bm B \bm f.$
As $\delta\lambda_{\min}(\bm G_{S})>-1$, $B$ is positive definite and $\|\bm f\|_{\bm B}$ is a norm.

If $i_t\in S$, noting that $B$ is symmetric, we have
\begin{align*}
    \|f(t+1)\|_B^2&=(\bm x(t+1)-\bm x^*)^\top \bm B(\bm x(t+1)-\bm x^*)\\
    &=(M_{i_{t+1}}(\bm x(t))-\bm x^*)^\top \bm B(M_{i_{t+1}}(\bm x(t))-\bm x^*)\\
    &=(\bm f-\bm e_{i_{t+1}}(\bm B\bm f)_{i_{t+1}})^\top \bm B (\bm f-\bm e_{i_{t+1}}(\bm B\bm f)_{i_{t+1}})\\
    &=\|\bm f\|_B^2-(\bm B\bm f)_{i_{t+1}}\bm e_{i_{t+1}}^\top \bm B\bm f-\bm f^\top \bm B\bm e_{i_{t+1}}(\bm B\bm f)_{i_{t+1}}+(\bm B\bm f)_{i_{t+1}}\bm e_{i_{t+1}}^\top\bm Be_{i_{t+1}}(\bm B\bm f)_{i_{t+1}}\\
    &=\|\bm f\|_B^2-(\bm B\bm f)_{i_{t+1}}^2,
\end{align*}
where we used that $\bm e_i^\top \bm B \bm e_i=1$ and the shorthand notation $\bm f =\bm f(t)$.
Hence, the norm only decreases. This is clear, as $\|\bm f\|_{\bm B}$ may be viewed as the local potential, but we estimated the size of this decrease.

Now we can prove attractivity.
We first do it simply.
Namely, we can bound the series $\sum_{t=0}(\bm B\bm f)_{i_t}^2\leq \|f(0)\|_B^2$ implying $(\bm B\bm f)_{i_t}\to0$ and so $(\bm B\bm f)_{i}\to 0$ almost surely.
This argument gives us no insights about convergence speed.
To obtain them, we show stability by proving that $(1-\frac{\lambda_{\min}(\bm B)}n)^{-t}\|\bm f(t)\|_B^2$ is a supermartingale (with the filtration $\mathcal F_t$ generated by the process $\bm x(t)$) and use standard martingale theory.
Note that 
\begin{align*}
    \mathbb E(\|f(t+1)\|_B^2|\mathcal F_n)&=\mathbb E(\|\bm f(t+1)\|_B^2|\bm f(t)=\bm f)\\
    &=\|\bm f\|_B^2-\mathbb E((\bm B\bm f)_{i_{t+1}}^2|\bm f(t)=\bm f)\\
    &=\|\bm f\|_B^2-\frac1n\sum_{i=1}^n(\bm B\bm f)_i^2\\
    &=\|\bm f\|_B^2-\frac{\|Bf\|_2^2}n\\
    &\leq\|\bm f\|_B^2-\frac{\lambda_{\min}(\bm B)\|\bm f\|_B^2}n\\
    &=(1-\frac{\lambda_{\min}(\bm B)}n)\|\bm f\|_B^2,
\end{align*}
where to obtain the inequality, we used the Rayleigh Quotient 
\begin{align*}
    \lambda_{\min}(\bm B)=\min_{x\not=0}\frac{\bm x^\top \bm B\bm x}{\bm x^\top\bm x}=\min_{x\not=0}\frac{(\sqrt{\bm B}\bm x)^\top \bm B\sqrt{\bm B}\bm x}{(\sqrt{\bm B}\bm x)^\top\sqrt{\bm B}\bm x}=\min_{x\not=0}\frac{\|\bm B \bm x\|_2^2}{\|\bm x\|_{\bm B}^2}.
\end{align*}
A simple calculation shows that 
\begin{align*}
    \rho=1-\frac{\lambda_{\min}(\bm B)}n= \frac{n-1-\delta \lambda_{\min}(\bm G_S)]}{n}<1.
\end{align*}
We can now deduce from Doob's martingale convergence theorem that $f(t)/\rho^t$ converges almost surely to a random variable with finite expectation.
This implies $f_t$ converging almost surely and, heuristically, we expected the size of the error term to be $O(\rho^t)$ (with the constant random).

We now need to show the lack of stability in the remaining cases.
Suppose now $\delta\lambda_{\min}(\bm G_{S})<-1$.
This implies that $\bm B$ has an eigenvalue $\lambda=1+\delta\lambda_{\min}(\bm G_{S})<0$. Suppose $\bm v$ satisfies $\bm B\bm v=\lambda \bm v$ and $\|\bm v\|_\infty<1$.

Consider now the function $V(\bm f) = \bm f^\top\bm B\bm f$. 
It no longer induces a norm but the previous calculations hold and it is nonincreasing.
\begin{align*}
    V(\bm f(t+1))&=V(\bm f)-(\bm B\bm f)_{i_{t+1}}^2.
\end{align*}
Furthermore, it should converge to $0$ whenever $\bm f(t)\to 0$.
We construct a point $\bm x$ which is arbitrarily close to $\bm x^*$ and $V(\bm x-\bm x^*)<0$.
Indeed, for any $\epsilon'<\epsilon$, we can take $\bm x^*+\epsilon' \bm v$ and obtain 
\begin{align*}
    V(\bm x-\bm x^*)&=V(\epsilon'\bm v)=\lambda\epsilon'^2\|\bm v\|_2^2<0.
\end{align*}
This means the process cannot converge to $\bm x^*$.

For $\delta\lambda_{\min}(\bm G_{S})=-1$, we also consider the strategy profile $\bm x= \bm x^*+\epsilon' \bm v$.
We have that \begin{align*}
    (\bm I +\delta\bm A)\bm x = (\bm I +\delta\bm A)\bm x^* + \bm0= \bm 1,
\end{align*}
so $\bm x$ is an equilibrium on the active set. By non-degeneracy and the choice of $\epsilon'$, the inactive agents remain inactive and so $\bm x$ is an equilibrium.
The dynamics are therefore constant and $\bm x^*$ cannot be attractive.

\end{autoproof}

This result only holds for equilibria with agents active or strictly inactive.
For almost any $\delta$, only such equilibria exist \cite[Footnote~16]{Bra14}.

\subsection{Path equilibria}
\begin{autoproof}{pr:path:uneq}
    Firstly, note that for $\delta<1/2$ the result is a trivial consequence of the observation regarding the relationship between $\delta$ and the maximal degree of the graph.
    We begin by proving the result for $n=2m+1$.
    For $\delta\geq1/2$, the strategy profile $x_i\equiv i\bmod 2$ is an equilibrium. As the equilibrium is unique, this is the only equilibrium.

    For $n=2m$, we split the analysis into the cases $\delta=1/2$ and $\delta>1/2$ and explicitly construct the equilibrium for each.
    Assume $\delta>1/2$ and let $\bm x=(\bm I+\delta \bm G)^{-1}\bm 1.$
    We need $x_i>0$ to show the claim, as then it must be an equilibrium.
    The entries of $\bm x$ satisfy
    \begin{align*}
        x_i+\delta(x_{i-1}+x_{i+1})&=1,\\
        x_0=x_{2m+1}&=0.
    \end{align*}
    Set $-1/\delta = 2\cos\theta$ with $\theta\in(\frac{2m}{2m+1}\pi, \pi]$.
    The solutions to this problem can be written in the form
    \begin{align*}
        x_i=A\cos\left(\left(i-\frac{2m+1}2\right)\theta\right)+B\sin \left(\left(i-\frac{2m+1}2\right)\theta\right)+\frac1{1+2\delta}.
    \end{align*}
    Using the boundary conditions, we obtain
    \begin{align*}
        A&=-\frac1{(1+2\delta)\cos\left((\frac{2m+1}2)\theta)\right)},\\
        B&=0,
    \end{align*}
    and the solution
    \begin{align*}
        x_i&=\frac1{(1+2\delta)}\left(1-\frac{\cos\left((i-\frac{2m+1}2)\theta\right)}{\cos(\frac{2m+1}2\theta)}\right)\\
        &=\frac1{(1+2\delta)}\left(\frac{\cos(\frac{2m+1}2\theta)-\cos\left((i-\frac{2m+1}2)\theta\right)}{\cos(\frac{2m+1}2\theta)}\right)\\
        &=\frac{-2\sin(\frac{i\theta}2)\sin\left((\frac{2m+1-i}2)\theta\right)}{(1+2\delta)\cos(\frac{2m+1}2\theta)}.
    \end{align*}
    We need to analyze the sign.
    Firstly, we always have $1+2\delta>0$. Set $\theta=\pi-\epsilon$. As $\theta\in(\frac{2m}{2m+1}\pi)$, we have $\epsilon\in (0,\frac{1}{2m+1}\pi)$.
    We now have $\frac{2m+1}2\theta=\frac{2m+1}2\pi-\frac{2m+1}2\epsilon\in(m\pi, \frac{2m+1}2\pi)$.
    This means that $\sgn(\cos(\frac{2m+1}2\theta))=(-1)^m$.
    Now we compute the sines.
    Note $\frac{i\theta}2=i\frac{\pi}{2}-i\epsilon/2$. As $i\in \{1,\dots, 2m\}$, we have $0< i\epsilon/2<\frac{m}{2m+1}\pi<\frac{\pi}{2}$ and therefore $\sgn(\sin(\frac{i\theta}2)=(-1)^{\lfloor\frac{i-1}{2}\rfloor}$. 
    This also yields $\sgn(\sin((\frac{2m+1-i}2)\theta)=(-1)^{\lfloor\frac{2m-i}2\rfloor}$.
    Combining these, we obtain
    \begin{align*}
        \sgn (x_k)=-(-1)^{\lfloor\frac{i-1}{2}\rfloor}(-1)^{\lfloor\frac{2m-i}2\rfloor}(-1)^m=(-1)^{\lfloor\frac{i-1}{2}\rfloor+\lfloor\frac{2m-i}2\rfloor+m+1}.
    \end{align*}
    Simplyfying, we get $$\left\lfloor\frac{i-1}{2}\right\rfloor+\left\lfloor\frac{2m-i}2\right\rfloor+m+1=\left\lfloor\frac{i+1}{2}\right\rfloor-1+\left\lceil\frac{i}2\right\rceil+2m+1=2m+2,$$ and so $x_i>0$.
    
    Finally, we need to handle the case $\delta=1/2$.
    We obtain the solution $x_i=\frac12+A(-1)^i+Bi(-1)^i=\frac{1}{2}+(-1)^i(\frac{i}{2m+1}-\frac12)$ which can be easily seen to be positive.

\end{autoproof}


\begin{autoproof}{pr:sequences}
    First, note that $b(\delta, 2m)$ is well defined for $\delta<1/|\lambda_{\min}(\bm G(2m))$ by Proposition~\ref{pr:path:uneq}. For all $\delta$, we have $b(\delta, 1)=1$.
    Furthermore, using the proof of Proposition~\ref{pr:path:uneq} we obtain the value of $b(\delta, 2m)$. It uses the substitution $\delta = -1/2\cos\theta$ with $\theta\in(\frac{2m}{2m+1}\pi, \frac{2m+2}{2m+3}\pi)$ corresponding to the interval $\delta\in (1/|\lambda_{\min}\bm G(2m+2)|, 1/|\lambda_{\min}(\bm G(2m))|)$.
    We have
    \begin{align*}
        b(\delta, 2m)&=\frac{-2\sin(\frac{\theta}2)\sin(m\theta)}{(1+2\delta)\cos(\frac{2m+1}2\theta)}=\frac{-2\sin(\frac{\theta}2)\cos\theta\sin(m\theta)}{(\cos\theta-1)\cos(\frac{2m+1}2\theta)}.
    \end{align*}
    To show statement~\ref{p:seq:1}, we therefore need
    \begin{align*}
    -2\cos\theta&<1+\frac{ -2\sin(\frac{\theta}2)\cos\theta\sin(m\theta)}{(\cos\theta-1)\cos(\frac{2m+1}2\theta)}\\
0&<1+2\cos\theta+\frac{-2\sin(\frac{\theta}2)\cos\theta\sin(m\theta)}{-2\sin^2(\frac\theta2)\cos(\frac{2m+1}2\theta)}\\
    0&<\frac{\sin(\frac{3\theta}{2})}{\sin(\frac{\theta}{2})}+\frac{\cos\theta\sin(m\theta)}{\sin(\frac\theta2)\cos(\frac{2m+1}2\theta)}\\
    0&<\frac{\sin(\frac{3\theta}{2})\cos(\frac{2m+1}2\theta)+\cos\theta\sin(m\theta)}{\sin(\frac\theta2)\cos(\frac{2m+1}2\theta)}\\
    0&<\frac{\sin((m+2)\theta)+\sin((1-m)\theta)+\sin((m+1)\theta)+\sin((m-1)\theta)}{2\sin(\frac\theta2)\cos(\frac{2m+1}2\theta)}\\
    0&<\frac{\sin((m+2)\theta)+\sin((m+1)\theta)}{2\sin(\frac\theta2)\cos(\frac{2m+1}2\theta)}\\
    0&<\frac{\sin(\frac{2m+3}{2}\theta)\cos(\frac{\theta}{2})}{\sin(\frac\theta2)\cos(\frac{2m+1}2\theta)}.
    \end{align*}
    In the transformations, we used standard trigonometric identities.

    For $\theta\in(\frac{2m}{2m+1}\pi, \frac{2m+2}{2m+3}\pi)$, we have $\sin(\frac\theta2), \sin(\frac\theta2)>0$.
    Furthermore, we also have
    \begin{align*}
        \frac{2m+1}{2}\theta \in(m\pi, m\pi+\frac{1}{2m+3}\pi)&\implies\sgn(\cos(\frac{2m+1}2\theta))=(-1)^m,\\
        \frac{2m+3}{2}\theta \in((m+1)\pi+\frac{-1}{2m+1}\pi, (m+1)\pi)&\implies\sgn(\sin(\frac{2m+3}{2}\theta))=(-1)^m.
    \end{align*}
    This yields the inequality.

    Now, we proceed Statement~\ref{p:seq:2}. A similar chain of calculations leads to the required inequality
    \begin{align*}
    0&>\frac{\sin(\frac{2k+3}{2}\theta)\cos(\frac{\theta}{2})}{\sin(\frac\theta2)\cos(\frac{2k+1}2\theta)}.
    \end{align*}
    We need to show that $\sin(\frac{2k+3}2\theta),\cos(\frac{2k+1}2\theta)$ have different signs. We verify that
    \begin{align*}
        \frac{2k+1}{2}\theta \in((k+\frac{m-k}{2m+1})\pi,(k+\frac{m-k+1}{2m+3})\pi)&\implies\sgn(\cos(\frac{2k+1}2\theta))=(-1)^k,\\
        \frac{2k+3}{2}\theta \in((k+1+\frac{m-k-1}{2m+1})\pi, (k+1+\frac{m-k}{2m+3})\pi)&\implies\sgn(\sin(\frac{2k+3}{2}\theta))=(-1)^{k+1},
    \end{align*}
    proving the result.

    Finally, we prove Statement~\ref{p:seq:3}. We begin by showing $2\delta b(\delta, 2m)<1$. This condition rearranges to
    \begin{align*}
        \tan(\frac{2m+1}{2}\theta)\cot(\frac\theta2)>0,
    \end{align*}
    which, due to identical considerations as earlier, is satisfied.

    It now remains to compute the limit, with $\theta\to\frac{2m}{2m+1}\pi$, of
    \begin{align*}
        \delta b(\delta, 2m)&=\frac{-1}{2\cos\theta}\frac{\cos\theta\sin(m\theta)}{\sin(\frac{\theta}2)\cos(\frac{2m+1}2\theta)}=\frac{-\sin(m\theta)}{\sin(\frac{\theta}2)\cos(\frac{2m+1}2\theta)}.
    \end{align*}
    We have 
    \begin{align*}
        \cos\left(\frac{2m+1}2\theta\right)&\to \cos(m\pi)=(-1)^m,\\
        \sin(m\theta)&\to \sin\left(\left(m-\frac{m}{2m+1}\right)\pi\right)=(-1)^{m+1}\sin\left(\frac{m}{2m+1}\pi\right),\\
        \sin(\theta)&\to \sin\left(\frac{m}{2m+1}\pi\right),
    \end{align*}
    and so $2\delta b(\delta, 2m)\to 1$, finishing the proof.
    
\end{autoproof}
\subsection{Classification of stable equilibria}
\label{appendix:allowed_graphs}

\begin{autoproof}{thm:equilibria}
    Condition~\ref{p:eqilibria:1} implies that the equilibrium indeed exists and Conditions~~\ref{p:eqilibria:2} and~\ref{p:eqilibria:3} imply that all agents in $S$ are active and agents not in $S$ strictly inactive. 
    As $G_{S_i}$ are the connected components of $G_S$, we obtain $\lambda_{\min}(\bm G_S)=\min_i\lambda_{\min}(\bm G_{S_i})$ by using \cite[Proposition~1.3.6]{brouwer2012spectra}.
    Condition~\ref{p:eqilibria:1}, along with Theorem~\ref{th:stability}, now implies that the obtained equilibrium is stable.
   
    On the contrary, if Condition~\ref{p:eqilibria:1} is not met but the equilibrium with the given active set exists it is either unstable or has agents which are inactive but not strictly inactive by Theorem~\ref{th:stability}.
    The latter can only happen for finitely many $\delta$ \cite[Footnote~16]{Bra14}.

    Here, we can also note that if an equilibrium exists and it does not satisfy Condition~\ref{p:eqilibria:1} but does satisfy Condition~\ref{p:eqilibria:2} and Condition~\ref{p:eqilibria:3}, it only has strictly inactive agents and is therefore unstable, regardless of $\delta$.
\end{autoproof}

\begin{autoproof}{thm:paths}
    We use Theorem~\ref{thm:equilibria} to deduce the result for paths.
    Firstly, we assume the conditions are met and show that the configuration string induces a stable equilibrium. This holds for all values of $\delta$.
    We set $$S=\left\{\sum_{j=1}^m(a_j+1)+i: 1\leq i\leq a_{m+1}, 1\leq m\leq k\right\}.$$
    Using condition ~\ref{p:paths:4}, this is spans the entire path, meaning that every inactive vertex has exactly two active neighbors.

    The connected components induced by $S$ are exactly the active paths of lengths $a_i$. 
    We can now see that Condition~\ref{p:paths:1} in the assumptions implies Condition~\ref{p:eqilibria:1} of Theorem~\ref{thm:equilibria}.
    Condition~\ref{p:eqilibria:2} implies that the unique equilibria on the paths of lengths $a_i$ have all active agents, with all activity levels strictly positive, implying Condition~\ref{p:eqilibria:2} of Theorem~\ref{thm:equilibria}.
    Finally, Condition~\ref{p:paths:3} implies Condition~\ref{p:eqilibria:3} of Theorem~\ref{thm:equilibria}.
    The constructed active set -- which is described by the given configuration -- supports a stable equilibrium, finishing the proof.

    If the configuration describes a stable equilibrium, each of the components must be stable and have positive strategy profiles, implying Conditions~\ref{p:paths:1} and~\ref{p:paths:2}. 
     is only implied for almost any $\delta$, as 
    For almost any $\delta$, all equilibria have only active or strictly inactive agents \cite[Footnote~16]{Bra14}. Therefore, Condition~\ref{p:paths:3} is satisfied for almost any $\delta$, as it is implied by no non-strictly inactive agents between the blocks.
    
    If Condition~\ref{p:paths:4} did not hold, the configuration would not be well defined.
\end{autoproof}
\subsection{The reshuffle is possible}
\begin{autoproof}{thm:p5}
    Suppose the first two strategy changes occur at positions $2$ and $4$ in any order.
    This has probability of  $2/25$.
    We can now construct the fastest possible update order and show that it is slow enough to obtain the result.
    We need to find the smallest number of best-response updates such that the inactive agent in the middle becomes active again.
    For this to occur, the convergence of the boundary pairs must progress far enough.

    It makes no sense to update an agent twice as the update will be idempotent, hence we need to update agents on the side in an alternating order. 
    It is also ineffective to only update agent $1$ or $5$ without updating $2$ or $4$ respectively.
    This specifies the update order up to deciding how many updates are made on each side.

    If $2t$ steps are assigned to a side, we obtain the endpoint value of
    \begin{align*}
        \frac{1-\delta^{2t}}{1+\delta}+\delta^{2t}=\frac{1+\delta^{2t+1}}{1+\delta},
    \end{align*}
    obtained by inspecting pair dynamics with $a_0=b_0=0$.

    Hence, the activity levels with $2t$ steps assigned to one side and $2k$ to the other is, up to an affine transformation, $\delta^{2t}+\delta^{2k}$. 
    With the constraint of $s+t=const$, this is minimised when $k=t$.

    We need to achieve the inequality
    \begin{align*}
        2\delta\frac{1+\delta^{2t+1}}{1+\delta}<1,
    \end{align*}
    which rearranges to 
    \begin{align*}
        t>\frac{\ln(1-\delta) -\ln2 -2\ln\delta}{2\ln \delta}.
    \end{align*}
    Combining with $t=5T-2$ (accounting for the two initial changes), we obtain
    \begin{align*}
        \mathbb P\left(T>\frac{\ln(1-\delta) -\ln2}{10\ln \delta}\right)\geq2/25.
    \end{align*}
\end{autoproof}
\subsection{Large $\delta$ equilibria}
Finding families of graphs with a given smallest eigenvalue is a classic problem in spectral graph theory \cite{cameron1976line}.
This classification can be used to restrict what the stable equilibria's active sets consist of.
Here, we are interested in graphs satisfying $\lambda_{\min}(\bm G)>-\phi$.
We will classify them in terms of forbidden induced subgraphs.

\begin{lemma}
\label{lemma:interlace}
    Any graph $G$ satisfying $\lambda_{\min}(\bm G)>-\phi=-\frac{\sqrt5+1}{2}$ is $P_4-, C_4-$ and $K_{1,3}-$free.\footnote{For two graphs $H, G$, we say that $G$ is $H-$free if it does not contain $H$ as an induced subgraph. For $H$ to be an induced subgraph of $G$, there must exist an injection $f:V(H)\to V(G)$ preserving both the relationships of being adjacent and being non-adjacent.}\footnote{The graph $P_k$ is the path on $k$ vertices, the graph $C_k$ is the cycle on $k$ vertices, obtained by taking a path on $k$ vertices and adding the edge $(1,n)$ and $K_{l,m}$ is the complete bipartite graph with parts of sizes $l$ and $m$.}
\end{lemma}
\begin{proof}
    If $H$ is an induced subgraph of $G$, then $\bm H $ is a principal submatrix of $\bm G$. 
    By the Cauchy Interlacing Theorem, we have $\lambda_{\min}(\bm G)\leq \lambda_{\min}(\bm H)$.

    We have $\lambda_{\min}(\bm P_4) = -\phi, \lambda_{\min}(\bm C_4)=-2$ and $\lambda_{\min}(\bm K_{1,3}) = -\sqrt3$, and so none of these graphs may be an induced subgraph of $G$
\end{proof}

\begin{defn}
    A clique extension of $G$ is a graph obtained by replacing vertices of $G$ with cliques and adding, for each edge of $G$, all possible edges between the clliques replacing its endpoints.
\end{defn}
\begin{lemma}
\label{lemma:structure}
    Any $P_4-, C_4-$ and $K_{1,3}-$free graph is a disjoint union of clique extensions of $P_1$ or $P_3$.
\end{lemma}
\begin{proof}
    We may assume that the graph is connected by considering the claim separately for each connected component.
    We may also assume $|V(G)|\geq4$ as all smaller graphs satisfy the theorem.
    If the graph is $P_3-$free, it is a clique and therefore a clique extension of $P_1$.
    We now assume it is not $P_3-$free.

    Let $v_1, v_2, v_3$ be the vertices of the graph $P_3$ viewed as an induced subgraph of $G$, in order.
    If $G\setminus\{v_2\}$ has at least $3$ connected components, $G$ must have the induced subgraph $K_{1,3}$, hence it has at most $2$.

    If $v_1, v_3$ lay in distinct components of $G\setminus\{v_2\}$, and one of them, suppose $v_1$, has a neighbor $v_4\not=v_2$, we must have $v_4\sim v_2$ as otherwise $G$ would have an induced subgraph isomorphic to $P_4$. Iterating this argument, all neighbors of $v_1$ need to be connected to $v_2$ and consequently all vertices in this connected component must also be connected to $v_2$. Suppose now there exists a vertex in the component of $v_1$ which is not a neighbor of $v_1$, call it $v_0$. This makes $v_0,v_1,v_2,v_3$ a copy of $K_{1,3}$ with the center in $v_2$.
    Hence, the component of $v_1$ is a clique. By repeating the reasoning for the other component, we obtain that the graph is a clique extension of $P_3$.

    Now assume $v_1$ and $v_3$ lay in the same component. If this is not the only connected component of $G\setminus\{v_2\}$, we may repeat the previous reasoning and obtain a clique extension of $P_3$ again.
    Hence, we assume it is.
    Suppose now there is a vertex $u_1$ neighboring none of the vertices $v_1, v_2, v_3$. 
    The shortest path from $v_2$ to $u_1$ must have exactly $3$ vertices (otherwise it has $2$ vertices, meaning $u_1$ and $v_2$ are connected, or $m\geq 4$ vertices, implying $G$ includes a $P_4$ as an induced subgraph). 
    Denote the middle vertex by $u_2$. It must be adjecent to $v_1$ and $v_3$ as otherwise we could extend the path between $u_1$ and $v_2$ to $P_4$. 
    But now $\{u_2, v_1, v_3, u_1\}$ form a copy of $K_{1,3}$, which is not allowed.
    This means that each vertex $u_1$ must be connected to at least one vertex among $v_1, v_2$ and $v_3$.

    If it is connected only to $v_2$, we obtain $K_{1,3}$ as an induced subgraph of $G$.
    If it is connected only to $v_1$ or only to $v_3$, we obtain $P_4$.
    If it is connected to $v_1, v_3$ but not $v_2$, we obtain $C_4$.
    Hence, it must be connected to $v_2$ and at least one of $v_1, v_3$.
    Let us denote the set of vertices connected to $v_2$ and $v_1$ as $V_1$, to $v_2$ and $v_3$ as $V_3$ and those connected to all three as $V_2$.
    We now show that this is a clique extension of $P_3$.
    Indeed, any two vertices within $V_1$ (resp. $V_3$) must be neighbors, as otherwise they would make up a $K_{1,3}$ with $v_2$ and $ v_3$ (resp. $v_2$ and $v_1$).
    Any two vertices within $V_2$ must be connected, as otherwise they would make a copy of $C_4$ with $v_1$ and $ v_3$.
    Each vertex $v_1'\not=v_1$ from $V_1$ must be connected to each vertex from $V_2$, as otherwise the not connected vertices would make a copy of $P_4$ with $v_1$ and $ v_3$ (as $v_1$ is connected to $V_2$ due to how $V_2$ was constructed).
    Finally, for any $v_1'\in V_1\setminus\{v_1\}, v_3'\in V_3\setminus\{v_3\}$, the vertices $v_1', v_3'$ cannot be connected as otherwise they would form a path $P_4$ together with $v_1, v_3$.
    This exhausts all possibilities and finishes the proof.
    
\end{proof}
Not all clique extensions of $P_3$ satisfy the eigenvalue condition and so possibly more graphs could be excluded.
However, none of them support valid equilibria for large enough values of $\delta$ and so further restriction is not necessary.
\begin{lemma}
\label{lemma:nop3active}
    Let $G$ be a clique extension of $P_3$. For $\delta\in(1/2, 1)$, there is no equilibrium with all agents active.
\end{lemma}
\begin{proof}
    Suppose the opposite.
    Let the clique sizes be $n_1, n_2, n_3\geq1$ and let $x_1, x_2, x_3\in [0,1]$ be the activity levels. By symmetry, the activity level must be constant for all vertices belonging to the same clique as they have the same neighborhoods.
    We must have
    \begin{align*}
    \begin{cases}
        x_1+\delta\left[(n_1-1)x_1+n_2x_2\right]= 1,\\
        x_2+\delta\left[(n_2-1)x_2+n_1x_1+n_3x_3\right]=1,\\
        x_3+\delta\left[(n_3-1)x_3+n_2x_2\right]=1.\\
    \end{cases}
    \end{align*}
    Subtracting half of row 1 and half of row 3 from row 2 yields
    \begin{align*}
        (1-\delta)x_2=x_1[1/2-(n_1+1)\delta/2]+x_3[1/2-(n_3+1)\delta/2].
    \end{align*}
    Left side of the equation is positive and the right side negative, giving a contradiction.
\end{proof}

Finally, we can combine the results to fully classify the equilibria for large values of $\delta$

\begin{autoproof}{thm:large_delta_equils}
    Let $S$ be an active set of a stable equilibrium.
    Let $S_i$ be a maximal connected component of this set.
    The equilibrium must be stable when restricted to $S_i$. With this restriction it has only active agents and so by Theorem~\ref{th:stability} $\lambda_{\min}(S_i)\leq-1/\delta<-\phi$.
    By Lemma~\ref{lemma:interlace}, we know that $S_i$ is $P_4, C_4, K_{1,4}$-free and so by Lemma~\ref{lemma:structure} it must be a clique extension of $P_1$ or $P_3$.
    By Lemma~\ref{lemma:nop3active} it is not a clique extension of $P_3$, meaning it is a clique extension of $P_1$, which is simply a clique.

    For each inactive agent, the sum of activity levels in their neighborhood must be at least $1$.
    An agent in an active clique of size $k$ has activity level $(1+k\delta)^{-1}$ at equilibrium. Hence, an inactive agent $i$ with $n_k^i$ neighbors in active cliques of size $k$ is receiving exactly the external activity levels $\sum_k \frac{n^i_k}{1+(k-1)\delta}$.
    For the agent to be inactive, this quantity must be at least $1$, completing the result.
\end{autoproof}

\newpage
\section{A closer look at the convergence stages}
\label{appendix:convergence_stages}
When showcasing the dynamics plots, we showcase the total time taken by both of the phases.
If one of the phases took much longer than the other, it would make it difficult to detect when the other phase is slow.
We can choose how long the asymptotic phase is by changing the value of $\epsilon$ in the stopping criterion.
If it is too small, the asymptotic phase will be much longer than the shuffle phase and the time taken by the shuffle phase will not be easily detectable.
If we choose a value which is too large, the criterion will stop the process too soon and the asymptotic phase will not be showcased at all.
In Figure~\ref{fig:lastSchange}, we compare the time taken to establish the final active set to the total time for $\epsilon=0.0001$.
The plots showcasing total time also showcase the features of the other plots.
Consequently, this value of $\epsilon$ allows us to capture features of both the asymptotic and shuffle phases in a single plot.
The final active set was established in the majority of the cases, meaning that the value of $\epsilon$ is small enough.

\begin{figure}[h!]
    \centering
    \begin{subfigure}[t]{0.5\textwidth}
    \includegraphics[width=0.5\linewidth]{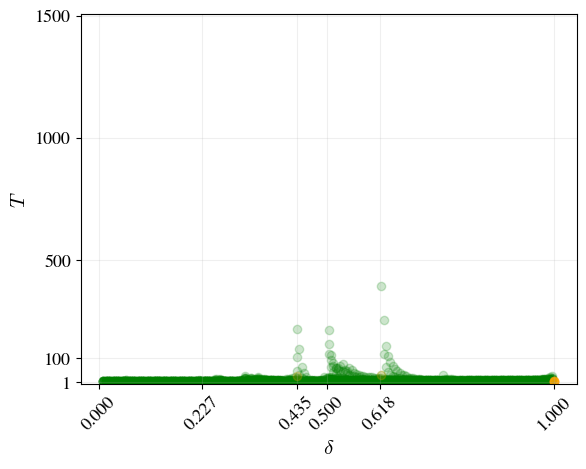}%
    \includegraphics[width=0.5\linewidth]{pictures/random_graphs/ba/ba1_100.png}
        \caption{BA(1,100) sample}
    \end{subfigure}%
    \begin{subfigure}[t]{0.5\textwidth}
    \includegraphics[width=0.5\linewidth]{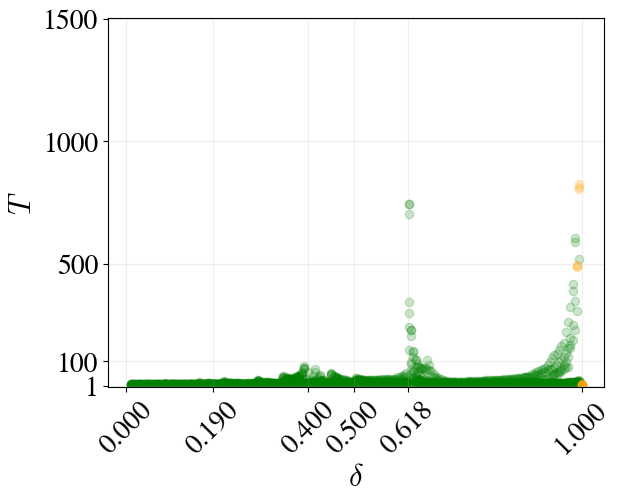}%
    \includegraphics[width=0.5\linewidth]{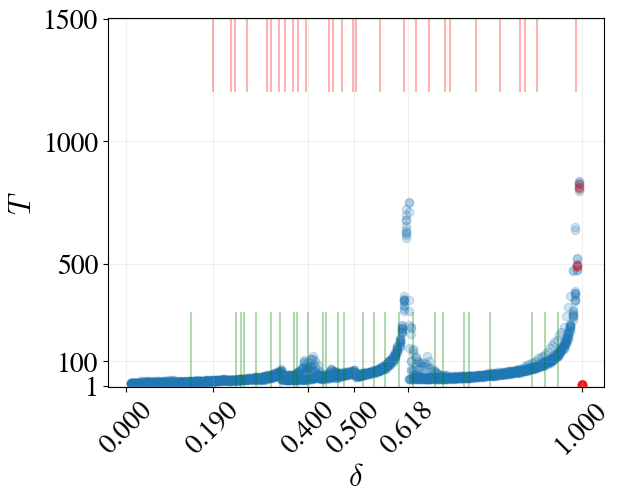}
        \caption{BA(2,100) sample}
    \end{subfigure}\\
    \begin{subfigure}[t]{0.5\textwidth}
    \includegraphics[width=0.5\linewidth]{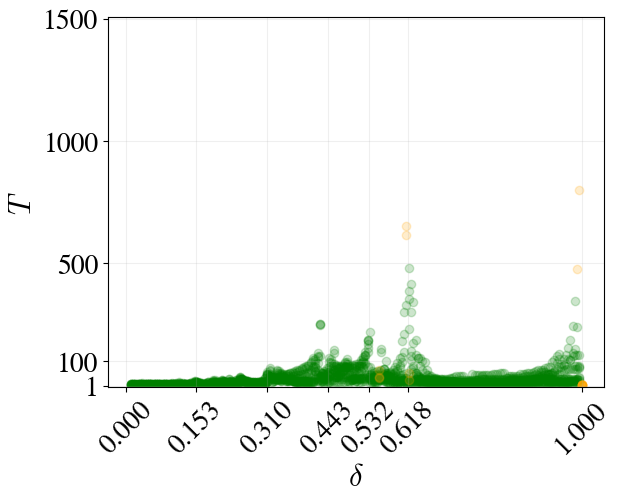}%
    \includegraphics[width=0.5\linewidth]{pictures/random_graphs/ba/ba5_100.png}
        \caption{BA(5,100) sample}
    \end{subfigure}%
    \begin{subfigure}[t]{0.5\textwidth}
    \includegraphics[width=0.5\linewidth]{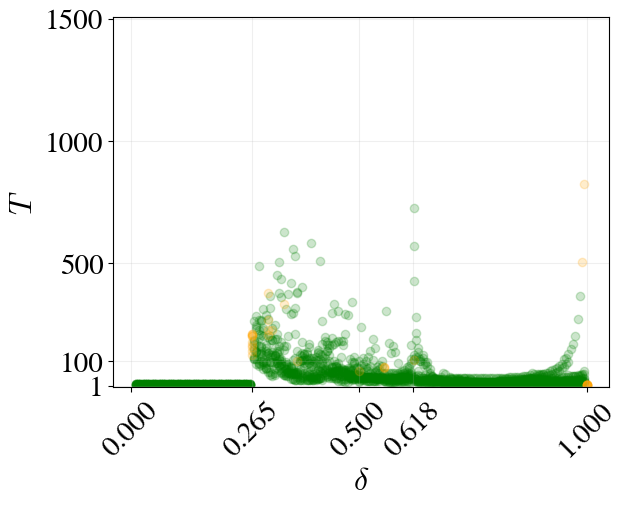}%
    \includegraphics[width=0.5\linewidth]{pictures/random_graphs/rr/rr5_100.png}
        \caption{RR(5,100) sample}
    \end{subfigure}\\
    
    \begin{subfigure}[t]{0.5\textwidth}
    \includegraphics[width=0.5\linewidth]{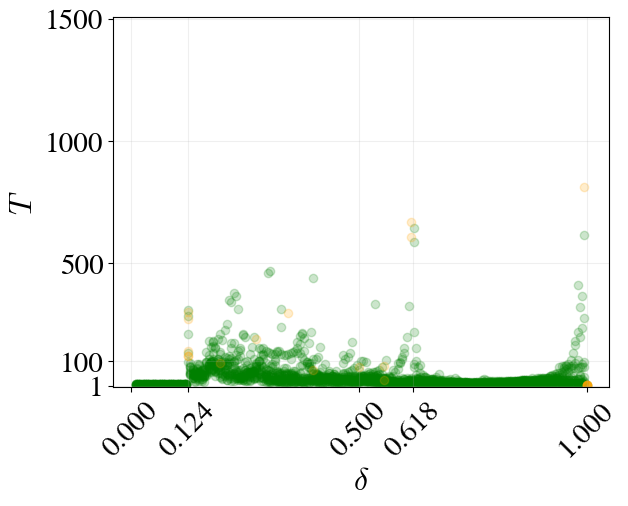}%
    \includegraphics[width=0.5\linewidth]{pictures/random_graphs/rr/rr20_100.png}
        \caption{RR(20,100) sample}
    \end{subfigure}%
    \begin{subfigure}[t]{0.5\textwidth}
    \includegraphics[width=0.5\linewidth]{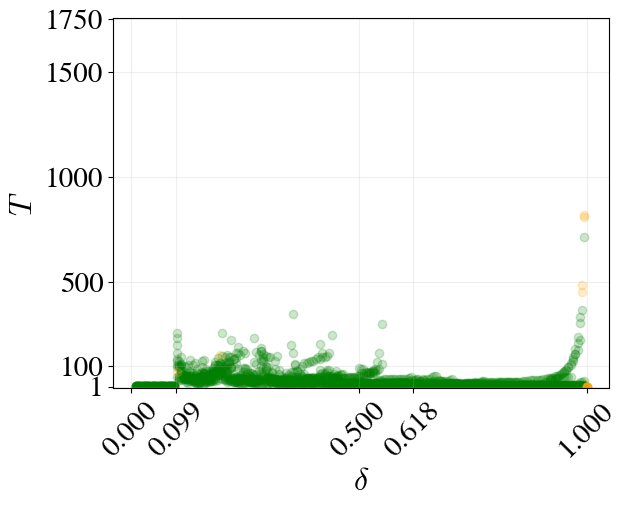}%
    \includegraphics[width=0.5\linewidth]{pictures/random_graphs/rr/rr40_100.png}
        \caption{RR(40,100) sample}
    \end{subfigure}\\

    \begin{subfigure}[t]{0.5\textwidth}
    \includegraphics[width=0.5\linewidth]{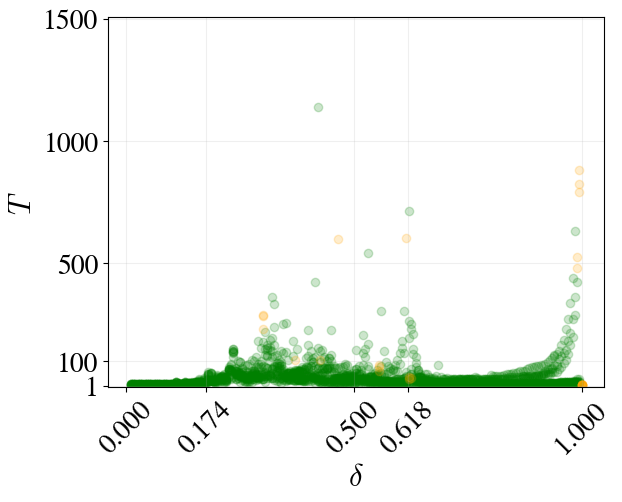}%
    \includegraphics[width=0.5\linewidth]{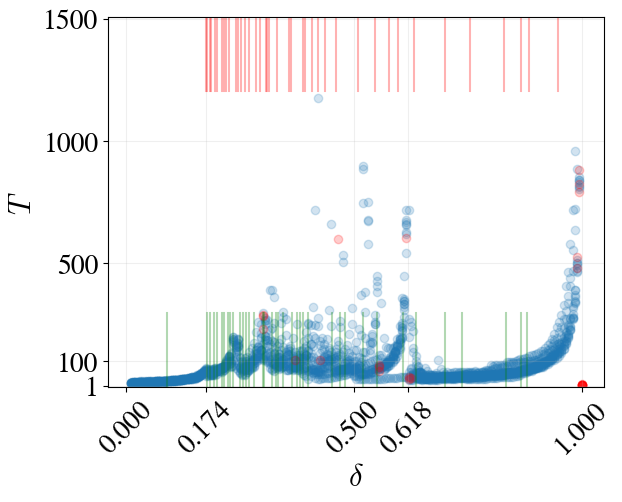}
        \caption{ER(0.1,100) sample}
    \end{subfigure}%
    \begin{subfigure}[t]{0.5\textwidth}
    \includegraphics[width=0.5\linewidth]{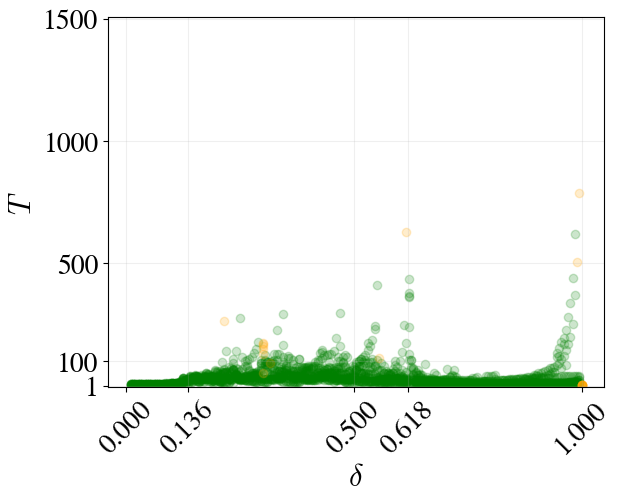}%
    \includegraphics[width=0.5\linewidth]{pictures/random_graphs/er/er02_100.png}
        \caption{ER(0.2,100) sample}
    \end{subfigure}\\

    \caption{The plots on the left side of each subfigure show the number of rounds taken until the final active set change. Green points correspond to attempts which reached an active set of a stable equilibrium, the yellow ones to those which did not. The plots on the right show the total time.}
    \label{fig:lastSchange}
\end{figure}

\newpage
\section{The impact of parity of path length in large paths.}
\label{appendix:parity}
Here, we showcase multiple trajectories for different values of $\delta$ and values of $n$ of different parity (c.f. Table~\ref{tab:parity_figures}).
The differences between the graphs of different parities exist, as graphs of different parities converge to a different equilibrium near $\delta=\frac12$. The length of the interval within which the difference is visible decreases as $n$ increases.

    \begin{table}[h!]
        \centering
        \begin{tabular}{cM{37mm}M{37mm}M{37mm}M{37mm}}
           \toprule
            $\delta$ & $n=100$ & $n=101$ & $n=300$ & $n=301$ \\
            \midrule
            $0.499$ & \includegraphics[width=0.23\textwidth]{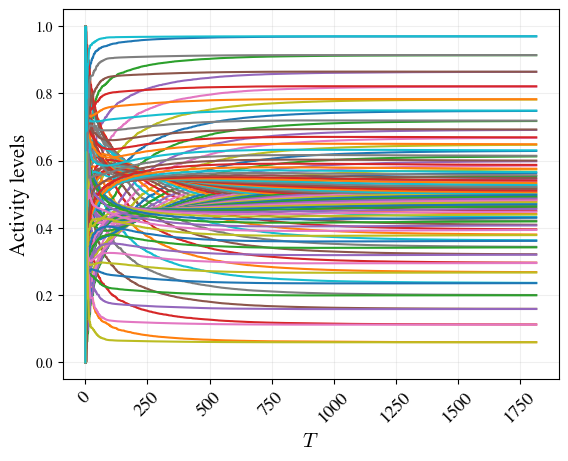} & \includegraphics[width=0.23\textwidth]{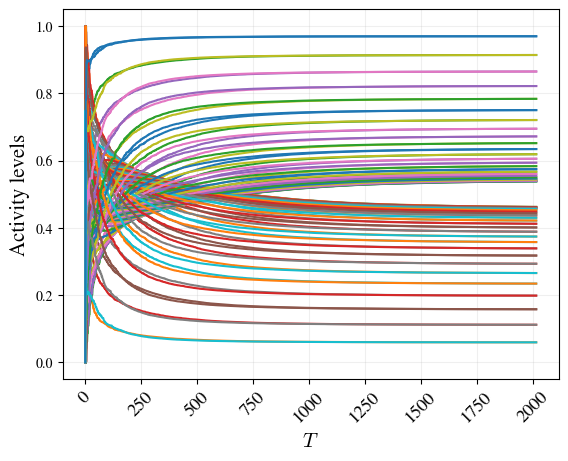} & \includegraphics[width=0.23\textwidth]{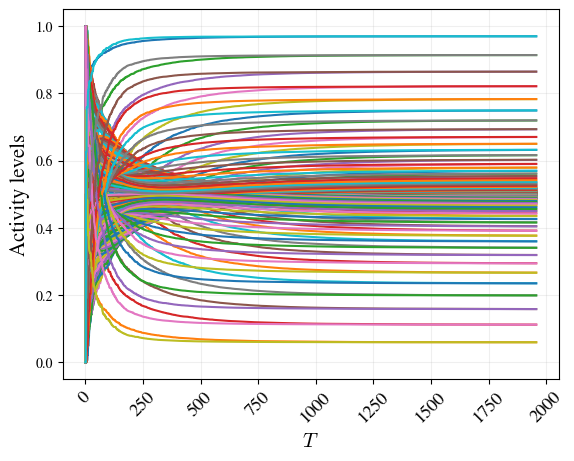} & \includegraphics[width=0.23\textwidth]{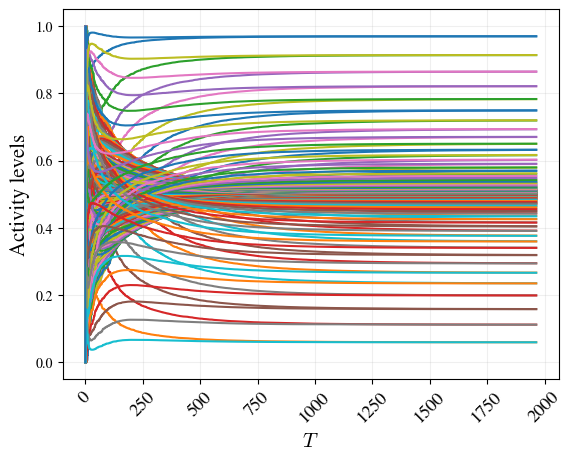} \\
            $0.4999$ & \includegraphics[width=0.23\textwidth]{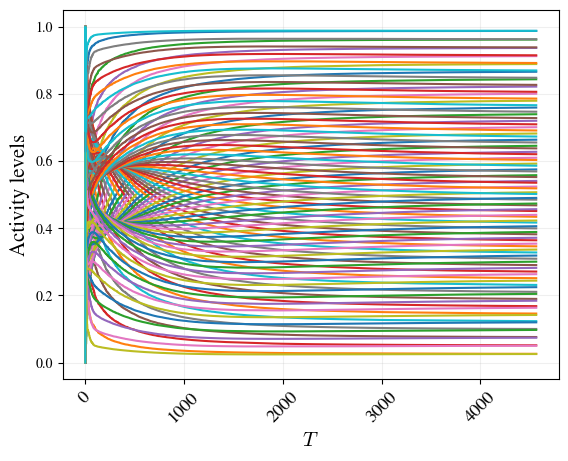} & \includegraphics[width=0.23\textwidth]{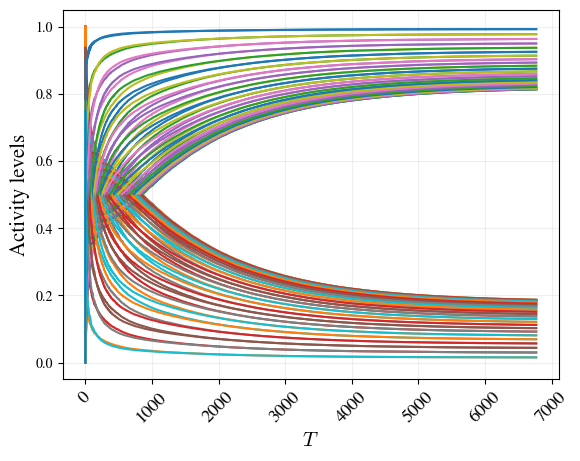} & \includegraphics[width=0.23\textwidth]{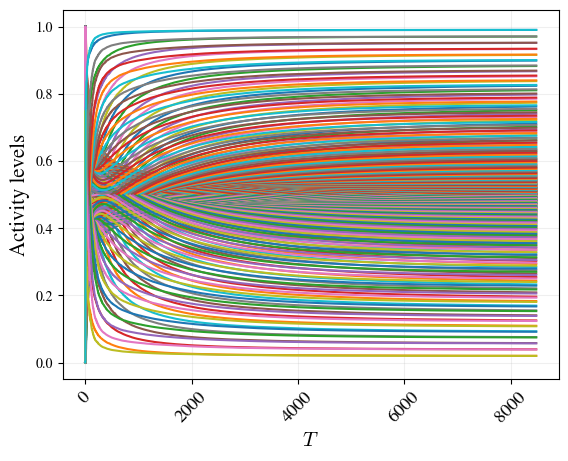} & \includegraphics[width=0.23\textwidth]{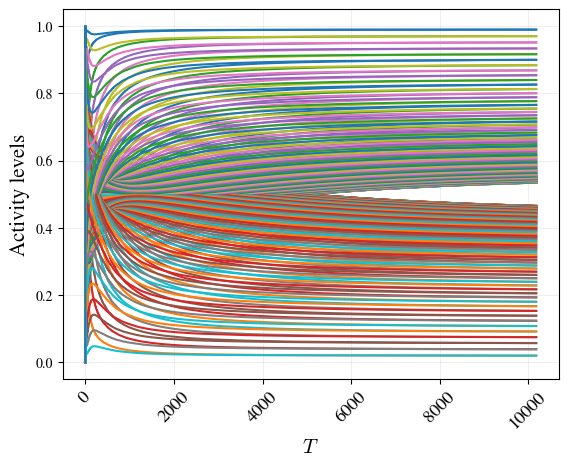} \\
            $0.5$ & \includegraphics[width=0.23\textwidth]{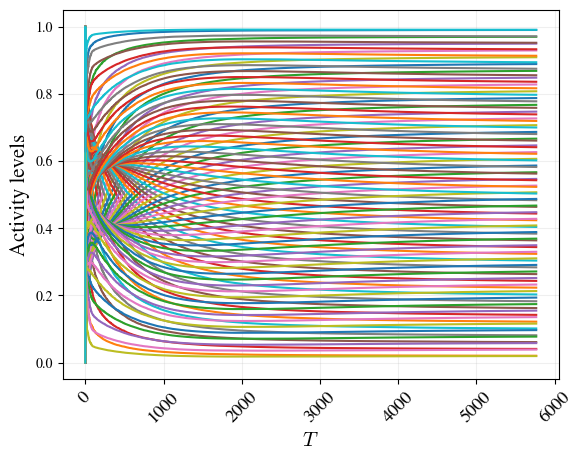} & \includegraphics[width=0.23\textwidth]{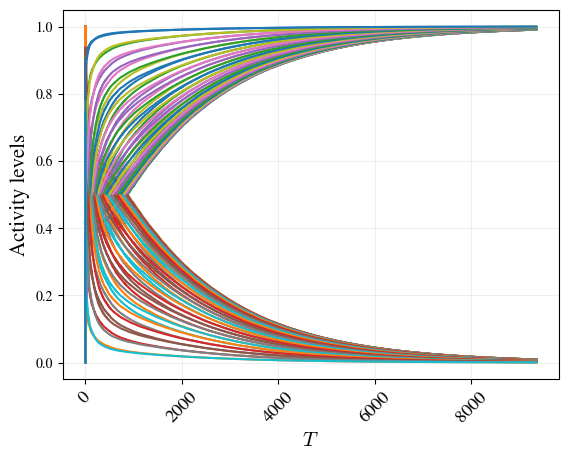} & \includegraphics[width=0.23\textwidth]{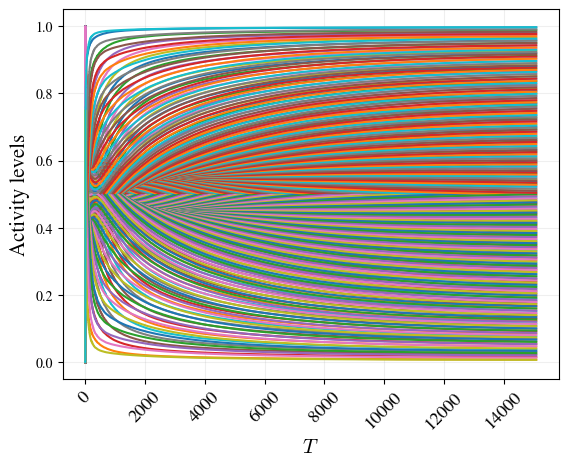} & \includegraphics[width=0.23\textwidth]{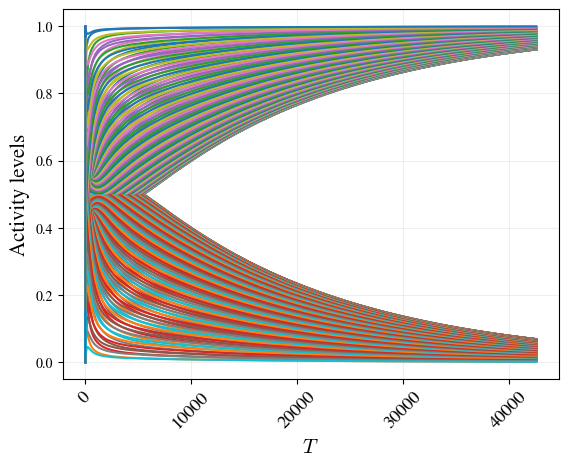} \\
            $0.502$ & \includegraphics[width=0.23\textwidth]{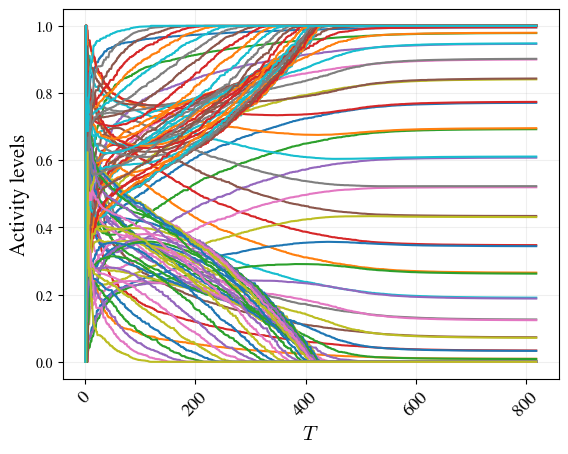} & \includegraphics[width=0.23\textwidth]{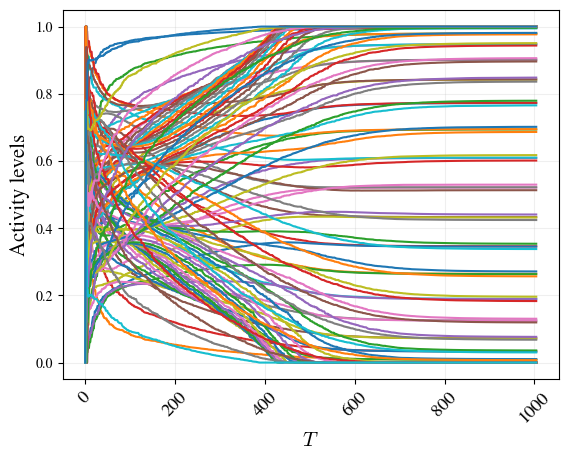} & \includegraphics[width=0.23\textwidth]{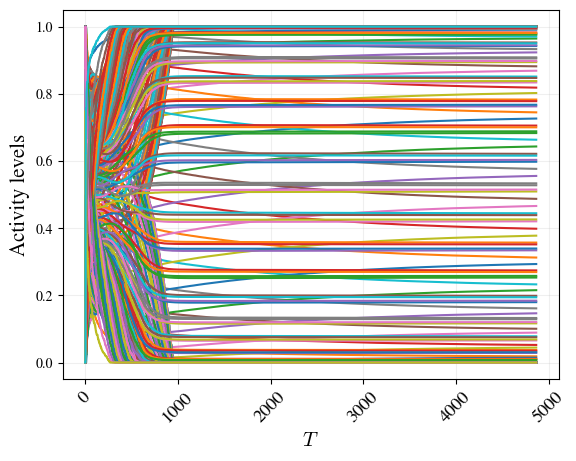} & \includegraphics[width=0.23\textwidth]{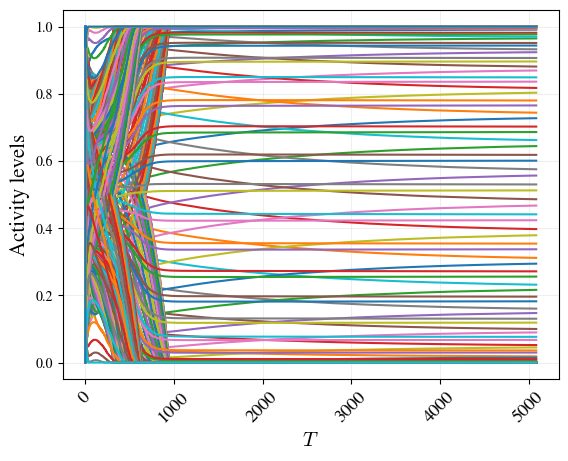} \\
            $0.51$ & \includegraphics[width=0.23\textwidth]{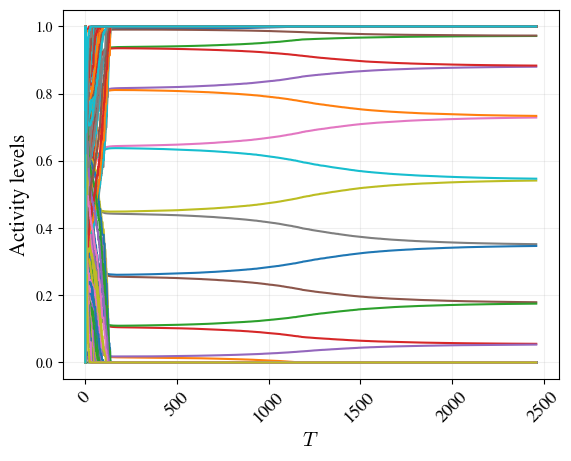} & \includegraphics[width=0.23\textwidth]{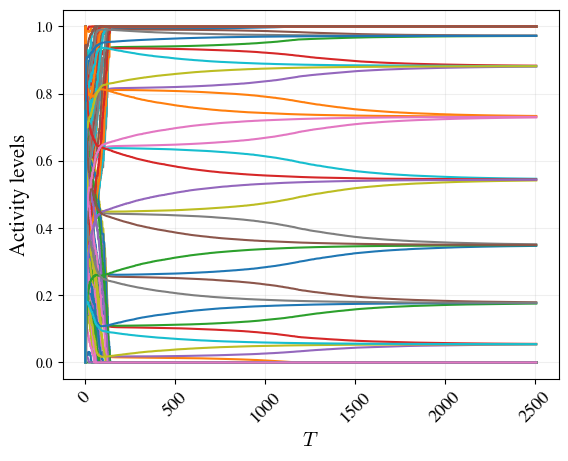} & \includegraphics[width=0.23\textwidth]{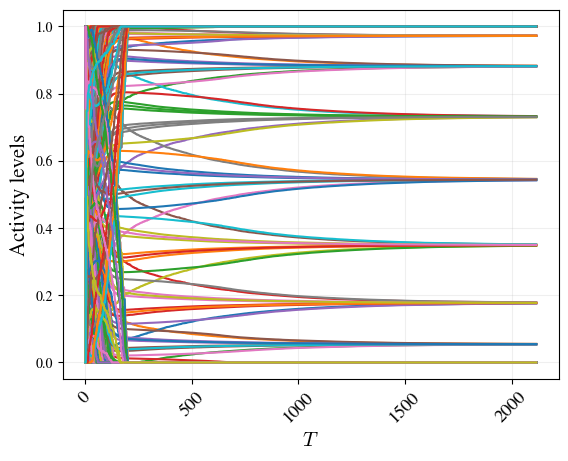} & \includegraphics[width=0.23\textwidth]{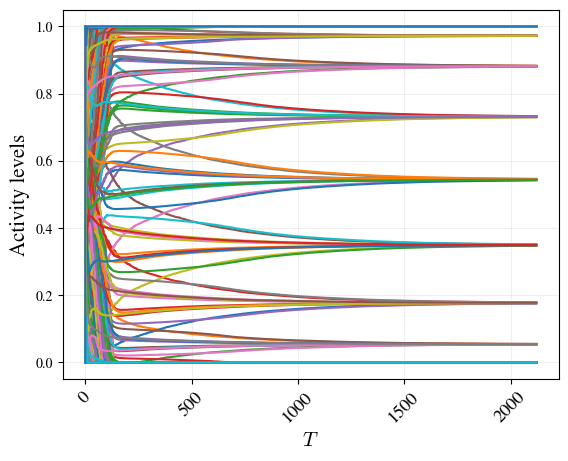} \\
            $$0.52$$ & \includegraphics[width=0.23\textwidth]{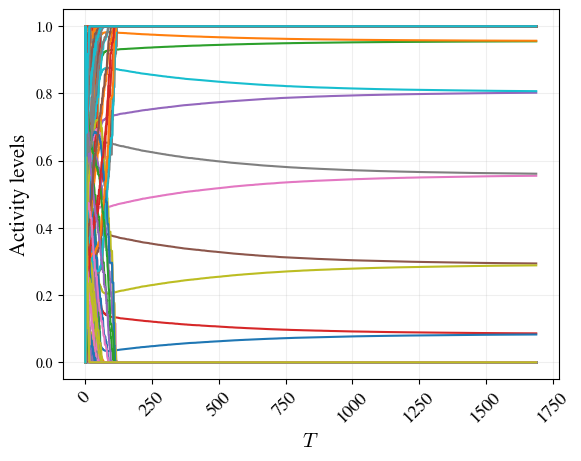} & \includegraphics[width=0.23\textwidth]{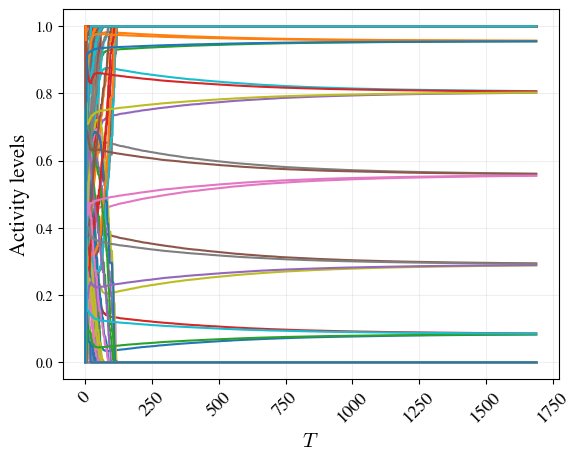} & \includegraphics[width=0.23\textwidth]{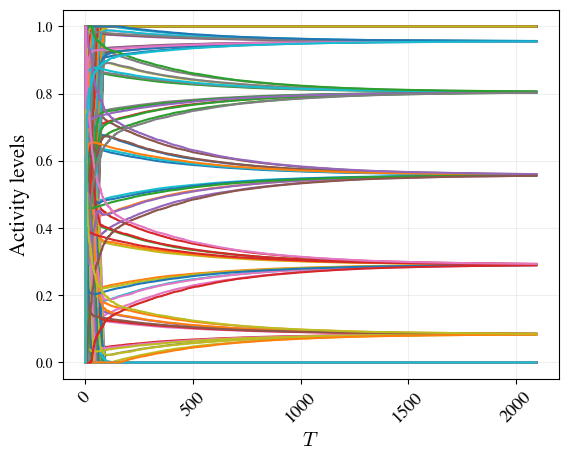} & \includegraphics[width=0.23\textwidth]{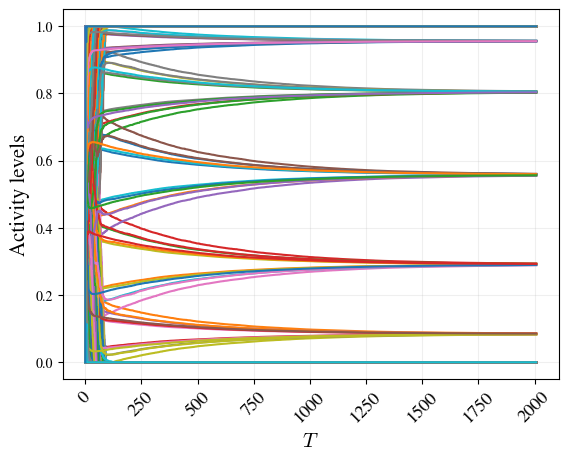} \\
            \bottomrule
        \end{tabular}
        \caption{Sample trajectories for different values of $n$ and $\delta$. All trajectories use the same value of $\epsilon=10^{-5}$}
        \label{tab:parity_figures}
    \end{table}
Furthermore, we include the full version of the convergence plot for the path on $n=101$ vertices in Figure~\ref{fig:101tall}. Note the significant increase of convergence time at $\delta=1/2$. This may be also observed in the trajectories in Table~\ref{tab:parity_figures}.
\begin{figure}[htb]
    \centering
    \includegraphics[width=0.5\linewidth]{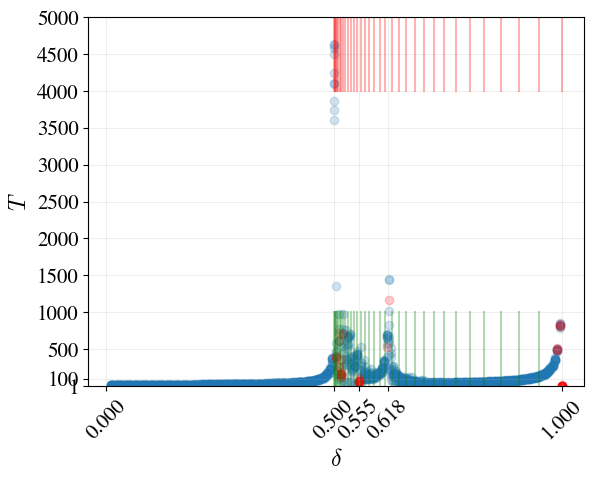}
    \caption{The convergence times $T$ for the path on $n=101$ vertices.}
    \label{fig:101tall}
\end{figure}

\newpage
\section{Additional considerations about random graphs.}
\label{appendix:random_trajs}

\begin{figure}[h]
\centering
\begin{minipage}{.3\textwidth}
\centering
\begin{subfigure}{\textwidth}
    \includegraphics[width=\linewidth]{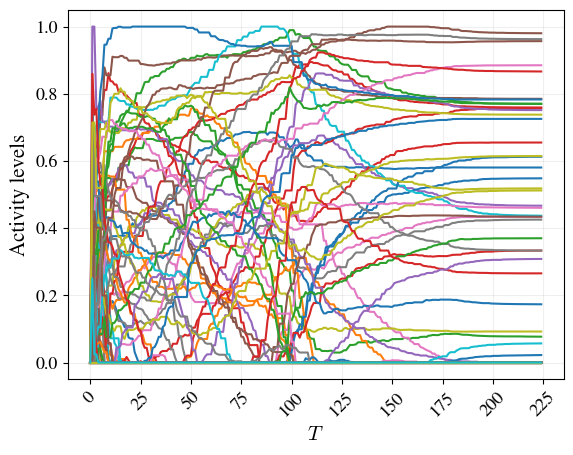}
    \caption{$\delta=0.25$}
\end{subfigure}\\
\begin{subfigure}{\textwidth}
    \includegraphics[width=\linewidth]{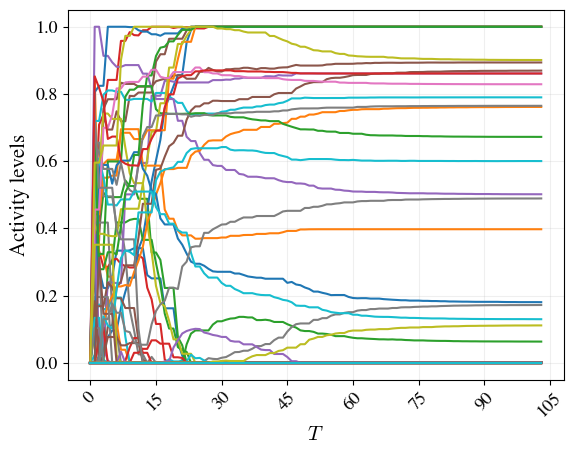}
    \caption{$\delta=0.35$}
\end{subfigure}\\
\begin{subfigure}{\textwidth}
    \includegraphics[width=\linewidth]{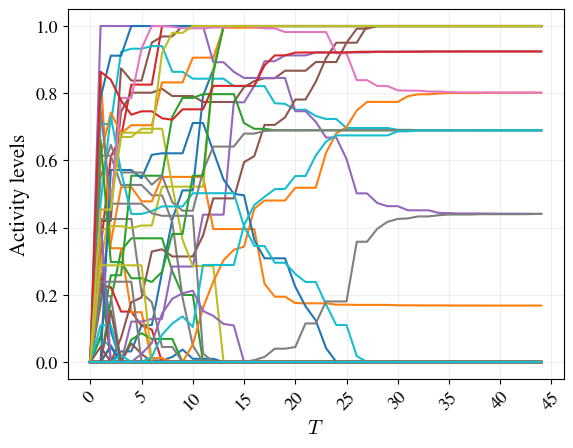}
    \caption{$\delta=0.45$}
\end{subfigure}\\
\begin{subfigure}{\textwidth}
    \includegraphics[width=\linewidth]{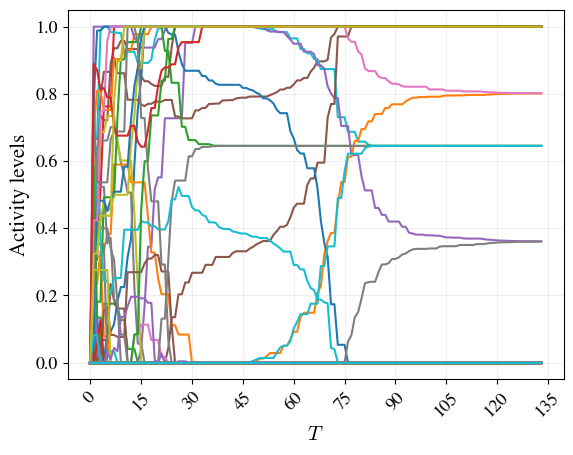}
    \caption{$\delta=0.55$}
\end{subfigure}
    \caption{Sample trajectories for a RR graph sample with $d=20$ and $n=100$ for different values of $\delta.$}
    \label{fig:rrtrajs}
\end{minipage}%
\hfill
\begin{minipage}{.3\textwidth}
\centering
\begin{subfigure}{\textwidth}
    \includegraphics[width=\linewidth]{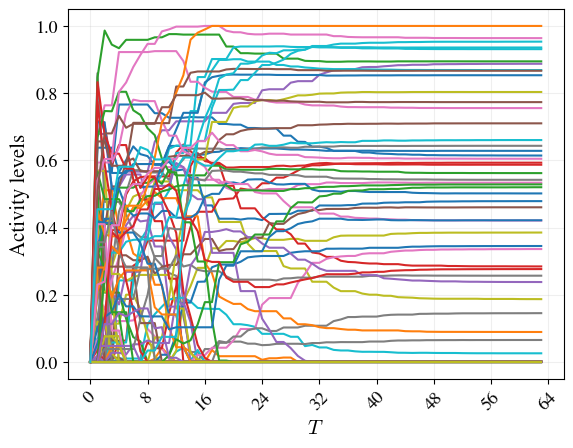}
    \caption{$\delta=0.25$}
\end{subfigure}\\
\begin{subfigure}{\textwidth}
    \includegraphics[width=\linewidth]{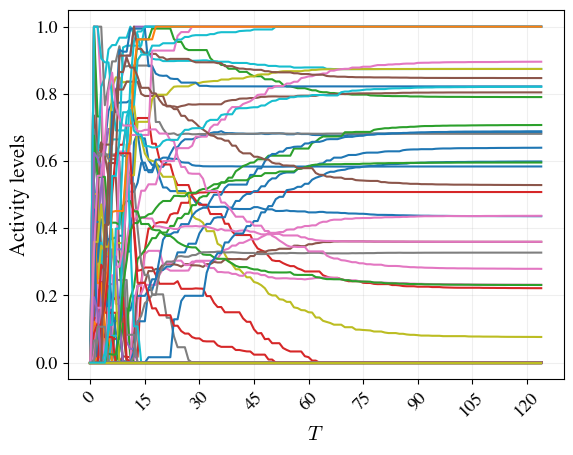}
    \caption{$\delta=0.35$}
\end{subfigure}\\
\begin{subfigure}{\textwidth}
    \includegraphics[width=\linewidth]{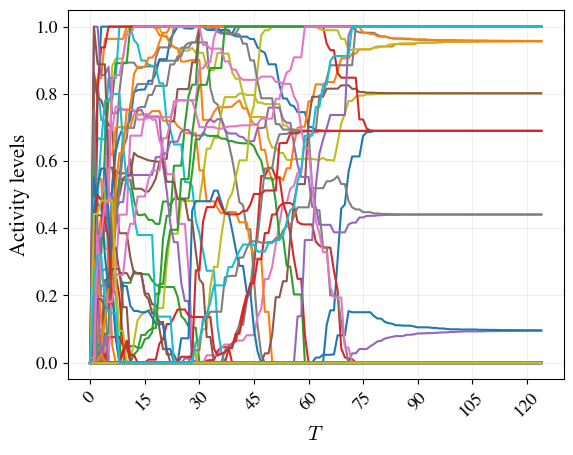}
    \caption{$\delta=0.45$}
\end{subfigure}\\
\begin{subfigure}{\textwidth}
    \includegraphics[width=\linewidth]{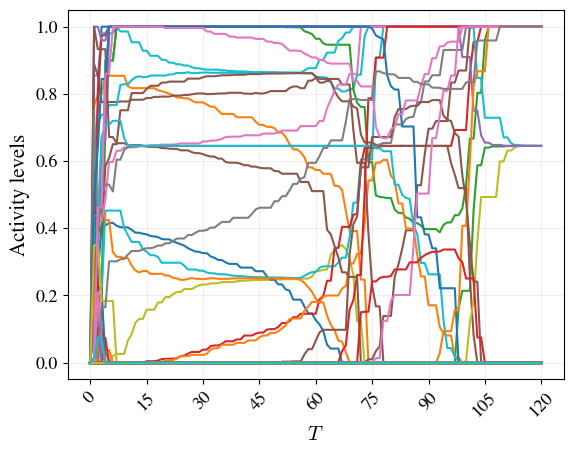}
    \caption{$\delta=0.55$}
\end{subfigure}
    \caption{Sample trajectories for an ER graph sample with $p=0.2$ and $n=100$ for different values of $\delta.$}
    \label{fig:ertrajs}
\end{minipage}%
\hfill
\begin{minipage}{.3\textwidth}
\centering
\begin{subfigure}{\textwidth}
    \includegraphics[width=\linewidth]{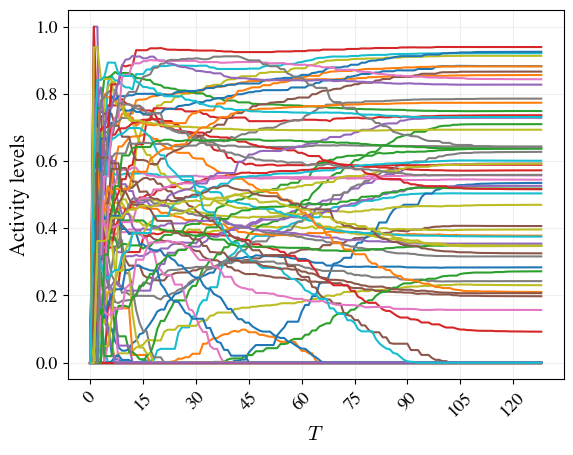}
    \caption{$\delta=0.25$}
\end{subfigure}\\
\begin{subfigure}{\textwidth}
    \includegraphics[width=\linewidth]{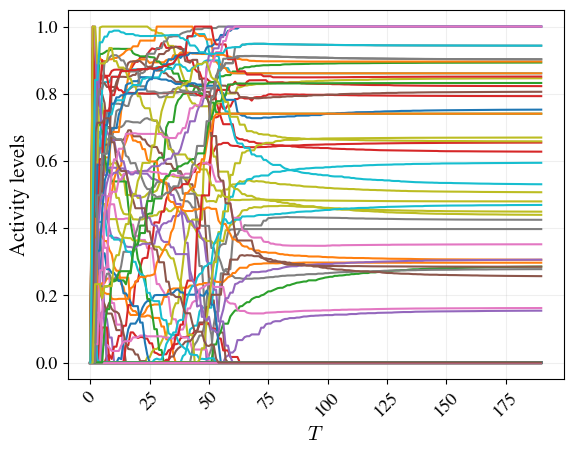}
    \caption{$\delta=0.35$}
\end{subfigure}\\
\begin{subfigure}{\textwidth}
    \includegraphics[width=\linewidth]{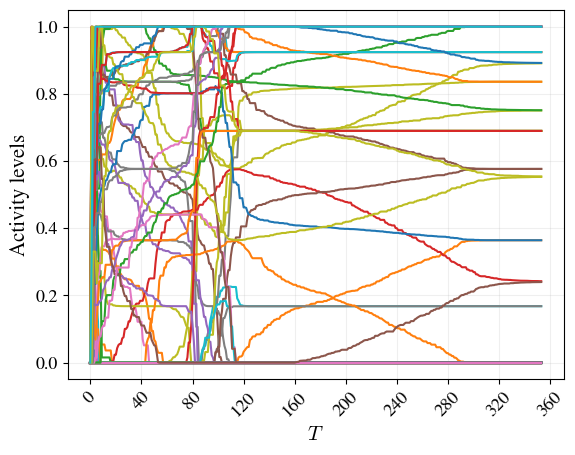}
    \caption{$\delta=0.45$}
\end{subfigure}\\
\begin{subfigure}{\textwidth}
    \includegraphics[width=\linewidth]{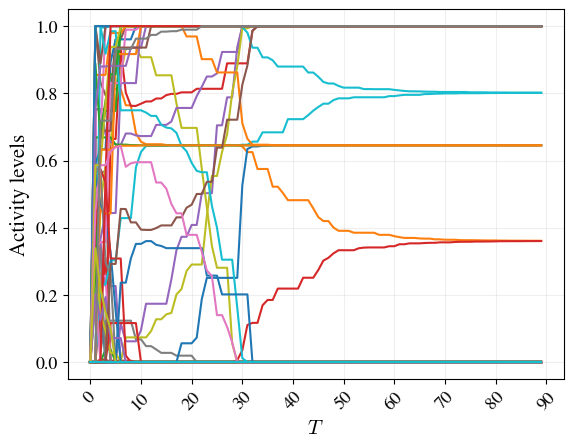}
    \caption{$\delta=0.55$}
\end{subfigure}
    \caption{Sample trajectories for a BA graph sample with $m=10$ and $n=100$ for different values of $\delta.$}
    \label{fig:batrajs}
\end{minipage}
\end{figure}

All of the considered random graphs show varied behavior within their trajectories (c.f. Figures~\ref{fig:rrtrajs},~\ref{fig:ertrajs} and~\ref{fig:batrajs}). While the individual trajectories are different and often difficult to interpret, we can focus on the similarities between the graphs.
As the value of $\delta$ increases, the individual trajectories of the dynamics become more structured, with fewer players changing strategies and less varied activity levels within the stable equilibria reached by the process.
Furthermore, the active sets supporting these stable equilibria decrease in size, the number of edges and contain an increasing number of agents who have no active neighbors (c.f. Figure~\ref{fig:rractives}).

\begin{figure}[h]
\centering
\begin{subfigure}{0.24\textwidth}
    \includegraphics[width=\linewidth]{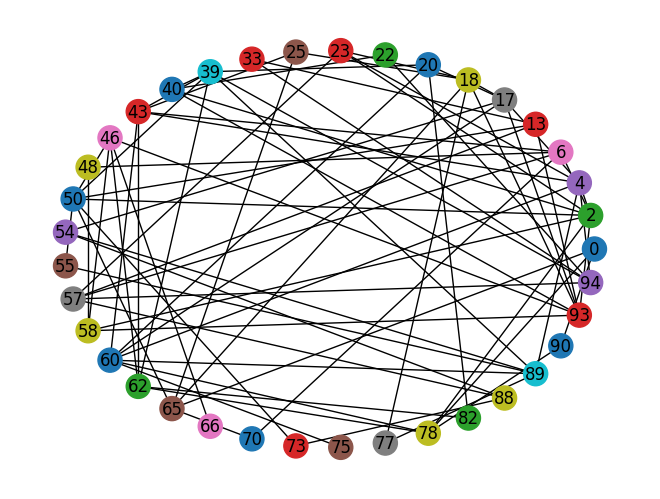}
    \caption{$\delta=0.25$}
\end{subfigure}%
\hfill
\begin{subfigure}{0.24\textwidth}
    \includegraphics[width=\linewidth]{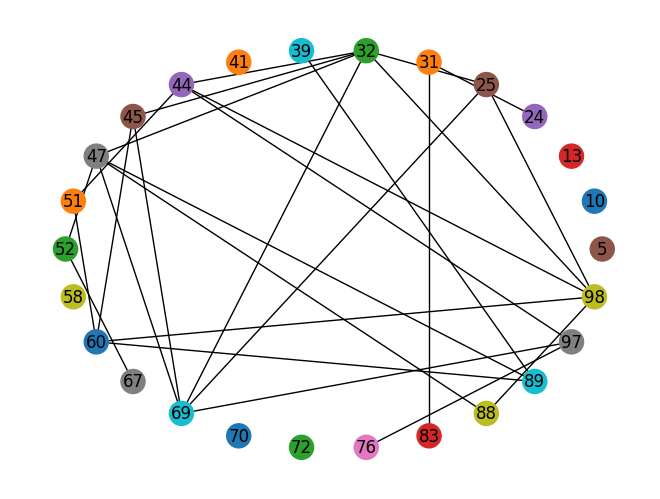}
    \caption{$\delta=0.35$}
\end{subfigure}%
\hfill
\begin{subfigure}{0.24\textwidth}
    \includegraphics[width=\linewidth]{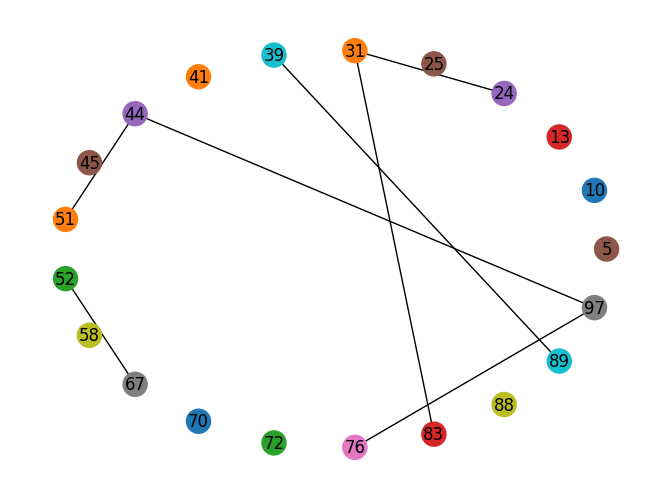}
    \caption{$\delta=0.45$}
\end{subfigure}%
\hfill
\begin{subfigure}{0.24\textwidth}
    \includegraphics[width=\linewidth]{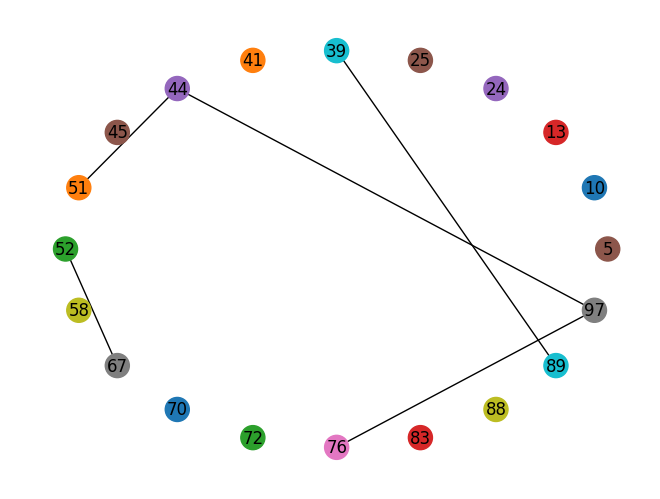}
    \caption{$\delta=0.55$}
\end{subfigure}
    \caption{The graph induced by the active set of the equilibrium the dynamics converged to for RR graphs with $d=20$ and $n=100$ for different values of $\delta.$ The colors correspond to the trajectories in Figure~\ref{fig:rrtrajs}.}
    \label{fig:rractives}
\end{figure}

These changes can also be captured quantitatively.
Let $S_{\max}$ be the largest connected component of the active set $S$ of the equilibrium the system converges to and by $I$ the set of \emph{isolated agents}, ie., the agents who are active but have no active neighbors.
The size of the largest component $|S_{\max}|$ decreases as the value of $\delta$ increases (c.f. Figure~\ref{fig:rractives}).
For small values of $\delta$, the active set is connected, implying $|S|=|S_{\max}|$.
For intermediate values of $\delta$, the sets no longer overlap and the number of isolated agents increases.
This is consistent with the basic intuition about the system --- when the externality factor is larger, less active agents suffice to exhaust the activity requirements of the agents.

Comparing this with the convergence time plots of the corresponding graphs (c.f. Figures~\ref{fig:rr},~\ref{fig:er},~\ref{fig:ba}), we can observe that, despite $S_{\max}$ being smaller and therefore simpler, the convergence is slow in this region.
The long convergence times in this region were often connected to the shuffle phase being long.
The analysis conducted in this section suggests that the process of establishing the active set is more intricate when the equilibria consist of smaller connected components.
This conclusion is consistent with the remainder of the analysis, as when the connected components of the active sets are smaller, more distinct stability thresholds can affect the convergence properties throughout the dynamics.

\begin{figure}[h]
\centering
\begin{subfigure}{0.33\textwidth}
    \includegraphics[width=\linewidth]{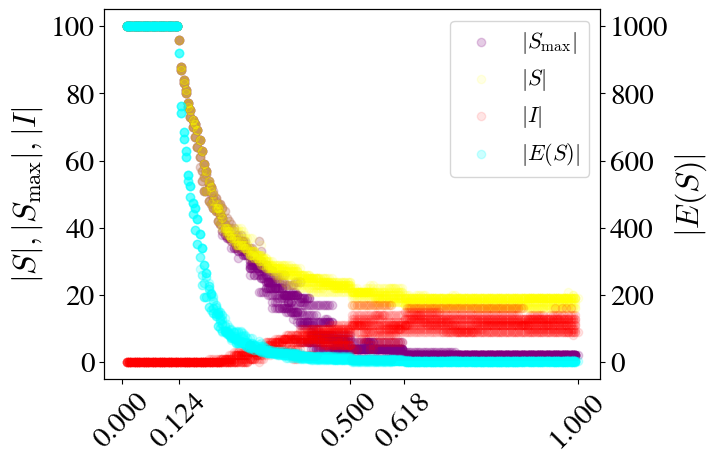}
    \caption{A sample RR graph with $d=20$.}
\end{subfigure}%
\hfill
\begin{subfigure}{0.33\textwidth}
    \includegraphics[width=\linewidth]{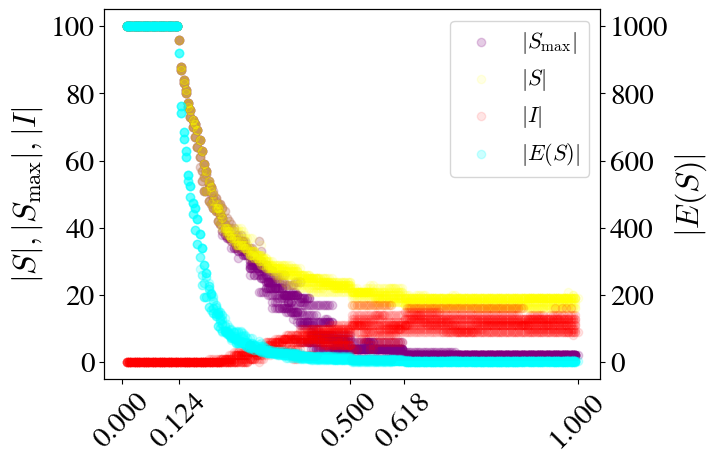}
    \caption{A sample ER sample graph with $p=0.2$.}
\end{subfigure}%
\hfill
\begin{subfigure}{0.33\textwidth}
    \includegraphics[width=\linewidth]{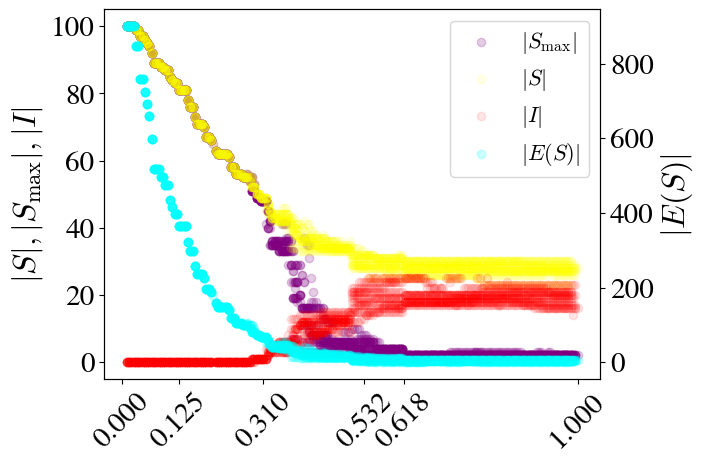}
    \caption{A sample BA sample graph with $m=10$.}
\end{subfigure}
    \caption{The quantities $|S|, |S_{\max}|, |I|$ and the number of edges $|E(S)|$ in the subgraphs induced by $S$ at the equilibrium for different values of $\delta$ for different random graph models. Each point corresponds to a single simulation. The first non-zero value on the $x-$axis denotes the loss of uniqueness threshold $\delta=1/|\lambda_{\min}(\bm G)|$.}
    \label{fig:rractive_counts}
\end{figure}

\section{Example of a graph with arbitrarily many active set changes}
\label{appendix:many_switchups}
This example showcases the possibility of slow communication.
\begin{lemma}
    For $\delta\in(1/\phi,1)$ there exists a graph $G_n$ on $5n$ vertices and a starting configuration $x$ for which the expected time taken to establish the final active set is at least $c(\delta)n$, with $c\geq1$.
\end{lemma}
\begin{proof}
    The proof is constructive.
    First, we take $n$ copies of the path graph with $n=5$, with $i$th vertex of $j$th path labeled $v_i^j$ for $1\leq i\leq 5, 1\leq j\leq n$.
    Now, we add the edges $(v_4^j,v_1^{j+1})$ and obtain $G_n$.

    Consider now the configuration given by
    \begin{align*}
        x^j=\begin{cases}
            (0,1,0,1,0) \text{ if } j=1\\
            (0,1,0,1/(1+\delta),1/(1+\delta)) \text{ if } j>1.
        \end{cases}
    \end{align*}
    All agents outside of the first path are playing their best response.
    In particular, the first agent in each path has their threshold met by the $4$th agent from the previous path. 

    Now, in the first path, we first have the dynamics presented in the proof of Fact~\ref{thm:p5}.
    When they converge close enough to the equilibrium on $2$ vertices, agent $3$ becomes active and we begin transitioning to the strategy profile in which agent $4$ is inactive.
    This forces the first agent from the second path to become active and subsequently the second path initiates convergence to the configuration $2-2$.

    This behavior repeats $n$ times, giving the total time until the last activity change $cn$, where $c$ is expected time of each step and is bounded from below in the proof of Fact~\ref{thm:p5}.
\end{proof}
The probability of the described event is low.
However, if we consider a disjoint union of enough copies of the graph, we obtain an arbitrarily long expected total convergence time.
This leads us to prove Proposition~\ref{pr:many_reshuffles}

\begin{autoproof}{pr:many_reshuffles}
    For a single copy of $G_n$, the probability of landing in the state described in the previous proof is at least $p=(5n)^{-2n}$.
    This occurs by randomly choosing a starting update order in which the target players are chosen first.
    This estimate is very crude but proves sufficient for the construction.
    Now, consider a disjoint union of $1/p$ disjoint copies of $G_n$. At least one of the copies reaches the strategy profile leading to many reshuffles with probability
    \begin{align*}
        q\geq 1-(1-p)^{1/p}\geq1-1.1e^{-1}>1/2.
    \end{align*}
    This means that the expected time to establish the active set is at least $c(\delta)n/2$.
    As $n$ is arbitrary, this suffices.
\end{autoproof}

\end{document}